\documentclass[11pt,letterpaper]{article}
\usepackage[latin1]{inputenc}
\usepackage{amsmath}
\usepackage{mathtools}
\usepackage{amsfonts}
\usepackage{amssymb}
\usepackage{amsthm}
\usepackage{ulem}
\usepackage{graphicx}
\usepackage{epstopdf}
\usepackage[margin=1in]{geometry}
\usepackage[hidelinks]{hyperref}
\usepackage{color}
\usepackage{todonotes}

\usepackage{indentfirst}
\usepackage{multirow}
\usepackage{bm}
\usepackage{subfigure}
\usepackage{multirow}
\usepackage{tabularx}
\usepackage{lineno}
\usepackage{xr}

\newtheorem*{remark}{Remark}

\usepackage{tikz}
\usetikzlibrary{shapes.geometric, arrows}
\tikzstyle{startstop} = [rectangle, rounded corners, minimum width=3cm, minimum height=1cm, text centered, draw=black, fill=red!30]
\tikzstyle{io} = [trapezium, trapezium left angle=70, trapezium right angle=110, minimum width=3cm, minimum height=1cm, text centered,  text width=3cm, draw=black, fill=blue!30]
\tikzstyle{process} = [rectangle, minimum width=3cm, minimum height=1cm, text centered, draw=black, fill=orange!30]
\tikzstyle{decision} = [diamond, minimum width=3cm, minimum height=3cm, text centered, text width=6.5em, draw=black, fill=green!30]
\tikzstyle{arrow} = [thick,->,>=stealth]

\makeatletter
\newcommand*{\addFileDependency}[1]{% argument=file name and extension
  \typeout{(#1)}% latexmk will find this if $recorder=0 (however, in that case, it will ignore #1 if it is a .aux or .pdf file etc and it exists! if it doesn't exist, it will appear in the list of dependents regardless)
  \@addtofilelist{#1}% if you want it to appear in \listfiles, not really necessary and latexmk doesn't use this
  \IfFileExists{#1}{}{\typeout{No file #1.}}% latexmk will find this message if #1 doesn't exist (yet)
}
\makeatother

\hypersetup{colorlinks,citecolor=blue}
\def\Dd{\mathcal{D}}

\def\Ff{\mathcal{F}}

\def\Hh{\mathcal{H}}

\def\Rr{{\mathcal{R}}}

\def\E{\mathbb{E}}

\def\R{\mathbb{R}}

\def\PD#1#2{\frac{\partial{#1}}{\partial{#2}}}

\def\EXPECT{{\mathbb{E}}}

\def\RELENT#1#2{\Rr\left(#1|#2\right)}

\def\COMMA{\,,}             % use for commas at the ends of formulae
\def\PERIOD{\,.}            % use for periods at the ends of formulae
\def\SEP{{\,|\,}}           % use for seperator in set defs, like {x\in R \SEP

\def\VIZ#1{(\ref{#1})}      % use for references to formulae
\def\BARIT#1{{\bar {#1}}}

\def\CGmu{\BARMt}

\def\RELENT#1#2{\mathcal{R}\left({#1}\SEP{#2}\right)}
\def\ENTRATE#1#2{\mathcal{H}({#1}\SEP{#2})}

\def\EXPECT{{\mathbb{E}}}

\def\COPINV{\COP^\dagger}                  % Notation for generalized inverse
\def\PULLBCK#1{\COPINV_*{#1}}

\def\RELENT#1#2{\Rr\left(#1\|#2\right)}
\def\argmin{\textrm{ argmin}}
\def\argmax{\textrm{argmax}}

\def\bbx{\mathbf{\bar q}}

\def\bbX{{X}}

\def\bxs{\mathbf{x}}
\def\bbxs{\bar{\mathbf{x}}}

\def\bx{\bar q}

\def\bbq{\mathbf{q}}

\def\kelvin{\mathrm{\ K}}

\def\ns{\mathrm{\ ns}}
\def\fs{\mathrm{\ fs}}
\def\ps{\mathrm{\ ps}}
\def\nm{\mathrm{\ nm}}

\def\PXT{\{\bbX_t\}_{t\geq 0}}

\def\0o{\mathbb{O}}

\def\BARIT#1{{\bar {#1}}}

\def\BARM{\BARIT{\mu}}
\def\BARMt{\BARIT{\mu}^{\theta}}

\def\BARV{\BARIT{U}}

\def\IP{U}
\def\PMF{{U}^{\mathrm{pmf}}}

\def\COP{\mathbf{\Pi}}

\def\ENTRATE#1#2{\mathcal{H}({#1}\SEP{#2})}

\title{{\bf Data-driven Uncertainty Quantification for Systematic Coarse-grained Models  
}}
\author{Tangxin Jin$^{1}$, Anthony Chazirakis$^{2,3}$, Evangelia Kalligiannaki$^{2,*}$, \\ Vagelis Harmandaris$^{2,3}$, and  Markos A. Katsoulakis$^{2}$
}
\date{\begin{flushleft} $^{1}$  Department of Mathematics and Statistics, University of Massachusetts, Amherst, USA; \\
 $^{2}$  Institute of Applied and Computational Mathematics,  Foundation for Research and Technology - Hellas,  Heraklion Crete,  Greece;\\
$^{3}$ Department of Mathematics and Applied Mathematics, University of Crete, Heraklion Crete,  Greece; \\
*E-mail: evangelia.kalligiannaki@iacm.forth.gr
 \end{flushleft}
}

\begin{document}
\maketitle

\begin{abstract}
In this work, we present methodologies for the quantification of confidence in bottom-up coarse-grained models for molecular and macromolecular systems. Coarse-graining methods   have been  extensively used in the past decades in order to extend the length and time scales accessible by simulation methodologies. The quantification, though, of induced errors due to the limited availability of fine-grained data  is not yet established. 
Here, we employ rigorous statistical methods to deduce guarantees for the optimal coarse models obtained via approximations of the multi-body potential of mean force,  with the relative entropy, the relative entropy rate minimization, and the force matching methods.
Specifically, we present and apply statistical approaches, such as bootstrap and jackknife, to infer confidence sets for a limited number of samples, i.e., molecular configurations. 
Moreover, we estimate asymptotic confidence intervals assuming adequate sampling of the phase space.
We demonstrate the need for non-asymptotic methods and quantify confidence sets through two applications. The first is a two-scale fast/slow diffusion process projected on the slow process. With this benchmark example, we establish the methodology for both independent and time-series data. Second, we apply these uncertainty quantification approaches on a polymeric bulk system. We consider an atomistic  polyethylene melt as the prototype system for developing coarse-graining tools for macromolecular systems.  For this system, we  estimate the coarse-grained force field and present confidence levels with respect to the number of available microscopic data.

 \end{abstract}


\newpage
 
 \section{Introduction}
The research in systematic bottom-up coarse-graining methods for molecular systems has significantly advanced in the past decades. When adequate information is provided through the fine-grained data, the resulting coarse force fields are describing well   structural  properties, \cite{MulPlat2002,tsop1,Reith2003,Noid2013a}. 
Moreover, there is active research and  considerable progress   on  the dynamics of coarse models, \cite{Harmandaris2009a, Harmandaris2009b,HKKP2016,Rudzinski2019}.
However, there is a gap in the literature regarding the quantification of the induced errors due to the limited availability of fine-grained data. In the current work, we aim to incorporate rigorous statistical methods   with coarse-graining methods   to provide data-driven confidence sets.

Coarse-graining (CG) is a model reduction methodology that is used in order to extend the spatio-temporal scales accessible by microscopic (atomistic) simulations and to study molecular systems properties at mesoscale regimes. 
Systematic (chemistry specific) CG models are obtained by lumping
groups of chemically connected atoms into CG particles (or CG beads) and deriving the effective coarse-grained interaction potentials from the microscopic details of the atomistic models.
Such models are capable of predicting  the properties of specific systems quantitatively and have been applied with great success to a vast range of molecular systems. 
To build CG models, one needs to derive (a) CG interaction potentials to describe equilibrium properties and (b) dynamical models to describe kinetic properties, directly from more detailed (microscopic) simulations.
The effective CG potentials approximate the many-body potential, describing the equilibrium distribution of CG particles. These CG potentials can be developed through different numerical parameterizing methods at equilibrium,  such as the iterative Boltzmann inversion (IBI) \cite{MulPlat2002,tsop1,Reith2003}, the inverse Newton (or inverse Monte Carlo)~\cite{Lyubartsev1995,Lyubartsev2010}, the force matching (FM) or Multiscale  Coarse-Graining (MSCG) ~\cite{IzVoth2005a,IzVoth2005}, \cite{Noid2007,Voth2008a,Voth2008b}, the Relative Entropy (RE)    methods~\cite{Shell2009,Shell2008}, and the cluster expansion based methods ~\cite{Tsourtis2017}. Also, during the last decade, bottom-up CG methods for treating molecular systems under non-equilibrium conditions have been developed.    Such are, the  recently  introduced,   path-space relative entropy (PSRE), relative entropy rate (RER), and  path-space force matching (PSFM) methods for providing effective CG models at equilibrium, non-equilibrium, transient, or stationary time regimes, \cite{KP2013, HKKP2016}.
The path-space methods have been further applied successfully to the     dimensionality reduction of stochastic reaction networks \cite{KatsVilanova2020}, and the sensitivity analysis of molecular models \cite{TPKH2015}. All these  methods fall under the umbrella of  statistical inference methods. Statistical inference is our point of view in the current study from which we draw the rigorous mathematical and statistical tools, \cite{EfronHastie2016, wasserman2010all, wasserman2006all}.  
 
Quantifying parametric uncertainties accounts for assessing the model accuracy, variability, and sensitivity. Thus, naturally, a primary challenge in all above CG approaches is to quantify uncertainties in  effective CG model due to the involved approximations. We are concidering  the (limited) size of the available microscopic data, and the numerical/algorithmic errors. 
Two general 'philosophies' in inferential statistics are   frequentist inference and  Bayesian inference.
The Bayesian approach has been studied recently for the coarse-graining   of  molecular systems. For example, Voth and co-workers \cite{VothBayes2008} have applied the empirical Bayes technique to estimate the force field  parameters for the  FM  method. Authors in \cite{Farrell1,Farell2}, in addition to parameter estimation, propose a methodology for model selection based on the Bayesian approach.
Furthermore, in refs. \cite{Koutsourelakis2017, Koutsourelakis2019}, authors focus on the derivation of credible intervals for CG models of water.
Bayesian uncertainty estimation has also  been applied to parametrize atomistic molecular models in refs.~\cite{Koumoutsakos2012, Koumoutsakos2013, Brotzakis2018, Frederiksen2004}.  
The Bayesian perspective can provide a range of probabilistic properties, but it relies on prior knowledge  often not available. 
Thus any credible interval estimation relies on uninformative priors. In contrast, estimating frequentist parametric and non-parametric confidence intervals requires no prior information, \cite{EfronHastie2016}.
 
Estimates of confidence intervals are given by asymptotic and non-asymptotic methods, chosen based on the available data. The asymptotic approach relies on the central limit theorem and the asymptotic Gaussian convergence theory. 
Additionally, concentration inequalities can provide reliable bound estimates for quantities of interest, \cite{DKPP}.
The non-asymptotic methods concern estimating parameter statistics for finite data; typical examples are the jackknife and the bootstrap ones \cite{efron1979, Jackknifereview1974}. 
Such methods have been employed in the past   to obtain estimates of the parameters in classical force fields.  For example, Reiher and collaborators \cite{Reiher1, Reiher2}, employed frequentist statistical tools. Specifically, they utilize non-parametric bootstrapping, to obtain reliable estimates of the fit parameters present in semi-classical  dispersion interactions based on the Density Functional Theory (DFT).
Recently, authors in ref. \cite{Longbottom_2019} introduce a probabilistic potential ensemble method to estimate uncertainties in classical potential fitting based on DFT calculations.
In addition, uncertainty quantification studies for the parameters of molecular models   appear in   \cite{TPKH2015}
using information theory tools,  and  in \cite{UQJacobson2014}  via a polynomial chaos approach.

Despite the above studies, according to our knowledge, asymptotic and non-asymptotic methods have not yet been explored in the context of CG modeling of high dimensional systems, and in particular for macromolecular systems.
Here we address the accuracy of CG models for molecular systems by employing frequentist statistical data analysis.
 Our goal is to present and apply rigorous statistical approaches, i.e., bootstrap and jackknife, to infer confidence sets for a \textit{limited number of samples}.
 
We apply these methodologies to: (a) a relatively simple benchmarking problem, of a two scale fast/slow diffusion process and (b) a realistic  bulk polymer model, as a prototype example of a high dimensional macromolecular system. 
The latter is essential if we consider that independent data are required to deduce the confidence sets with the non-asymptotic methods, though obtaining sufficiently uncorrelated data of high molecular weight model polymers is challenging, ~\cite{Harmandaris2003a,Harmandaris2009a}. 

The structure of this work is as follows. Firstly, we present a short review of the bottom-up coarse-graining methodologies of molecular systems  from the perspective of statistical inference. Next, we construct the asymptotic and non-asymptotic confidence intervals for the RE, RER, and FM methods.
We benchmark the methodology with a multi-scale diffusion system with known corresponding stochastic averaging limits. 
We derive the bootstrap and jackknife estimates for the fitted interaction potential for a  high dimensional polyethylene melt, based on data derived from detailed atomistic simulations. 
Finally, we conclude and discuss our findings.

\section{Physics-based data-driven coarse-graining} {\label{sec:Background}
Assume a prototypical problem of $n$ particles (atoms or molecules) in a box of volume $V$ at temperature $T$.  Let $\bbq=(q_1,\dots,q_n) \in \R^{3n}$ describe the position of the  particles in the atomistic (microscopic)  description  with potential energy $\IP(\bbq)$. The probability of a state $\bbq$ at the temperature $T$ is given by the Gibbs canonical  probability density
\begin{equation}\label{eq:Gibbs}
\mu(\bbq) =Z^{-1}\exp\{ -\beta \IP(\bbq)\}\COMMA
\end{equation}
where $Z= \int_{\R^{3n}}e^{-\beta \IP(\bbq)} d\bbq$ is the configurational the partition function, $\beta=\frac{1}{k_B T}$, and  $k_B$ is the Boltzmann constant. We should note that the studies and analyses presented in this work are  performed on the configuration space.
Moreover, we assume that the configurational time evolution of the particles is described by a continuous time process $\PXT=\{\bbq_t\}_{t\ge 0}$ in $\R^{3n}$, with path space distribution $P_{[0,t]}\COMMA$ and  Gibbs probability density \eqref{eq:Gibbs}. If we assume Markovianity, then a temporal discretization of  the process leads to a Markov chain with the transition probability kernel $p(\bxs,\bxs')$. Thus, the path space probability {\ density} of   $ \{\bbX_1,\dots,\bbX_N\} $,  observed at   $ t_1\dots, t_N  $ respectively,  is
\begin{equation}\label{eq:psdensity}
    P(\bbX_1,\dots,\bbX_N) = \nu(\bbX_1)\prod_{i=1}^{N-1} p(\bbX_i, \bbX_{i+1})\COMMA
\end{equation}
where $\nu$ is the initial state probability density. 
We define coarse-graining through the configurational CG mapping    $\COP:~\R^{3n} \to \R^{3m} $,
determining the $m(<n)$ CG particles as a function  of the microscopic configuration $\bbq$. The  mappings most commonly considered in coarse-graining of molecular systems are linearly  represented by a set of non-negative real constants $\{\zeta_{ij}, i=1,\dots, m,\ j=1,\dots, n \}$, for which
  $\bbx_i :=\COP_i(\bbq) = \sum_{j} \zeta_{ij} q_j \in \R^{3}, \ i=1,\dots,m \PERIOD$
The probability that the CG system has configuration $\bbx= (\bx_1,\dots,\bx_m)\in \R^{3m}$  is 
 $ \BARM(\bbx) =  \int_{\Omega(\bbx)}\mu(\bbq)d\bbq,  \ \ \Omega(\bbx) =\{\bbq\in \R^{3n}: \ \COP(\bbq) = \bbx\}  
 \PERIOD$
  The corresponding free energy at the CG level, described by  the $m-$body potential of the mean force (PMF), is  
\begin{equation*}\label{PMF}
  \PMF(\bbx)=   -\frac{1}{\beta} \log \int_{\Omega(\bbx)} e^{-\beta \IP(\bbq)} d\bbq\PERIOD
\end{equation*}
Bottom-up structural-based CG methods look for approximations of the m-body PMF $\PMF(\bbx)$  
\begin{equation}\label{CGpotential}
\BARV(\bbx;\theta)  \,, \quad \theta \in \Theta\subseteq \R^K \COMMA
\end{equation}
which defines the corresponding approximating probability density 
\begin{equation}\label{eq:CGdensity}
\CGmu(\bbx) = (Z^{\theta})^{-1}\exp\{-\beta \BARV(\bbx;\theta)\} \,, \quad \theta \in \Theta  \COMMA
\end{equation}
where $Z^{\theta} = \int_{\R^{3m}} e^{ -\beta \BARV(\bbx;\theta) }d\bbx $ is the normalization constant.

 We introduce  a Markov process $\{\bar X_t\}_{t\ge 0} $ in $\R^{3m}$  to approximate the time evolution of the coarse  variables $ \{\COP \bbX_t\}_{t\ge 0}$.  The CG process $\{\bar X_t\}_{t\ge 0} $ is defined through its parametric path space distribution 
  \begin{equation}\label{eq:CGpsdistribution}
      \bar{Q}^{\theta}_{[0,t]} \,, \quad \theta \in  \Theta \subseteq  \R^{K}\PERIOD
  \end{equation}
The goal is to find the most effective CG model given  a set of  either  {\  independent and identically distributed (i.i.d.)}\  or time-series data. In this work, we elaborate with the relative entropy minimization, relative entropy rate minimization, and the force matching methods to find the effective CG model.

\medskip
\noindent
{\bf I. Independent, identically distributed data.} Given $ N $  i.i.d.  configurational observations from the microscopic Gibbs density \eqref{eq:Gibbs}, 
 \begin{equation}\label{eq:DataIID}
     \Dd^{iid}_{N}= \{\bbX_1,\dots,\bbX_{N}\}\COMMA 
 \end{equation} 
 we aim to infer the CG probability density \eqref{eq:CGdensity}.
 
 The Force Matching  method determines a CG approximating force $\BARIT F(\bbx;\theta) = -\nabla \BARV(\bbx;\theta)  $,  and thus an effective potential from atomistic force information, as the solution of  the mean least-square minimization  problem
\begin{equation}\label{eq:FM}
\min_{  \theta \in \Theta}\E_\mu\left[ \| F(\bbq) - \BARIT F(\COP(\bbq);\theta)\|^2  \right] \COMMA
\end{equation}
where $\E_\mu [\cdot] $ denote the average  with respect to the probability density $\mu(\bbq)$, and $\| \cdot \|$  the Euclidean norm in $\R^{3m}$.
The reference field $F(\bbq)\in \R^{3m}$  is the  local mean force whose component $F_I(\bbq), \ I=1,\dots, m$  is the force exerted at the $I$-th CG particle and is a function of the microscopic forces.  
For example, if the CG particle corresponds to the center of mass of a group of atoms then $F_I(\bbq) = \sum_{j\in \{\textrm{group } I \}}  f_j(\bbq)$, $ I=1,\dots m$, where $ f_j(\bbq)$ is the force exerted at the $j$-th microscopic particle. 
Thus, given the set of i.i.d.\ data  $\Dd^{iid}_{N} $ described in \eqref{eq:DataIID}, the discrete  optimization problem  corresponding to \eqref{eq:FM} is
\begin{equation}\label{eq:FMestimator}
\hat\theta^{iid,fm}_{N}(X_1,\dots,X_N) = \underset{\theta  \in \Theta}\argmin \frac{1}{N}\sum_{i=1}^{N} \| F(X_i) - \BARIT F(\COP(X_i);\theta)\|^2 \PERIOD
\end{equation}
 
The  Relative Entropy minimization  method determines a CG effective potential $\BARV(\bbx;\theta)$ by minimizing the relative entropy $\RELENT{\mu}{ \mu^{\theta}}$ between the microscopic Gibbs measure $\mu(\bbq)$ and a back-mapping $\mu^\theta(\bbq ) = \CGmu(\bbx)\nu(\bbq|\bbx) $ of the approximate CG measure  $ \CGmu(\bbx)$. That is 
\begin{equation}\label{eq:REmin}
 \underset{\theta  \in \Theta}\min  \RELENT{\mu}{\mu^{\theta}  } \COMMA  
\end{equation}
 where 
\begin{equation*}
\RELENT{\mu}{ \mu^{\theta}} =
\E_\mu\left[ \log\frac{\mu(\bbq)}{\mu^{\theta}(\bbq)} \right]\PERIOD
\end{equation*}
Thus, the RE estimator for the CG model is
\begin{equation}\label{eq:RE_equi_N}
\hat\theta^{iid, re}_{N} (X_1,\dots,X_N) =  \underset{\theta\in\Theta}\argmin   \frac{1}{N  }\sum_{i=1}^{N }\log \frac{ {\mu}(  X_{i})}{{\CGmu(\COP X_i)}}\COMMA
\end{equation}
assuming that the back-mapping distribution does not depend on $\theta$.

\medskip
\noindent
{\bf II. Time-series data.} 
In path space  we estimate the probability density \eqref{eq:CGpsdistribution} at dynamical regimes, given $N_p$ i.i.d. path observations 
 \begin{equation}\label{eq:DataPath}
\Dd^{ts}_{N_p,N_t}=\{\bbX^{k}_1,\dots,\bbX^k_{N_t}\}_{k=1}^{N_p}\COMMA
 \end{equation}
from the microscopic path space probability density \eqref{eq:psdensity}. Each path (or trajectory) observation consists of $N_t$ discrete time observations, which, for simplicity, we consider of uniform time step. Also, each path observation can have different size $N_k,\ k=1,\dots,N_p$. 
   
The best approximation is given by entropy based criteria to find the best Markovian approximation of the coarse-grained process. The optimization principle is defined in terms of the path-space relative entropy,
 \begin{equation}\label{VP-path1}
 \min_{\theta\in \Theta} \RELENT{P_{[0, t]}}{ Q^\theta_{[0, t]}}\, ,
 \end{equation}
where 
$ Q^\theta_{[0, t]} := \PULLBCK{\bar Q^\theta_{[0, t]}}$ 
is the 
back-mapping to the microscopic space of the parameterized path-space coarse-grained 
distribution.
The relative entropy rate (RER) is defined by
\begin{equation*}
    \ENTRATE{P}{Q^\theta } := \lim_{t\to \infty}\frac{1}{t} \RELENT{P_{[0, t]}}{ Q^\theta_{[0, t]}}\PERIOD
\end{equation*} 
Therefore, the minimization of the RER 
\begin{equation*}
\min_{\theta\in \Theta}\ENTRATE{P}{Q^\theta }\COMMA
\end{equation*}
is the appropriate  optimization problem {\  for  $ t\to \infty $, and  for stationary Markov processes} \cite{KP2013}. 
In work \cite{HKKP2016}, we prove that the path-space variational inference problem \VIZ{VP-path1} in continuous time reduces to a   path-space force matching optimization, for a class of CG mappings. 
In addition,  the RER reduces to the FM for stationary processes with invariant probability density $\mu(\bbq)$ defined in \eqref{eq:Gibbs}.
 
 For discrete time observations the CG path-space distribution \eqref{eq:CGpsdistribution}, assuming Markovianity for the CG model, is
  \begin{equation*}\label{eq:CGdiscretsps}
      \bar{Q}^{\theta}(\bar\bbX_1, \dots, \bar\bbX_{N}) =  \bar\nu(\bar\bbX_1)\prod_{i=1}^{N-1} \bar{q}^{\theta}(\bar\bbX_i, \bar\bbX_{i+1})\COMMA
 \end{equation*} 
 where $\bar q^{\theta}(\bbxs, \bbxs')$ is the transition probability kernel of the proposed approximate CG process, and $\bar\nu(\bbxs)$ denotes the initial distribution.
 %of the process  for the microscopic time series $\Dd_{N_p,N_t}$ \eqref{eq:DataPath}. 
 Introducing an unbiased estimator for the relative entropy,  the optimal parameter estimate  for $\Dd^{ts}_{N_p,N_t}$ is given by  
\begin{equation}
  \hat \theta^{ts}_{N_pN_t}(X^{1}_1,\dots,X^{N_p}_{N_t}) = \underset{\theta\in\Theta}\argmin \frac{1}{N_p}\sum_{k=1}^{N_p}\log \frac{ {P}(  X^k_1,   X^k_2,...,  X^k_{N_t})}{ 
  	  \bar{Q}^\theta({\COP X^k_1, \COP X^k_2,...,\COP X^k_{N_t}})}\COMMA
  \end{equation}
where we assume that the $Q^\theta_{[0, t]}$ in relation \eqref{VP-path1} is given as the product of  $\bar{Q}^\theta$ and a back-mapping probability independent of $\theta$, which for notation simplicity we do not present here. 
In terms of the transition probability kernels, the parameter estimator is 
\begin{equation} \label{eq:thetaTrans}
  \hat \theta^{ts}_{N_pN_t}(X^{1}_1,\dots,X^{N_p}_{N_t}) = \underset{\theta\in\Theta}\argmin \frac{1}{N_p} \sum_{k=1}^{N_p} \frac{1}{N_t-1}\sum_{i=1}^{N_t-1}\log \frac{ {p}(  X^k_i, X^k_{i+1})}{{\bar{q}}^\theta(\COP X^k_i,\COP X^k_{i+1})}\PERIOD
\end{equation}
Note that when the time series are  stationary,  then they are statistically indistinguishable and 
the path-space  optimization problem \VIZ{eq:thetaTrans}  reduces to the RER optimization, \cite{HKKP2016}. That is, for observations $\Dd_{N_t} = \{X_1,\dots,X_{N_t}\} $ the optimal parameter set is given by 
\begin{equation}\label{eq:RE_path_N}
  \hat \theta^{ts}_{N_t}(X_1,\dots,X_{N_t}) = \underset{\theta\in\Theta}\argmin   \frac{1}{N_t-1}\sum_{i=1}^{N_t-1}\log \frac{ {p}(  X_i, X_{i+1})}{{\bar{q}}^\theta(\COP X_i,\COP X_{i+1})}\PERIOD
\end{equation}
The  RER  estimator becomes the RE estimator  when  the samples  are replaced by i.i.d. generated from the stationary probability distribution $\mu(\bbx)$ and   ${\bar{q}}^\theta(\COP X_i,\COP X_{i+1}) = \bar\mu(\COP X_{i+1};\theta) $. 

 Note, that the major difference between the RE minimization and the RER minimization is that in the first we need i.i.d.\ data from $\mu(\bbq) $ while in later we need time series data from $P_{[0,t]}$. This is an advantage of the path-space methods since there is no computational effort to generate the i.i.d.\ data.   
 On the other hand, due to the ergodic theory, when the time-series data  is long enough we can  substitute the configuration space average with the time space average where correlated data are admissible. Thus, the effort to generate i.i.d. data is transferred to the effort to generate long time correlated data.  

\section{Confidence intervals for coarse-grained methods}\label{sec:CI}
In this section, we asses the uncertainty of the estimated parameters $\theta$, as well as quantities of interest given as composite functions of the parameters. Specifically,  we construct confidence intervals (CIs) on the CG model parameters for a given set of data for both equilibrium and path-space models. We demonstrate the methodology of constructing  non-asymptotic and asymptotic confidence intervals in detail for the  relative entropy estimation  $\hat \theta^{iid,re}_N$. The  methodology is also valid for the force matching  estimation $\hat \theta^{iid,fm}_N$, if we consider  it as a regression problem with the corresponding likelihood, which is proportional to   $\exp\left\{-\|F(\bbx) - \bar{F}(\bbx;\theta)\|^2\right\} $.
 
\subsection{Statistical estimation and path-space relative entropy optimization}\label{sec:SE-RE}     
As described in the previous section, we consider two types of data; i.e., sets of configurations derived from the more detailed microscopic  simulations in the form of:
(a)  {\it independent and identically distributed} data,  $\Dd^{iid}_{N}$  generated from the invariant  distribution  $\mu$, and
(b) {\it discrete time-series} data $\Dd^{ts}_{N_p,N_t}$, eq. \eqref{eq:DataPath}, generated from the path distribution of the original microscopic process $P_{[0,t]}$. 
Note that eq. \eqref{eq:RE_equi_N} simplifies further since the  invariant measure $\mu$ is  independent of $\theta$,
\begin{equation}
\label{eq: RE_equi_N2}
\hat \theta^{iid, re}_{N}(X_1,\dots,X_N) = \underset{\theta\in\Theta}\argmax \frac{1}{N} \sum_{i=1}^N \log \CGmu(\COP X_i)     \PERIOD
\end{equation}
For the time-series data  the optimization problem  \eqref{eq:RE_path_N} is equivalent to
\begin{equation}
\label{eq: RE_path_N2}
\hat \theta^{ts}_{N_pN_t}(X^{1}_1,\dots,X^{N_p}_{N_t}) = \underset{\theta\in\Theta}\argmax \frac{1}{N_p} \sum_{k=1}^{N_p} \frac{1}{N_t-1}\sum_{i=1}^{N_t-1} \log    \bar{q}^\theta(\COP X^k_i, \COP X^k_{i+1}) \PERIOD
\end{equation}
  
Thus, to derive the optimal CG model parameter, in both cases, we  need (a) the data from the microscopic process,  (b) the pre-defined CG mapping $\COP$, and  (c) the parameterized coarse-grained model. These characterize  the data and physics driven nature of the coarse-graining approach, which relates the true CG model to its digital-twin, the approximate CG model, \cite{HKK-ERCIM2018}.

However, in many situations only a small number $N$ of data is available due to the extreme cost to generate them, either experimentally or numerically. This is precisely the case in the coarse-graining of macromolecular (polymeric) systems, where the cost to generate i.i.d.\ samples increases strongly with the molecular length. For example, for polymer melts, the maximum relaxation time of entangled linear chains
scales with the cubic power of their length; for other architectures the dependence is even stronger, e.g., for star  polymers becomes exponential~\cite{doi1988theory}.
This is evident in section~\ref{sec:polyethylene}, where we derive the optimal CG force field for a polyethylene melt and the corresponding confidence intervals.

\subsection{Non-asymptotic confidence intervals}
\label{sec:Nonasymptotic_results}
There is a vast need for  statistical information about parameters in CG models, especially when the size of data is limited. Such information would provide estimates of whether those parameters are in a reasonable region, and whether they are sensitive to the data. 
Here we present two statistically rigorous non-asymptotic methods to compute standard errors and construct confidence intervals,  namely the  jackknife and the  bootstrap~\cite{wasserman2010all, EfronHastie2016}.
 These techniques are valid for the i.i.d.\ case $\Dd^{iid}_{N}$, as well as for multiple i.i.d. time-series $\Dd^{ts}_{ N_pN_t}$, but not for the correlated data of a single time-series.  We will  apply the  jackknife and bootstrap methods 
to construct  confidence bounds for the CG parameters.

\medskip
\noindent
\subsubsection{ The Jackknife}
Let us  denote $\hat\theta_N= \hat\theta^{iid}_N$
and $\hat{\theta}_{(-i)} $ the  estimators of the CG parameters, from  $\Dd_{N}=\{X_1,\ldots,X_N\}$ and with the $i$-th observation $X_i$ removed respectively, i.e., 
\begin{equation*}
    \hat{\theta}_{(-i)} = \hat{\theta}_{N-1} (X_1,\dots,X_{i-1},X_{i+1},\dots, X_n) \PERIOD
\end{equation*}
Let also $\tilde{T}_i$ be the pseudo-values
$$\tilde{T}_i = N\hat{\theta}_N - (N-1)\hat{\theta}_{(-i)} \PERIOD$$
Then, the {jackknife} variance estimation is  
\begin{equation*}
    V_{jack} = \frac{\sum_{i=1}^{N} \left(\tilde{T}_i - \frac{1}{n}\sum_{i=1}^N \tilde{T}_i \right)^2 }{N(N-1)} = \frac{N-1}{N}\sum_{i=1}^N \left(\hat{\theta}_{(-i)} - \frac{1}{N}\sum_{i=1}^N \hat{\theta}_{(-i)} \right)^2\COMMA
\end{equation*}
and the  corresponding    standard  confidence interval is 
 \begin{equation}\label{eq:CIjackknife}
CI_{jack} = \left[\hat{\theta}_N- z_{\alpha/2}\sqrt{V_{jack}} \quad , \quad \hat{\theta}_N+ z_{\alpha/2}\sqrt{V_{jack}}\right]\PERIOD
\end{equation}
The jackknife method  consistently estimates the variance of $\hat \theta_N$, though it cannot produce consistent estimates of the standard error of sample quantiles. 
The bootstrap method, on the other hand, is able to produce not only variance estimation  but also quantile estimates and thus non-symmetric confidence intervals, as discussed below. 

\medskip
\noindent
\subsubsection{ The Bootstrap}\label{sec:Bootstrap}
To construct bootstrap  confidence intervals, firstly we assume that the empirical distribution of the data $ \Dd_N$  is $\hat{F}_N$ mimicking the true distribution. 
Then,  $B$ bootstrap samples are generated i.e.,   $B$   sets of samples $X_1^*, \ldots,X_N^*$ are drawn from $\hat{F}_N$. The procedure is described by the following steps:

\begin{enumerate} 
    \item Draw $N$ new samples $X_1^*, \ldots,X_N^* \sim \hat{F}_N $,  i.e., draw $X_i^*$ randomly from $\Dd_N=\{X_1,\ldots,X_N\}$ with equal probability and  with replacement.
    \item Compute $\hat{\theta}^*$ according to the chosen estimator, e.g.,\ \eqref{eq: RE_equi_N2} for i.i.d data.  
    \item Repeat steps 1 and 2,  $B$ times to get $\hat{\theta}^*_1,\ldots, \hat{\theta}^*_B $.
    \end{enumerate}
With this procedure we construct an approximate distribution of the statistical estimator $\hat{\theta}_N$.  
There are several approaches to construct bootstrap confidence intervals, such as the standard, the pivotal, the percentile, and the bootstrap-t intervals, \cite{EfronHastie2016}, \cite{wasserman2010all}. In the current work, we estimate the standard and percentile  confidence intervals which we present next.
The bootstrap variance estimation is 
    \begin{equation}
        V_{boot} = \frac{1}{B}\sum_{i=1}^B \left(\hat{\theta}^*_i - \frac{1}{B}\sum_{b=1}^B \hat{\theta}^*_b \right)^2\COMMA
    \end{equation}
and the  bootstrap  standard  confidence interval is  
\begin{equation}\label{eq:CIboot}
    CI_{s,boot} = \left[\hat{\theta}_N - z_{\alpha/2}\sqrt{V_{boot}} \quad, \quad \hat{\theta}_N + z_{\alpha/2}\sqrt{V_{boot}}\right]\PERIOD
\end{equation}
The bootstrap percentile   confidence interval is given directly from the bootstrap distribution of the statistical estimator $\hat{\theta}$, and is 
\begin{equation} 
    CI_{p,boot} = \left[  \hat{\theta}^*_{\alpha/2} \quad , \quad \hat{\theta}^*_{1-\alpha/2} \right]\COMMA
\end{equation}
where $ \hat{\theta}^*_{\alpha/2}$ is the $\alpha/2$ percentile of $\hat{\theta}^*_1,\ldots,\hat{\theta}^*_B$.

The percentile bootstrap intervals are not accurate though if {\  bootstrap estimates} are highly biased and skewed. {\  Highly biased bootstrap estimates can not represent the true distribution, and highly skewed bootstrap estimates concentrate more on one side of the distribution and thus has a long tail on the other side.} There are improved intervals but more complicated, such as the bias-corrected and accelerated bootstrap (BCa). 
BCa corrects for bias and skewness in the distribution of bootstrap estimates and improves the coverage accuracy of standard intervals from first order to second order, thus provides reasonably narrow intervals but is complicated  to implement \cite{diciccio1996bootstrap}.  
Both techniques, the jackknife,  and bootstrap  use part of the data to get several estimators for the parameters and then use those estimators to construct confidence intervals. 
Bootstrap can have higher computational cost if the number of bootstrap samples ($B$) is larger than the number of data $(N)$, which is often the case. Thus, the jackknife method is less computationally expensive  but is less general. 
As reported in literature, \cite{EfronHastie2016}, empirical evidence suggests that $B=200$ is usually sufficient for evaluating  the bootstrap estimate of the standard error, but larger values  should be considered for the bootstrap confidence intervals.

\subsection{Asymptotic confidence intervals}
\label{sec:asymptotic_results}
{\noindent \bf I. Independent, identically distributed data.}
Recall that for   i.i.d. data, the RE optimal parameter is 
\begin{equation*}
\hat{\theta}^{iid}_N =  \underset{\theta\in\Theta}\argmax\frac{1}{N}\sum_{i=1}^N \log \CGmu(\COP X_i)\PERIOD
\end{equation*}
Note that  $\hat{\theta}^{iid}_N$ is similar to the maximum likelihood estimator, \cite{EfronHastie2016}. The difference is that the maximum likelihood estimator assumes that $\COP X_i$ has measure $\CGmu$, while this assumption is not true here. Thus, a confidence interval directly obtained from maximum likelihood likelihood estimator is inaccurate. We resolve this issue by constructing a slightly different confidence interval along with two versions of the Fisher information:
\begin{align*}
\hat{\Ff}_1 &= - \frac{1}{N} \sum_{i=1}^N \nabla_\theta^2 \log \CGmu(\COP X_i)|_{\theta = \hat{\theta}^{iid}_N} \COMMA\\%\label{eq: F1_fisher} \\
\hat{\Ff}_2 &=  \frac{1}{N} \sum_{i=1}^N (\nabla_\theta \log \CGmu(\COP X_i))(\nabla_\theta \log \CGmu(\COP X_i))^T|_{\theta = \hat{\theta}^{iid}_N} \PERIOD %\label{eq: F2_fisher}
\end{align*}
These two Fisher information matrices are close if $\COP X_i$ has distribution $\CGmu$ and under the assumption that $N$ is large enough (see Corollary 1.1.3 in supplementary information). Whether $\hat{\Ff_1}$ is close to $\hat{\Ff_2}$ could be an indirect indicator of whether the parameterized CG distribution $\CGmu$ can mimic the distribution of $\COP X_i$. That is, $\hat{\Ff_1}$ is close to $\hat{\Ff_2}$ indicates that $\COP X_i$ has a measure  close to $\mu^{\theta^{iid}_N}$, which means that the parameterized family of $\mu^\theta$ can reconstruct the distribution of $\COP X_i$. But the inverse might not be true in general.
The asymptotic theory  provides the $1-\alpha$ confidence interval for $\theta$ in the equilibrium model (see Theorem 1.1.1 in supplementary information), which  is 
\begin{equation}
    \label{eq: CI_equil}
  CI_{iid} = \left[\hat{\theta}^{iid}_N - \frac{z_{\alpha/2}}{\sqrt{N}}\sqrt{\hat{\Ff}^{-T}_1\hat{\Ff}_2\hat{\Ff}^{-1}_1 }\quad, \quad \hat{\theta}^{iid}_N + \frac{z_{\alpha/2}}{\sqrt{N}}\sqrt{\hat{\Ff}^{-T}_1\hat{\Ff}_2\hat{\Ff}^{-1}_1 } \right]\PERIOD
\end{equation}

\medskip
\noindent
{\bf II. Time-series data.} 
In the  path-space models, the result is similar to the one in the equilibrium models where $\CGmu$ is replaced by transition probability density $\bar{q}^\theta$ and a more complicated Fisher information. Recall that
$$\hat{\theta}^{ts}_N =\underset{\theta\in\Theta}\argmax \frac{1}{N-1}\sum_{i=1}^{N-1}\log\bar{ q}^\theta(\COP X_i, \COP X_{i+1})\PERIOD $$
The first Fisher information matrix 
\begin{equation*}
    \label{eq: I1_fisher}
    \hat{I}_1 = - \frac{1}{N-1} \sum_{i=1}^{N-1} \nabla_\theta^2 \log \bar{q}^\theta(\COP X_i, \COP X_{i+1})|_{\theta = \hat{\theta}^{ts}_N} ,
\end{equation*}
while the second Fisher information matrix is given by  the Markov chain central limit theorem and is estimated by a batch means estimator, \cite{jones2004markov, jones2006fixed}
\begin{equation*}
    \label{eq: I2_fisher}
    \hat{I}_{2,BM} = \frac{b}{a-1} \sum_{j=1}^a(\bar{Y}_j - \bar{Y})\COMMA
\end{equation*}
where 
$\bar{Y}_j = \frac{1}{b} \sum_{i=(j-1)b+1}^{jb-1} \log \bar{q}^{\hat{\theta}^{ts}_N}(\COP X_i, \COP X_{i+1})$, $\bar{Y} = \frac{1}{N-1} \sum_{i=1}^{N-1} \log \bar{q}^{\hat{\theta}^{ts}_N}(\COP X_i, \COP X_{i+1})$ and $N=ab$.
Thus the $1-\alpha$ confidence interval for $\theta$ in the path-space models is
\begin{equation}
    \label{eq: CI_path_space}
  CI_{ts} =  \left[\hat{\theta}^{ts}_N - \frac{z_{\alpha/2}}{\sqrt{N-1}}\sqrt{\hat{I}^{-T}_1\hat{I}_{2,BM}\hat{I}^{-1}_1 }\quad, \quad \hat{\theta}^{ts}_N + \frac{z_{\alpha/2}}{\sqrt{N-1}}\sqrt{\hat{I}^{-T}_1\hat{I}_{2,BM}\hat{I}^{-1}_1 } \right] \PERIOD
\end{equation}
In the supplementary information accompanying this work, we present the mathematical justification for the confidence intervals provided here. 

\subsection{Estimating quantities of interest}
Thus far, we have estimated the coarse model parameters and assessed their accuracy. Now, our interest is in finding an  estimator and their uncertainty for  quantities of interest (QoI),
\begin{equation}\label{eq:QoI}
    \tau= g(\theta),
\end{equation}
which  are functions of $\theta$. 
Given a data set $\Dd^{iid}_N$,   the invariance principle ensures that  the MLE estimator of the QoI $\tau$,  is
\begin{equation}\label{eq:QoIestimator}
    \hat\tau   =  \hat\tau(X_1,\dots, X_N) = g\left(\hat\theta_N(X_1,\dots,X_N)\right)\COMMA
\end{equation}
where $\hat\theta_N $ is the set of estimated model parameters, \cite{casella2002statistical,wasserman2010all}.

%\medskip
%\noindent
%{\bf Asymptotic estimates.}
The delta method provides asymptotic standard errors for $\hat \tau$
\begin{equation*}
    \hat{se}(\hat \tau) = \sqrt{(\hat \nabla g)^{tr}\hat J_N \hat\nabla g}\COMMA
\end{equation*}
where $\hat J_N = \hat{\Ff}^{-1}_1$, and $ \hat \nabla g$ is  $ \nabla g = (\frac{\partial{g}}{\partial{\theta_1}}, \dots,\frac{ \partial{g}}{\partial{\theta_K}})^{tr}$ evaluated at $ \theta= (\hat \theta_N)$. {\  Here $(\cdot)^{tr} $ denotes matrix transpose.}

%\medskip
%\noindent
%{\bf Non-asymptotic estimates.}
The non-parametric resampling  methods, jackknife and bootstrap described in \ref{sec:Nonasymptotic_results}, apply  straightforwardly on 
$\hat\tau(X_1,\dots,X_N)$ through $\theta_N$ and the use of the invariance property. Indeed, the percentile bootstrap CI is 
\begin{equation}\label{eq:CI-QoI}
    CI^{qoi}_{p,boot} = \left[  \hat{\tau}^*_{\alpha/2} \quad , \quad \hat{\tau}^*_{1-\alpha/2} \right]\COMMA
\end{equation}
where $ \hat{\tau}^*_{\alpha/2}$ is the $\alpha/2$ percentile of $\hat{\tau}^*_1,\ldots,\hat{\tau}^*_B$, and $ \hat{\tau}^*= \hat\tau(X^*_1,\dots,X^*_N)$, for a bootstrap sample $(X^*_1,\dots,X^*_N)$, described in section   \ref{sec:Bootstrap}.

 \begin{remark}  {\normalfont 
 Bayesian analysis  can provide a range of information about the model through the posterior   probability distribution of the  model parameters. Credible intervals are  thus obtained from the posterior.  However, the need for prior information for the parameters is a drawback,  since it is often not available.   Of course, there exist techniques to overcome this, such as uninformative priors and  hyper-parameters, but still some prior knowledge is necessary.
 In contrast, frequentist parametric and non-parametric confidence intervals require no prior information.
 %, making them more convenient. 
  Also, the non-parametric, uninformative posterior distribution  can be approximately represented by a  bootstrap distribution which may be   much easier to obtain, \cite{friedman2001elements}. }
 \end{remark}

%%%%%%%%%%%%%%%%%%%%%%%%%%%%%%%%%%%%%%%%%%%%%%%%%%%%%%%%%%%%%%%%%%%
% Applications% 
%%%%%%%%%%%%%%%%%%%%%%%%%%%%%%%%%%%%%%%%%%%%%%%%%%%%%%%%%%%%%%%%%%% 
\section{Test-bed 1: Two-scale diffusion processes}\label{sec:two-scale}
In this section, we benchmark our methodology by considering  a two-dimensional  two-scale diffusion process. This diffusion process is a good, relatively simple, example that allows us to  (a) test and compare the accuracy of the estimated parameters by the different optimization methods, (b) provide the corresponding confidence intervals, and (c) validate the results  since we  know the effective dynamics analytically.

The two-scale diffusion process consists of a slow variable $X_t^\epsilon \in \R$ and a fast variable $Y_t^\epsilon \in \R$, for $t\ge 0 $, which satisfy the system of stochastic differential equations, 
\begin{eqnarray}\label{eq:two-scale}
dX_t^\epsilon &=& -Y_t^\epsilon dt + dW_t^1  \COMMA\\
dY_t^\epsilon &=& -\epsilon^{-1}(Y_t^\epsilon - X_t^\epsilon) dt + \epsilon^{-0.5}dW_t^2 \COMMA \nonumber
\end{eqnarray}
for $\epsilon >0$,  where $dW_t^1$ and $dW_t^2$ are independent standard Wiener processes, \cite{Oksendal2003}.
As  $\epsilon \to 0$, $X_t^\epsilon $ follows the effective process $\hat{X}_t$  which  is  proved to satisfy,  by the averaging principle, \cite{freidlin2012random} 
\begin{equation}\label{eq:eff.model}
d{ \hat{X} }_t = -  \hat{X}_t dt + d\hat{W}_t \PERIOD
\end{equation}
 Note that  the effective potential driving the  process $\hat{X}_t$  is the harmonic potential $U(x) = \frac12 x^2$, depicted in Figure \ref{fig:two_scale_potential_plot}. 

\begin{figure}[htbp]
    \centering
    \includegraphics[width = 0.6\textwidth]{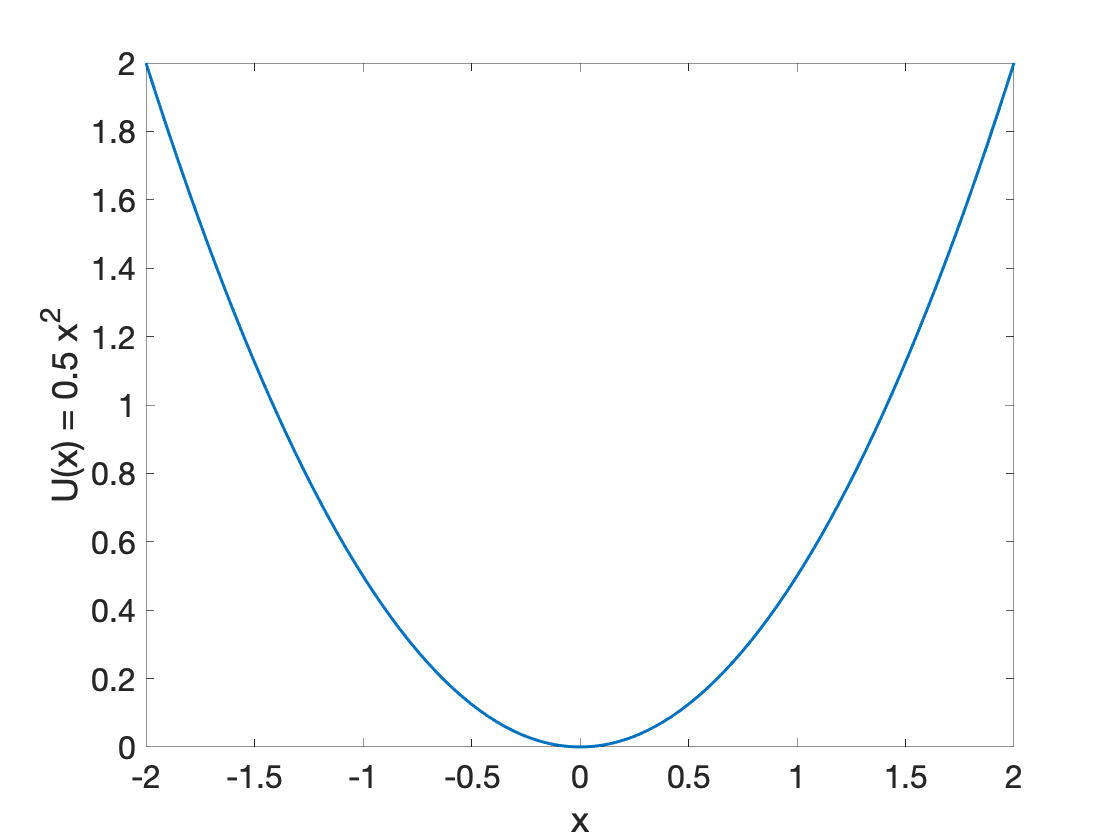}
    \caption{The effective  potential for the coarse process $\hat X_t$ is the harmonic interaction potential.}
    \label{fig:two_scale_potential_plot}
\end{figure}

We are interested in constructing a coarse-grained model for the slow variable $X_t^\epsilon$ and for a finite value of $\epsilon > 0$. Thus,  the CG map is $\Pi: (X_t^\epsilon, Y_t^\epsilon) \rightarrow X_t^\epsilon$.  The CG process $X_t^{CG}$, approximating 
 $\Pi (X_t^\epsilon, Y_t^\epsilon)$, is assumed to satisfy 
\begin{equation}
dX_t^{CG} = a(X_t^{CG};\theta) dt + dW_t \COMMA\label{eq:CGmodel}
\end{equation}
where $W_t$ is a standard Wiener process. 
To approximate the coarse-grained dynamics we propose an effective   drift
 \begin{equation}\label{eq:CGdrift}
 a(x;\theta) = \sum_{k=1}^K \theta_k x^{k-1} \COMMA
 \end{equation}
that is an approximation over the set of polynomials $\{1, x, \dots, x^{K-1} \}$. In the example presented   we choose $K=5$.
Note that in this example, we expect that the estimated parameters of the coarse grained model are close to $\theta^* = [0, -1, 0, 0, 0]$, due to the known analytical form of the effective dynamics for the process, \eqref{eq:eff.model}.
We present next a comparison between \eqref{eq:eff.model} and \eqref{eq:CGmodel} by investigating the uncertainty of parameters through confidence intervals.

\bigskip
{\noindent \bf I. Independent, identically distributed data.} Firstly, we investigate the results with i.i.d data $\Dd_N=\{ (X_i, Y_i)\}_{i=1}^N$, corresponding to the invariant density of \eqref{eq:two-scale}.  We omit the notation of $\epsilon$-dependence for notation simplicity. 
We minimize the RE  between the invariant densities of $X_t^{\epsilon}$ and of $X_t^{CG}$. The invariant density of  $X_t^{CG}$ is 
\begin{equation*}
   \CGmu(x) = \frac{1}{Z^{\theta}} e^{-2 \bar U(x;\theta)},
\end{equation*}
where $\bar U(x;\theta) $ is defined by
$a(x;\theta) = - \frac{d}{dx} \bar U(x;\theta) $ and $ Z^{\theta} = \int e^{-2 \bar U(x;\theta)} dx$. The optimal parameter   is   given by
\begin{eqnarray}\label{eq:tsRE}
   {\theta}^{iid, re}  &=& \underset{\theta\in\Theta}\argmax \EXPECT_{\mu}[  \log \CGmu \circ \COP ] \nonumber\\
             &=& \underset{\theta\in\Theta}\argmax \left\{-2 \EXPECT_{\mu}[   \bar U(\cdot;\theta) ] - \log Z^{\theta} \right\}\PERIOD
\end{eqnarray}
The  RE estimator is  described in eq. \eqref{eq: RE_equi_N2}, while the optimization method in section 2.3 of the supplementary information.  We also apply the FM method for which the  optimal parameter  estimator is
\begin{equation}\label{eq:FMtheta-twoscale}
    \hat{\theta}^{iid, fm}_{N} = \underset{\theta\in\Theta}\argmin \frac{1}{N} \sum_{i=1}^N |Y_i + a(X_i;\theta)]^2\PERIOD
\end{equation} 

\bigskip
{\noindent \bf II. Time-series data.} 
{\  
Secondly, we estimate the parameter for   time-series data $\Dd_{N_pN_t} = \{ (X^k_i, Y^k_i)\}_{k=1,i=1}^{N_p,N_t}$. 
The approximate transition probability density of \eqref{eq:CGmodel}  is  
\begin{equation} 
\label{eq:transition_prob}
    q_h^\theta (X_i,X_{i+1}) =\frac{1}{z} e^{-\frac{1}{2h^2}|X_{i+1} - X_i - a(X_i;\theta)h|^2},
\end{equation}
where $z$ is a normalized factor independent of $\theta$,  $h$ is the discretization time step for the Euler-Maruyama approximation of $X_t^{CG}$ with corresponding transition density $q_h^\theta (x,x') $.
For multiple time-series data $\Dd_{N_pN_t}$, the appropriate estimator is the path-space RE (PSRE). The optimal estimate  given by the minimization problem  \eqref{eq: RE_path_N2} is
\begin{equation}\label{eq:PSREtheta-twoscale}
    \hat{\theta}^{ts, psre}_{N} = \underset{\theta\in\Theta}\argmin  \frac{1}{N_p} \sum_{k=1}^{N_p}\frac{1}{N_t-1}\sum_{i=1}^{N_t-1} |X^{k}_{i+1} - X^{k}_{i} - a(X^{k}_i;\theta) h|^2\PERIOD
\end{equation}
The corresponding  RER estimator, valid for a long, stationary time series, is
\begin{equation}\label{eq:RERtheta-twoscale}
   \hat{\theta}^{ts,rer}_{N} = \underset{\theta\in\Theta}\argmin  \frac{1}{N_t-1}\sum_{i=1}^{N_t-1} |X_{i+1} - X_{i} - a(X_i;\theta) h|^2\PERIOD
\end{equation}
  as described in section \ref{sec:Background}.
Moreover, in the equilibrium region the RER minimization is equivalent to the FM minimization. We describe the proof in section 2.2 of the supplementary information.}

\bigskip
{\noindent \bf III. Asymptotic results.} We begin with  reporting the  results for a 'large' sample size and the corresponding asymptotic confidence intervals as described in section \ref{sec:CI}.
In all the numerical tests we fix $\epsilon= 0.005$ and $h=0.01$.
For the RE estimation, \eqref{eq:tsRE}, we applied the Newton-Raphson (NR) algorithm. To estimate the normalization parameter $Z^{\theta}$ which changes at each NR iteration, we  generated $5,000$ CG i.i.d.\ samples from $ \CGmu(x)$ 
 with a Hamiltonian Monte Carlo sampler. 
The NR algorithm converged after $20$ iterations,  with initial value of $\theta$ near $\theta^*$. The details of the NR method are described in the supplementary information. 
For the FM and RER estimation we solve  the corresponding least squares problem described in \eqref{eq:FMtheta-twoscale} and \eqref{eq:RERtheta-twoscale}.

Firstly, we generate two sets of samples: (a) $\Dd_N$,  i.i.d.\ samples from the invariant distribution of the exact process with   $N=500$, and (b) $\Dd_{N_t}$ one time-series samples  with $N_t=50,000$. 
 Note that we have experimented with various values of $N$ and $N_t$.
  We chose to  report the $N=500$ and $N_t=50,000$ so that the optimization methods show variance estimates of the same order.

Figures \ref{fig:two-scale_FM} and  \ref{fig:two-scale_RE} show the results for the FM  and the relative entropy minimization with $N=500$ i.i.d.\ data  respectively. 
In both figures, the right hand side  depicts the invariant probability density function of the estimated coarse process $X_t^{CG}$ and  of the exact process $X_t^{\epsilon}$. 
The left hand side figure presents the estimated parameters and the corresponding $95\%$ asymptotic standard confidence interval, defined in \eqref{eq: CI_equil}.
Similarly, figure \ref{fig:two-scale_TS} depicts the parameter estimates with the asymptotic CI and the  invariant probability density functions of the estimated and the exact process with one correlated time-series with time step $h=0.01$ and size $N_t=50,000$.   Also, in table~\ref{tab:FM vs RE} we present the point parameter estimates, the asymptotic variance and the computational cost for the RE, the FM and the RER optimization methods. 
 
\begin{table}[htbp]
      \centering
      \resizebox{\columnwidth}{!}{\begin{tabular}{|c|c|c|c|c|c|}
      \hline
            Method &  $\hat{\theta}$ & $ \hat{\sigma}^2$ & CI & Number of samples & CPU time (sec)\\
            \hline
            \hline  
           FM & $\left[\begin{matrix}  \ 0.0236\\ -1.0240 \\ \ 0.0039\\ -0.0012 \\ -0.0338   \end{matrix}\right] $  &
            $\left[\begin{matrix} 0.0021 \\ 0.0063\\  0.0138\\  0.0023\\ 0.0019   \end{matrix} \right]$ & $\begin{bmatrix} -0.0663  &  0.1135\\
   -1.1790 &  -0.8689\\
   -0.2265  &  0.2342\\
   -0.0947  &  0.0922\\
   -0.1189  &  0.0513 \end{bmatrix}$ & 500 & 0.02\\
          \hline  
          RE  &
          $\left[\begin{matrix}   0.0247\\ -0.9827\\ 0.0260 \\ -0.0640 \\ 0.0001  \end{matrix}\right]  $ & 
          $\left[\begin{matrix}  0.0046\\ 0.0261\\ 0.0439 \\ 0.0145 \\ 0.0048 \end{matrix} \right] $ & $\begin{bmatrix}-0.1390 & 0.1151\\
   -1.2287 &  -0.7505 \\
   -0.3005 &   0.3572 \\
   -0.1345 &   0.1509 \\
   -0.0913 &   0.0572 \end{bmatrix}$  & 500 & 5.32\\
          \hline  
      RER &  	$\begin{bmatrix}   0.0746 \\   -0.9805 \\   0.0483 \\ -0.0313\\   -0.0255 \end{bmatrix}$  &
          $\left[\begin{matrix}  0.0040 \\  0.0112   \\ 0.0158  \\  0.0040 \\   0.0013 \end{matrix} \right] $ & $ \begin{bmatrix}  -0.0491  &  0.1983 \\
   -1.1876 &  -0.7733 \\
   -0.1979 &   0.2944 \\
   -0.1550 &   0.0925 \\
   -0.0965 &   0.0455 \end{bmatrix}$	 & 50,000 &  0.19\\     
          \hline 
      \end{tabular}}
      \caption{Parameter and asymptotic variance estimates for 'large sample' sets for the two-scale diffusion benchmarking problem. The exact  parameters for $\epsilon \to 0$ are  $\theta^*= [0, -1, 0, 0, 0]$.}
    \label{tab:FM vs RE}
\end{table}

All methods approximate well the expected $\theta^*=  [0, -1, 0, 0, 0]$, corresponding to the asymptotic model as $\epsilon \to 0$, as $ \theta^*$ falls into the confidence interval for all methods.  The RE method presents larger asymptotic variance compared to the FM. Moreover, the RE has higher computational cost than the FM. Its benefit though is the better estimation of the 'true' probability density, which in return will give better estimations of  quantities of interest given as expected values. We can notice an excellent match of the CG invariant density with RE estimation  to the exact one, while there is a small difference with the FM and RER estimation. We attribute this difference to the fact the  RE matches directly the probability densities while the FM and RER match the drift terms (i.e.\ the force).  

{\  Next, we comment on the FM and the  RER  methods} from  the point of view of comparing an i.i.d.\ method and a path-space method. 
The results show that we can achieve estimates with the same order of  magnitude with the FM with i.i.d. data and the RER with correlated data if we use about hundred times more data in the later.  This naturally increases the computational cost of the optimization problem.
However, there is a  computational benefit on the generation of the samples, since for the path-space samples (time series) we do not need to reject any generated data. On the contrary, to generate the i.i.d.\ observations we have to reject a large number of simulated data. 
This is extremely insufficient in high dimensional applications, as in long polymer chains discussed in the next section. Therefore the path-space methods  can be  advantageous when we have to generate high-dimensional samples. 
 
On the other hand,  the ergodic theory  ensures that  we can apply the FM  method  for  correlated time series data, as long as the time-series is long enough. 
Therefore, our next numerical study  examines the validity of the FM for  short and long correlated  time-series. 
{\  
That is, we use  the FM estimator for correlated data, and thus introduce the estimator for time series data $\Dd_{N_pN_t}$
\begin{equation}\label{eq:RERtheta2-twoscale}
    \hat{\theta}^{ts,fm}_{N}  = \underset{\theta\in\Theta}\argmin   \frac{1}{N_p} \sum_{k=1}^{N_p}\frac{1}{N_t}\sum_{i=1}^{N_t} |Y^k_{i} + a(X^k_i;\theta) |^2\PERIOD
\end{equation}
}
Table \ref{tab:RER_multiple_estimators} reports the point estimates for time correlated samples, resulting from the FM estimator \eqref{eq:RERtheta2-twoscale} and the PSRE (and RER) estimator \eqref{eq:RERtheta-twoscale}. The table with the estimates for all parameters is provided the supplementary information.
We  observe that the FM point estimates 
improve as the size of the  time-series  increases, as expected.
Comparing the $N_t=50,000$ for the FM and the $N_p=100, N_t=500$ cases for which the number of samples is the same, \ %\sout{we can conclude that the \  RER is closer to the truth.} 
they both yield estimates close to the truth. 
We notice that the FM estimator gives slighlty better estimates  than the  PSRE estimator for short time trajectories, e.g., $N_t =500$, $N_t =5,000$. 
We ascribe  this difference to  that the first uses all fine-scale observations $(X_i, Y_i)$, while the latter only uses the partial observations $(X_i)$. 
To have thus reliable  PSRE estimates, we need to guarantee either the trajectory is long enough or the number of trajectories is large enough. 
Important to note is that for the RER, we can estimate the asymptotic CIs  while  the FM CIs are no longer valid.
}

\begin{figure}[htbp]
\begin{minipage}{0.48\textwidth}
\includegraphics[width=\textwidth]{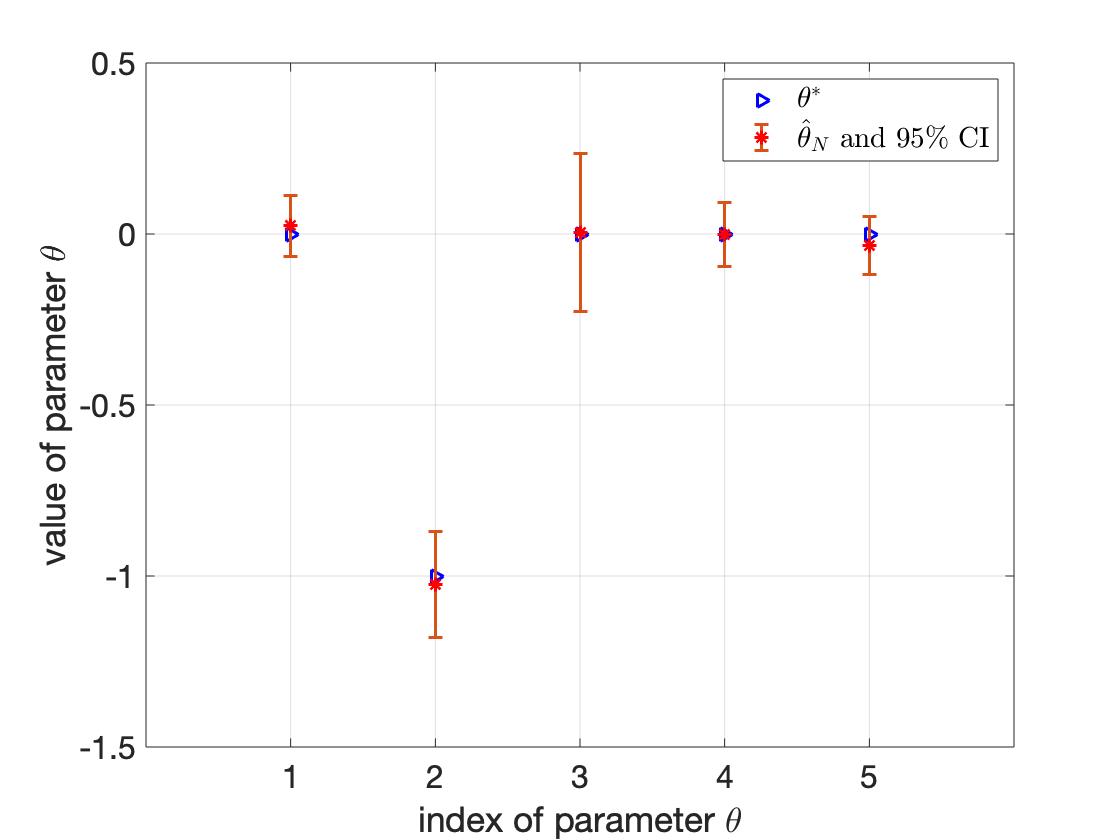}
\end{minipage}
\begin{minipage}{0.48\textwidth}
\includegraphics[width=\textwidth]{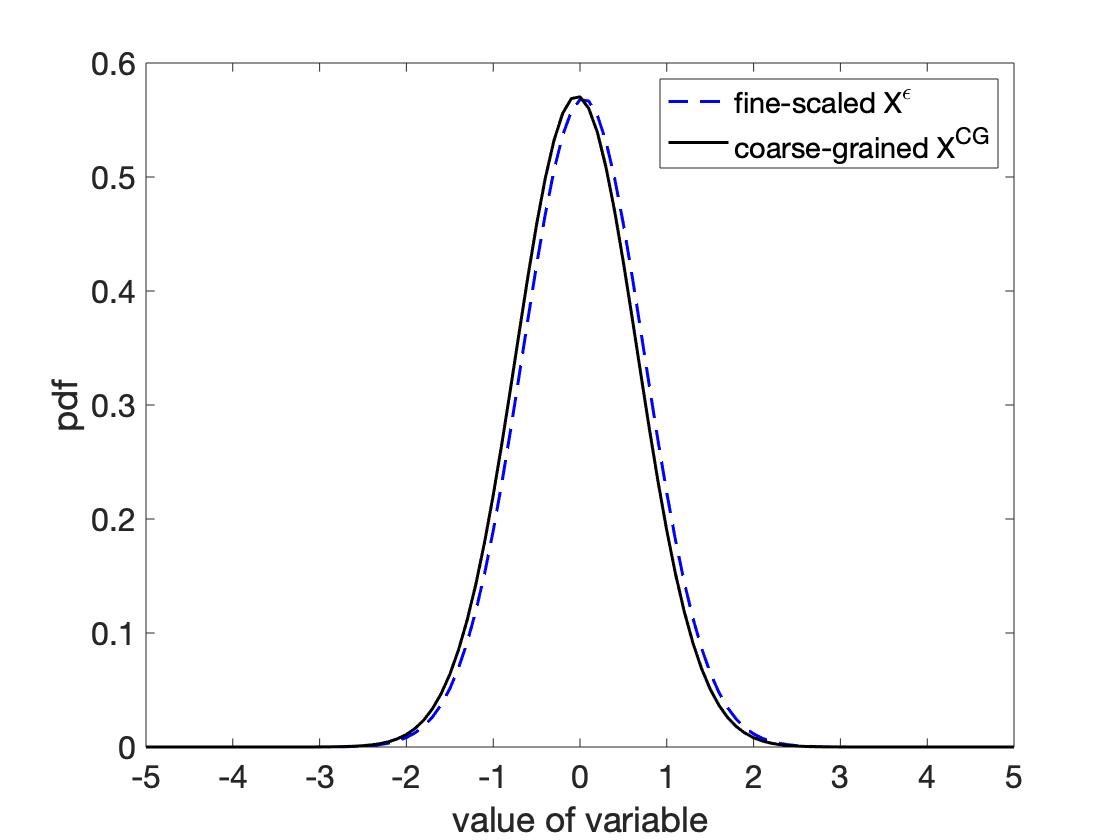}
\end{minipage}
\caption{The estimator and $95\%$ confidence intervals for FM method is shown on the left. The expected values [0,-1,0,0,0] are located inside the intervals. A distribution of constructed CG variable by using $\hat{\theta}_N$ and a comparison with fine-scaled $X^\epsilon$ is shown on the right.}
\label{fig:two-scale_FM}
\end{figure}

\begin{figure}[htbp]
\begin{minipage}{0.48\textwidth}
\includegraphics[width=\textwidth]{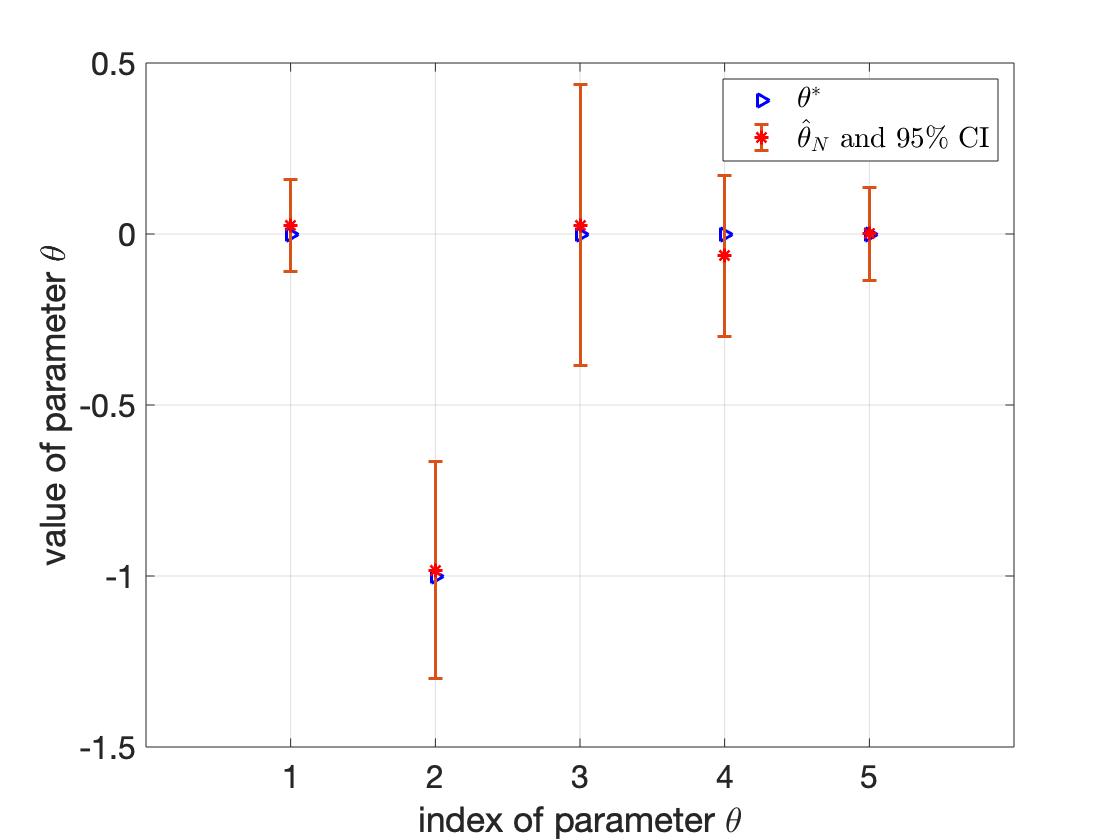}
\end{minipage}
\begin{minipage}{0.48\textwidth}
\includegraphics[width=\textwidth]{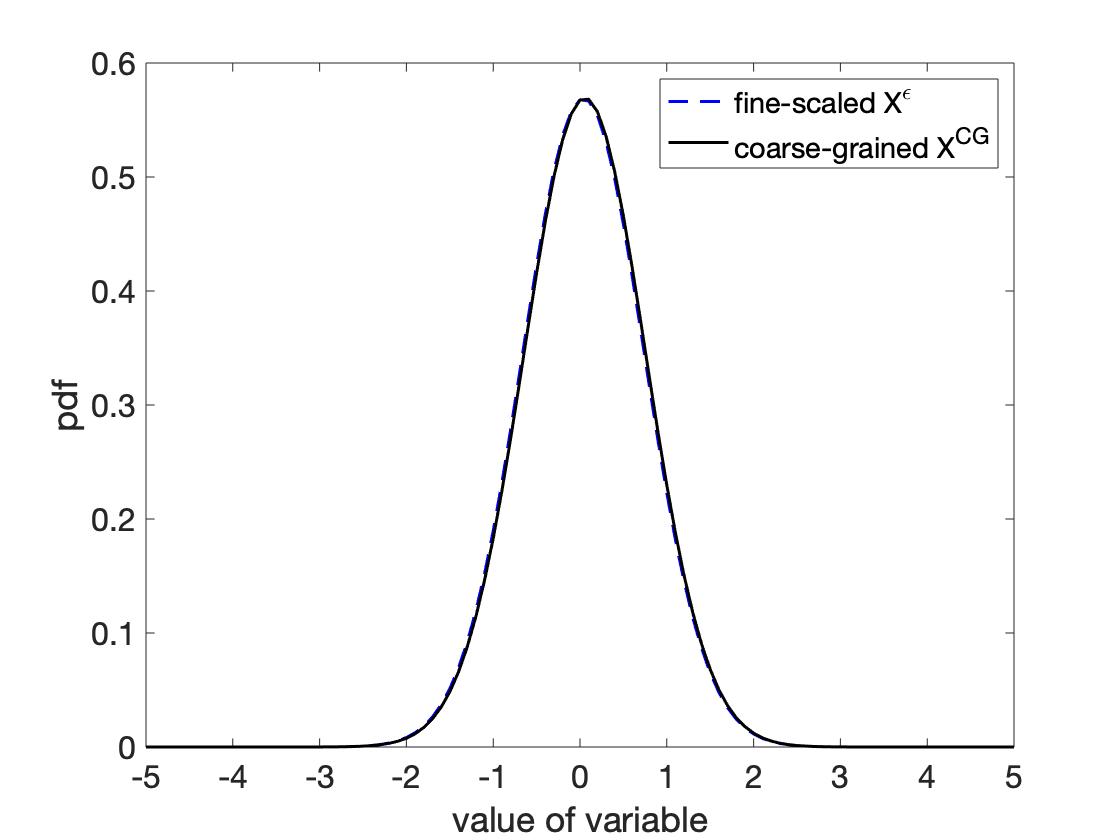}
\end{minipage}
\caption{The left plot shows the results for Relative Entropy minimization. The density of CG variable matches  well the one of $X^\epsilon$. }
\label{fig:two-scale_RE}
\end{figure}

\begin{figure}[htbp]
\begin{minipage}{0.48\textwidth}
\includegraphics[width=\textwidth]{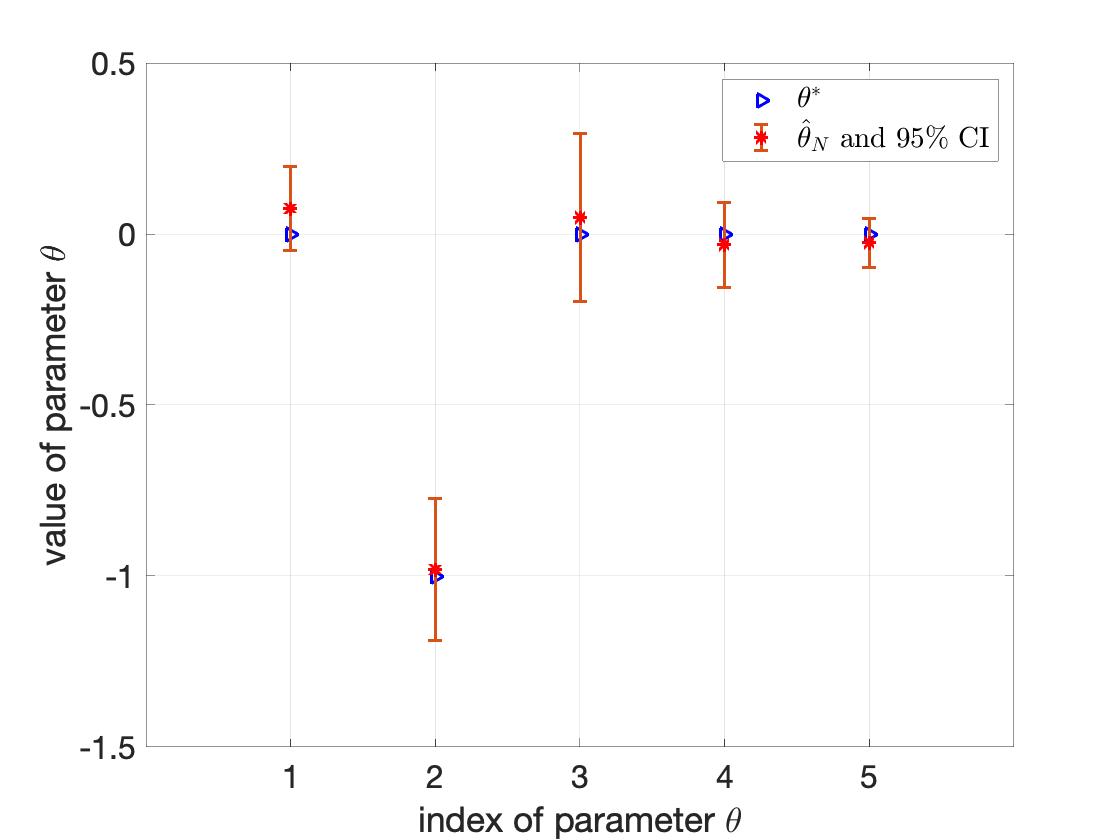}
\end{minipage}
\begin{minipage}{0.48\textwidth}
\includegraphics[width=\textwidth]{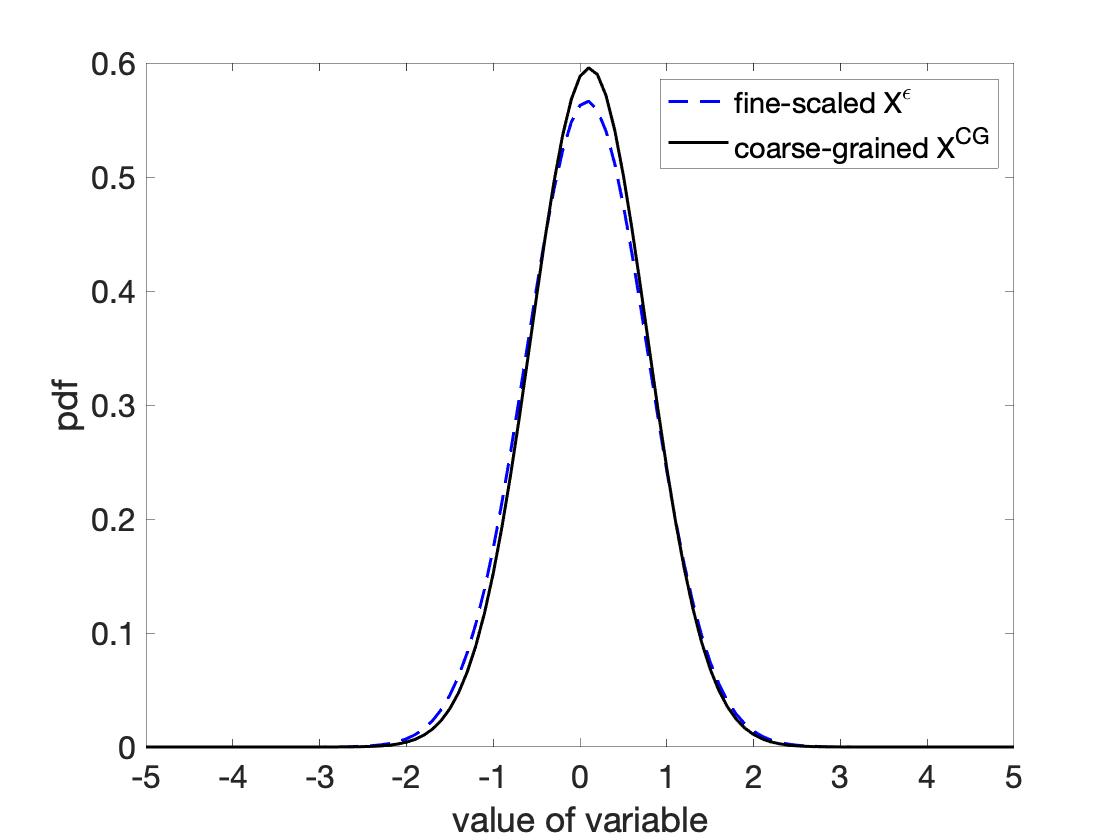}
\end{minipage}
\caption{Path-space optimization with the RER method.}
\label{fig:two-scale_TS}
\end{figure}

\begin{table}[htbp]
\begin{minipage}[t]{0.95\linewidth}
    \centering
    \begin{tabular}{|c|c|c|c|}
    \hline
   $N_p$ & $N_t$ & $\hat{\theta}^{ts,fm}_{N}$ &  $\hat{\theta}^{ts,psre}_{N}$ \\\hline
   1 & 500 &  $ -0.4328   $ & $ -6.0944 $\\\hline
     1  & 5,000 &$  -0.9485  $  & $   -1.1251  $ \\\hline
     1 & 50,000 &$   -0.9734 $ & $ -0.9805  $\\\hline
      10 & 500 &$    -0.8960  $ &$    -1.3288  $\\\hline
     100 & 500 &$  -0.9728   $ & $ -0.7976  $\\\hline
     100 & 5,000 &  $-0.9777$ & $   -0.9673   $ \\\hline
    \end{tabular}
    \caption{{\  Point estimates  for  the $\theta_2=-1$, with  correlated time-series data, and with the different estimators \eqref{eq:RERtheta2-twoscale} and \eqref{eq:PSREtheta-twoscale} for the FM and PSRE. Note that the 
    point estimates by FM \eqref{eq:RERtheta2-twoscale}  has the same form with  FM estimator.}}
    \label{tab:RER_multiple_estimators}
\end{minipage}
\end{table}

\bigskip
{\noindent \bf IV. Non-asymptotic results.} For a 'small' number of samples we  test  the case (a) of i.i.d.  samples with $N=50, 100$, $200$, and $500$  and (b) of multiple i.i.d. trajectories  consisting of correlated time-series data, $N_p=1,  N_t =50,000, $ and $ N_p=100, N_t = 500$. In all results presented next, the number of bootstrap samples is $B= 200$. 

We compare the RE and FM estimates and  confidence intervals for the sets of $N=50$, $N=200$ and $N=500$, see table \ref{tab:FM_RE_nonparametric_theta2}.
For better readability we  report results  only for the parameter $\theta_2$. The complete table is given  the supplementary information.  Moreover,  in the supplementary information,  we report the parameter estimates  and  the asymptotic, jackknife, and bootstrap variance estimates obtained with the FM method for $N=50, 100$, and $200$. 

First, we notice that the jackknife estimate for the RE gives an inconsistent value of variance, thus making the confidence intervals too wide and useless. This is a common issue for the jackknife method with leaving one observation out each time, especially for non-smooth estimators. This well-known deficiency can be rectified by using a more general jackknife with leaving $d >1$ observations out, \cite{shao1989general}, but with an extremely large computational cost, as it is common to choose $d = \sqrt{N}$. For example,  with  sample size $N=50$ the jackknife with leaving $d=\sqrt{N}\approx 7$ observations out we need to compute $10^9$ jackknife estimates  which is impracticable even for this toy example. 

\begin{table}[htbp]
    \centering
    \resizebox{\columnwidth}{!}{\begin{tabular}{|c|c|c|c|c|c|c|c|}
    \hline
     N& & Asymptotic with FM & Jackknife with FM & Bootstrap with FM & Asymptotic with RE & Jackknife with RE & Bootstrap with RE\\\hline
     \multirow{3}{*}{50} & $\hat{\theta}_2$ & \multicolumn{3}{c|}{ $ -0.8575 $ } &  \multicolumn{3}{c|}{$  -0.8720 $}\\\cline{2-8}
     &$\hat{\sigma}^2$ & 0.0654 & 0.0968 & 0.0982 & 0.0156 & 0.5296 & 0.0099 \\\cline{2-8}
     & CI & $\begin{bmatrix}    -1.3589 &  -0.3562  \end{bmatrix}$ & $\begin{bmatrix}   -1.4674 &  -0.2477  \end{bmatrix}$  & $ \begin{bmatrix} -1.4717 &  -0.2433 \end{bmatrix}$ & $\begin{bmatrix}   -1.1165 &  -0.6275  \end{bmatrix}$ & $\begin{bmatrix}   -2.2984 &   0.5543  \end{bmatrix}$ & $\begin{bmatrix}    -1.0670 &  -0.6771  \end{bmatrix}$ \\\hline
     \multirow{3}{*}{200} & $\hat{\theta}_2$ & \multicolumn{3}{c|}{ $-0.9702$ } & \multicolumn{3}{c|}{ $-0.9759$ }\\\cline{2-8}
     &$\hat{\sigma}^2$ & 0.0167 & 0.0200 & 0.0197 & 0.0206 & 1.0325 & 0.0081 \\\cline{2-8}
     &CI & $\begin{bmatrix}   -1.2233 &  -0.7172 \end{bmatrix}$ & $\begin{bmatrix}
	-1.2476 &  -0.6929 \end{bmatrix}$ & $\begin{bmatrix} 	
	-1.2452 &  -0.6953  \end{bmatrix}$ & $\begin{bmatrix}
    -1.2574 &  -0.6945  \end{bmatrix}$ & $\begin{bmatrix} 
    -2.9675 &   1.0157 \end{bmatrix}$ & $\begin{bmatrix} 
    -1.1523 &  -0.7996 \end{bmatrix}$\\\hline
     \multirow{3}{*}{500} & $\hat{\theta}_2$ & \multicolumn{3}{c|}{$ -1.0240$} &  \multicolumn{3}{c|}{$ -0.9827$}\\\cline{2-8}
	 &$\hat{\sigma}^2$ & 0.0063 & 0.0072 & 0.0069 & 0.0140 & 2.8518 & 0.0069 \\\cline{2-8}
     &CI & $\begin{bmatrix}   -1.1790 &  -0.8689\end{bmatrix}$ & $\begin{bmatrix}    -1.1900 &  -0.8579 \end{bmatrix}$ & $\begin{bmatrix}    -1.1868 &  -0.8611 \end{bmatrix}$ & $\begin{bmatrix}   -1.2287 &  -0.7505  \end{bmatrix}$ & $\begin{bmatrix}    -4.2963 &   2.3272  \end{bmatrix}$ & $\begin{bmatrix}    -1.1491 &  -0.8200  \end{bmatrix}$ \\\hline
    \end{tabular}}
    \caption{Comparison of the FM and  RE methods and the corresponding asymptotic, jackknife, and bootstrap,  variance $\sigma^2$ and $95\%$ CIs, with i.i.d.\ samples.}
    \label{tab:FM_RE_nonparametric_theta2}
\end{table}
Notice that the jackknife and bootstrap variances are  slightly larger than the asymptotic variance. 
On the other hand, the jackknife and bootstrap variance estimates are very close for all cases. 
We observe also the improvement of the variance as the number of  samples increases in all methods, as is expected. 
Despite the fact that the  {   non-asymptotic and asymptotic variances  here   are comparable, the non-asymptotic approaches} have advantages over standard asymptotic methods.
{\  Indeed, when the variance of the estimator does not admit an analytic formulation or is too complicated to calculate,  the non-parametric methods are the only option.}
However, non-parametric methods have higher computational cost than the  asymptotic due to the need for resampling and computing estimates repeatedly.

We calculate the bootstrap and jackknife estimates for the {\  PSRE}  minimization, applied  for multiple i.i.d.\ trajectories. With $N_p=100$ trajectories, and $N_t=300$ correlated samples in each trajectory, the 95\% standard confidence intervals are reported in table \ref{tab:nonparametric_CI_RER}. We also validate the accuracy of jackknife and bootstrap CI for the RER minimization, by computing confidence intervals with different sets of samples. On average, 94\% Jackknife CIs and 94.8\% Bootstrap CIs include the true values of the parameters $\theta^*$.

\begin{table}[htbp]
    \centering
    \resizebox{0.8\columnwidth}{!}{\begin{tabular}{|c|c|c|c|c|c|}
    \hline
    $N_p$& $N_t$& $\hat{\theta}$ & Asymptotic CI & Jackknife CI     & Bootstrap CI \\\hline
   1 & 50,000 & $\begin{bmatrix}   0.0746 \\   -0.9805 \\   0.0483 \\ -0.0313\\   -0.0255 \end{bmatrix}$  & $\begin{bmatrix} -0.0491 &   0.1983 \\
   -1.1876 &  -0.7733\\
   -0.1979 &   0.2944\\
   -0.1550 &   0.0925\\
   -0.0965 &   0.0455 \end{bmatrix}$ & NA & NA \\\hline
   100 & 300 &$\begin{bmatrix}-0.0807 \\  -0.9358 \\  0.0270 \\   -0.0546 \\  0.0169 \end{bmatrix}$ & NA &
    $\begin{bmatrix}-0.2546 &   0.0932 \\
   -1.1718 &  -0.6997 \\ 
   -0.4219 &   0.4759 \\
   -0.1722 &   0.0630 \\
   -0.1009 &   0.1346 \end{bmatrix}$     &   $\begin{bmatrix} -0.2498  &  0.0884 \\
   -1.1897 &  -0.6818 \\
   -0.3791 &   0.4331 \\
   -0.1928 &   0.0836 \\
   -0.0939 &   0.1276\end{bmatrix}$\\\hline
    \end{tabular}}
    \caption{Jackknife and bootstrap 95\% CI by RER on correlated data with multiple trajectories. NA means that the CI does not apply to that type of data.}
    \label{tab:nonparametric_CI_RER}
\end{table}

\iffalse
\begin{table}[htbp]
    \centering
    \begin{tabular}{|c|c|c|}\hline
    $N_p$ & $N_t$ & FM $\hat{\theta}$ \\\hline
    1 & 50,000 & $\begin{bmatrix}0.0026 &  -0.9734 &  -0.0105 &  -0.0062 &   0.0038 \end{bmatrix}$ \\\hline
    100 & 300 & $\begin{bmatrix} 0.0020 &  -0.9873 &  -0.0034 &  -0.0013 &   0.0004 \end{bmatrix}$ \\\hline 
    \end{tabular}
    \caption{FM estimates}
    \label{tab:my_label}
\end{table}
\fi

Figure~5 %\ref{fig:bootstrap_CI_drift} 
presents the induced bootstrap confidence intervals for the drift as the QoI $a(x,\hat\theta)$, computed by the corresponding quantiles of the set $ \{a(x,\hat\theta_i)\}_{i=1}^B$, following \eqref{eq:CI-QoI}. We observe that for only $N=50$ samples the bootstrap confidence interval captures the 'large sample' ($N=500$ in FM and RE, $N=50000$ in RER) parameter estimates for all values of $x$. Note though, that the CI is wider for  larger absolute values of $x$ which depicts a wider uncertainty in the estimate.

\begin{figure}[htbp]\label{fig:bootstrap_CI_drift}
    \centering
    \includegraphics[width = 0.6\textwidth]{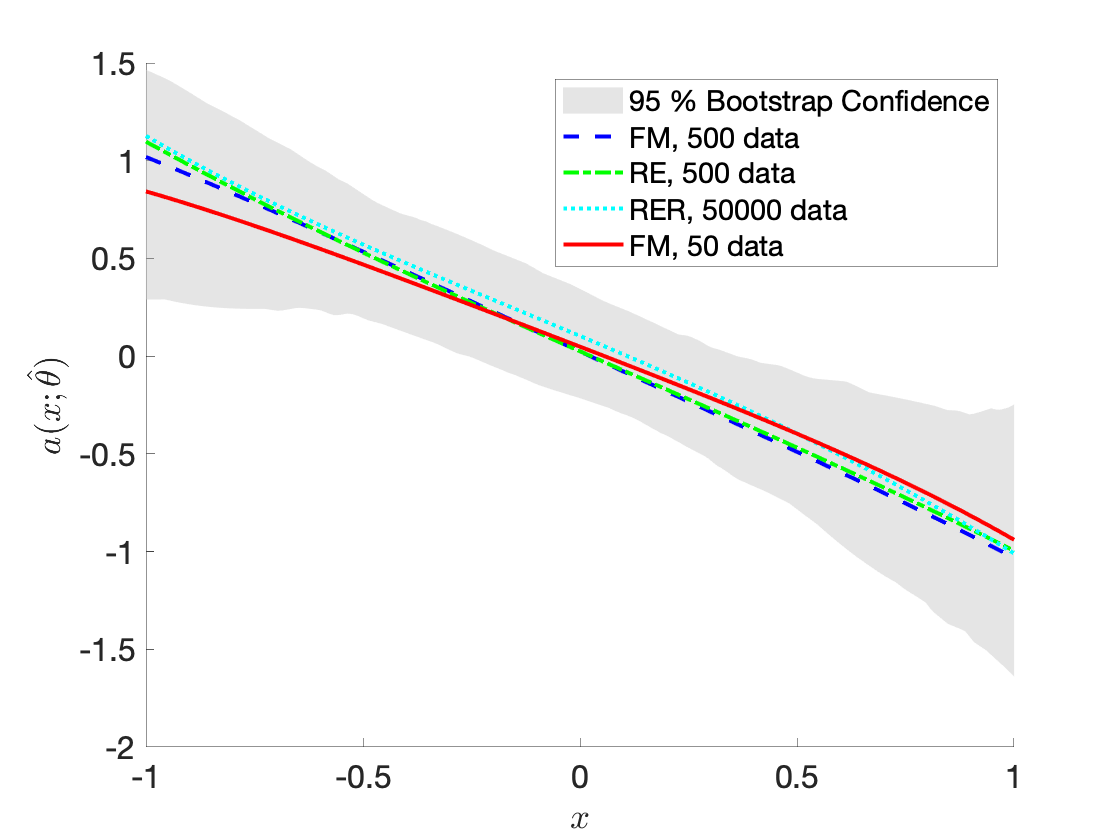}
\caption{Drift function $a(x;\hat{\theta})$ with $\hat{\theta}$ estimated by different methods and 95\% bootstrap percentile confidence interval. The number of bootstrap samples is 200.}
\end{figure}

\bigskip
\noindent {\bf V. Validation of confidence intervals}. In the previews sections  we estimated the asymptotic confidence intervals for i.i.d. data and time-series data. Table \ref{tab:CI_validation} shows the experiment results on validating those confidence intervals. For each method, sample size, and confidence level we calculate the corresponding confidence intervals for $500$ independent sets of synthetic samples generated from \eqref{eq:two-scale}. 
%by different random seeds in MATLAB, 
Then,  we calculate the  percentage of those confidence intervals containing the true value of the parameters. Those probabilities are close to the confidence levels, with RER's probability being slightly smaller than the confidence level. 
\begin{table}[htbp]
    \centering
    \begin{tabular}{|c|c|c|c|c|}
    \hline
    Method   & Sample size & 90\% CI &95\% CI & 99\% CI \\\hline
    FM & 50 &  89.40\% & 93.84\% & 98.28\% \\\hline
    FM & 500 &  90.16\% & 95.68 \%& 98.88\% \\\hline
    FM & 5,000 & 87.56 \%& 92.96 \% & 98.56 \% \\\hline
    RER & 5,000 & 87.76 \%& 92.80 \%& 97.56 \% \\\hline
    RER & 50,000 & 88.00\% & 93.84 \% &  98.76\%\\\hline
    \end{tabular}
    \caption{The percentage of the  estimated confidence intervals that  include the true values of the parameters. The percentage presented is the average  over the corresponding percentages of the parameters.}% $(\theta^*_1,\ldots,\theta^*_5) = (0,-1,0,0,0)$.}
    \label{tab:CI_validation}
\end{table}

%%%%%%%%%%%%%%%%%%%%%%%%%%%%%%%%%%%%%%%%%%%%%%%%%%%%%%%%%%%%%%%%%%%%%%%%%%%%%%%%%%%%%
% Polyethylene application
%%%%%%%%%%%%%%%%%%%%%%%%%%%%%%%%%%%%%%%%%%%%%%%%%%%%%%%%%%%%%%%%%%%%%%%%%%%%%%%%%%%%%

\section{Test-bed 2: Effective force-fields and confidence in coarse-graining of linear polymer chains}\label{sec:polyethylene}
%We describe  the application of the methodology described in section~\ref{sec:CI} involving 
{\  In the present section, we apply the methodology described in section~\ref{sec:CI}  on }
the CG approximation of a polyethylene bulk system.
Specifically, we derive effective force fields with the FM method and the corresponding confidence intervals. Our focus is to understand  the behavior of the output model when the available data is limited. Thus, we concentrate  on  the non-asymptotic methods of section~\ref{sec:Nonasymptotic_results}.
  
To generate the simulated data sets $\Dd_N$ of i.i.d.\ configurations, Molecular Dynamics (MD) simulations of a united-atom polyethylne (PE) system were performed using home made (parallel) MD code.
 %the  GROMACS package  \cite{GROMACS1995, GROMACS2008}.
A PE chain is represented via a united-atom model in which each methylene $CH_2$ and methyl $CH_3$ group  are considered as a single Van der Waals interacting site. Details of the model parameters are given in the supplementary information.  
The model system consists of 96 polyethylene chains of 99 monomer units ($-\textrm{CH}_2-$), i.e., the number of atomistic degrees of freedom is $n=9504$.  
The simulations were performed under NVT conditions at  temperature  $\text{T} = 450 \kelvin $, and density $ \rho  =  0.76868\textrm{ gr/cm}^3$.
%in a box of volume $\text{V} = 287.434\nm^3$, 
The integration time step was $2\fs$. 
We record system configurations every $50\ps$ for about $500 \ns = 500,000 \ps$. %after discarding the first $50\ns$ to ensure the system has equilibrated.
Thus, the size of the available data set is $ 10,000$ (number of configurations).
In the following, results are reported for a large $N=2,000$ and  smaller ($N=200, 100, 30$) data sets, to examine the dependence of the predictions on the size of the actual data set. Note that we choose $N_k$ configurations from the  $10,000$ to be equidistant , e.g.,  the configurations for the  set $N=2,000$ have distance $    (10,000/2,000)* 50 \ps = 250\ps$, and for  $N=200$  the distance is  $  (10,000/200)* 50 \ps = 2,500\ps$.
We should also note that the maximum relaxation time of the polymer chains is around % $1.7\ns = 
$1,700\ps$, calculated  by fitting the end-to-end vector autocorrelation function using a stretched exponential.
This suggests that the large set ($N=2,000$) is composed of correlated data while the smaller ones ($N=200, 100, 30$) are uncorrelated. 

For the coarse-grained representation of the PE, we consider a $3:1$ mapping representation, i.e., three monomer units form one CG particle,  see figure~\ref{fig:PEsnap}. 
Thus, the total number of CG particles in the system is $m = 33 \times 96 = 3168$.

\begin{figure}[htbp]
    %\centering
    \includegraphics[width=0.5\textwidth,height=0.3\textheight]{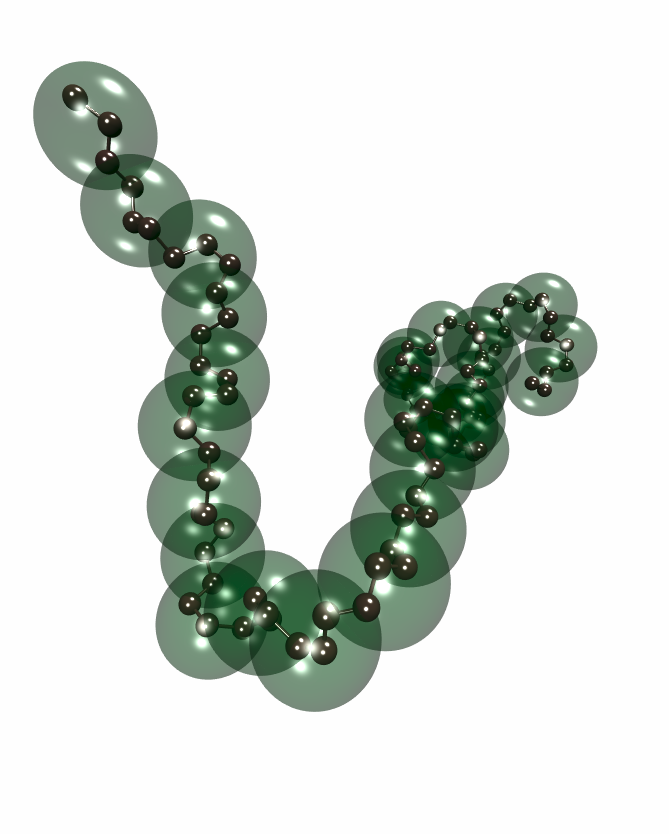}\includegraphics[width=0.5\textwidth,height=0.3\textheight]{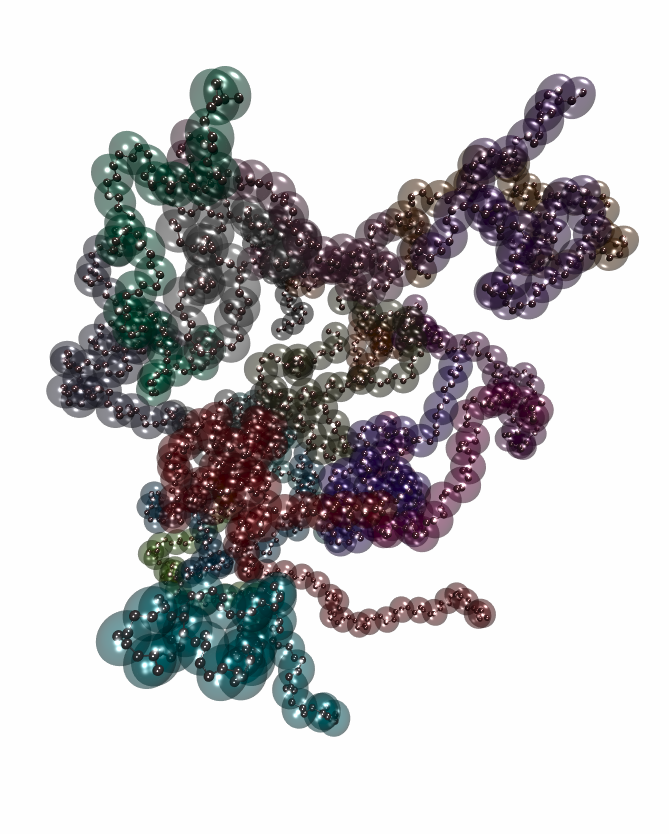}
    \caption{Snapshot of model polyethylene bulk system, shown in atomistic and CG (3:1 mapping scheme) description: a) single PE chain, b) PE bulk system.}
    \label{fig:PEsnap}
\end{figure} 
The CG PE model exhibits both bonded and non-bonded interactions. 
First, we estimate all the CG interactions  with the Iterative Inverse Boltzmann (IBI)\cite{Reith2003}  method, both non-bonded and   bonded interactions (bonds, angles, dihedrals) presented in a tabulated form. We  disregard the non-bonded estimates and keep only the bonded ones which are input for the FM method applied next. The resulting  interaction potentials are reported in the supplementary information.
Then, we estimate the non-bonded interactions with the FM method. We  represent the  non-bonded interactions   with a two-body pair potential, which only depends on the distance between monomers.
That is, the proposed CG potential described in Eq. \eqref{CGpotential} is  
\begin{equation*}
 \BARV(\bbx;\theta) = \sum_{I=1}^m \sum_{J\neq I }^m   u(r;\theta)\COMMA
 \end{equation*}
where $r = |\bar q_I -\bar q_J|$.
The CG pair interaction potential $u(r;\theta)$ is approximated via a functional basis of the form:
\begin{equation*}
u(r;\theta) = \sum_{k=1}^K\theta_k \phi_{k}(r)\COMMA
\end{equation*} 
using the linear  
or the cubic B-splines $ \{\phi_k(r)   \}_{k=1}^K$,
 ~\cite{numer_rec}.
The cutoff range for the non-bonded interactions is $1.4\nm$.
The size of the parameter set $\theta$ is determined by the number of knots $K$.  We present results for a varying number of (a) parameters $K$ and (b) all-atom configurations $N$, given in table~\ref{tab:PEnumbers}. 

The QoI is the pair potential ${u}(r; \theta), \ r>0 $,  
with estimator   the random variable in $\R$ 
\begin{equation}
     \hat{u}_N(r; \theta) = \sum_{k=1}^K \hat{\theta}_{k,N} \phi_{k}(r)\COMMA
\end{equation}
that is a linear combination of the parameter estimators $\hat\theta_N$. 
 
For each combination of parameters and data set size, we find the optimal force field with the FM method. For the small data sets, we also derive the bootstrap and jackknife statistics for the parameters.
Next, we report the results for the cubic B-splines representation with $K =30$ parameters. 
The results for the linear B-splines and comparisons between the different functional basis are reported in the supplementary information. 
In the results reported below, the number of bootstrap samples is $B=200$.
We consider the large data set ($N=2,000$) estimation as a reliable approximation, and thus we use it as {\it reference} result to compare to estimations with the small data sets.

Figure~\ref{fig:1.CIboot_RSTD}~(a) depicts the parameters $\hat\theta_N$ estimated with the large ($N=2,000$) and the small ($N=200$) data set, along with the $95\%$ bootstrap  percentile CI.  
In figure~\ref{fig:1.CIboot_RSTD}~(b) we report the relative standard deviation (RSTD) of each nonzero parameter, defined by $\textrm{RSTD} =  \hat\sigma_k / \hat\theta_k  $, where $\hat\sigma_k$ is the standard deviation estimated with the bootstrap method. The RSTD reveals the most uncertain parameters, e.g., the spline parameters with index $12,\ 20,\ 23,\  26,\  29$. We compare the asymptotic and bootstrap standard deviation $\hat{\sigma}$  in  figure \ref{fig:1.1_AsymptoticVar} for the  $N=30$   and   $N=200$  configurations. It is evident that the asymptotic and bootstrap variance differ for the $N=30$  but are very close when $N=200$. 

\begin{figure}[!htbp]
 \includegraphics[width=0.48\textwidth,height=0.55\textheight,keepaspectratio]{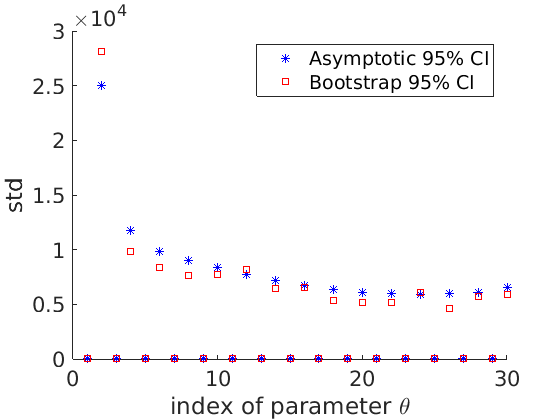}
  \includegraphics[width=0.48\textwidth,height=0.55\textheight,keepaspectratio]{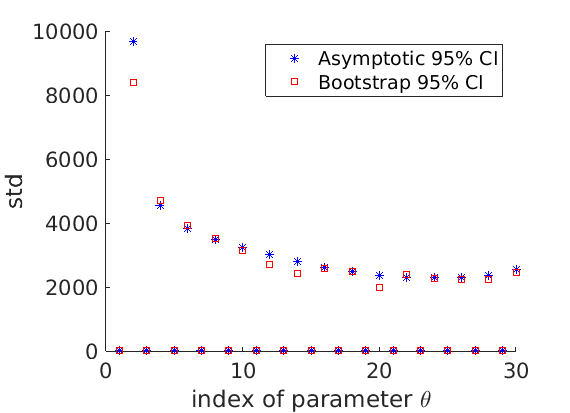}
\caption{Asymptotic and bootstrap estimated standard deviation for  $N=30$ (left figure)  and $N=200$ (right figure) configurations. }\label{fig:1.1_AsymptoticVar}
\end{figure}

Figure~\ref{fig:2.PairPotCIs} depicts the small data set, $N=200$, estimated pair potential, as well as the $80\%,\ 95\%$ and $99\%$ bootstrap percentile confidence intervals. 
We observe that the estimated potential captures the  minimum value point of the reference, though there is an amplitude deviation. Most importantly, the reference potential falls inside the  $95\%$ and $99\%$ bootstrap  CI for the whole range of distances. This observation suggests that for $N=200$, the bootstrap CI is capable of providing useful information for  the range and minima of the potential.

In order to examine the dependence of the CG potential on the size of the data set,  we present in  Figure~\ref{fig:3.PairSmallN}   the resulting effective potential of CG PE beads, as well as its $95\%$ CI, analyzing an increasing number of atomistic configurations. 
We observe that the CI for $N=30$ is practically uninformative, as its range is too wide.

\begin{table}[!htbp]
     \centering
     \begin{tabular}{|c|c|c|c|}
     \hline
          & Number of parameters & Small data set size & Large data set size \\
          \hline
          \hline
        Linear  &  $  \begin{matrix} 75 \\ 30  \end{matrix}$ & $  \begin{matrix} 300 \\ 200  \end{matrix}$  & $ \begin{matrix} 5000 \\ 2000  \end{matrix} $\\
        \hline
        Cubic   & 30 & $ \begin{matrix} 30 \\100\\ 200  \end{matrix}$ & 2000 \\
        \hline
     \end{tabular}
     \caption{Available sample sets for the PE model.}
     \label{tab:PEnumbers}
 \end{table}

\begin{figure}[!htbp]
 \includegraphics[width=0.48\textwidth,height=0.55\textheight,keepaspectratio]{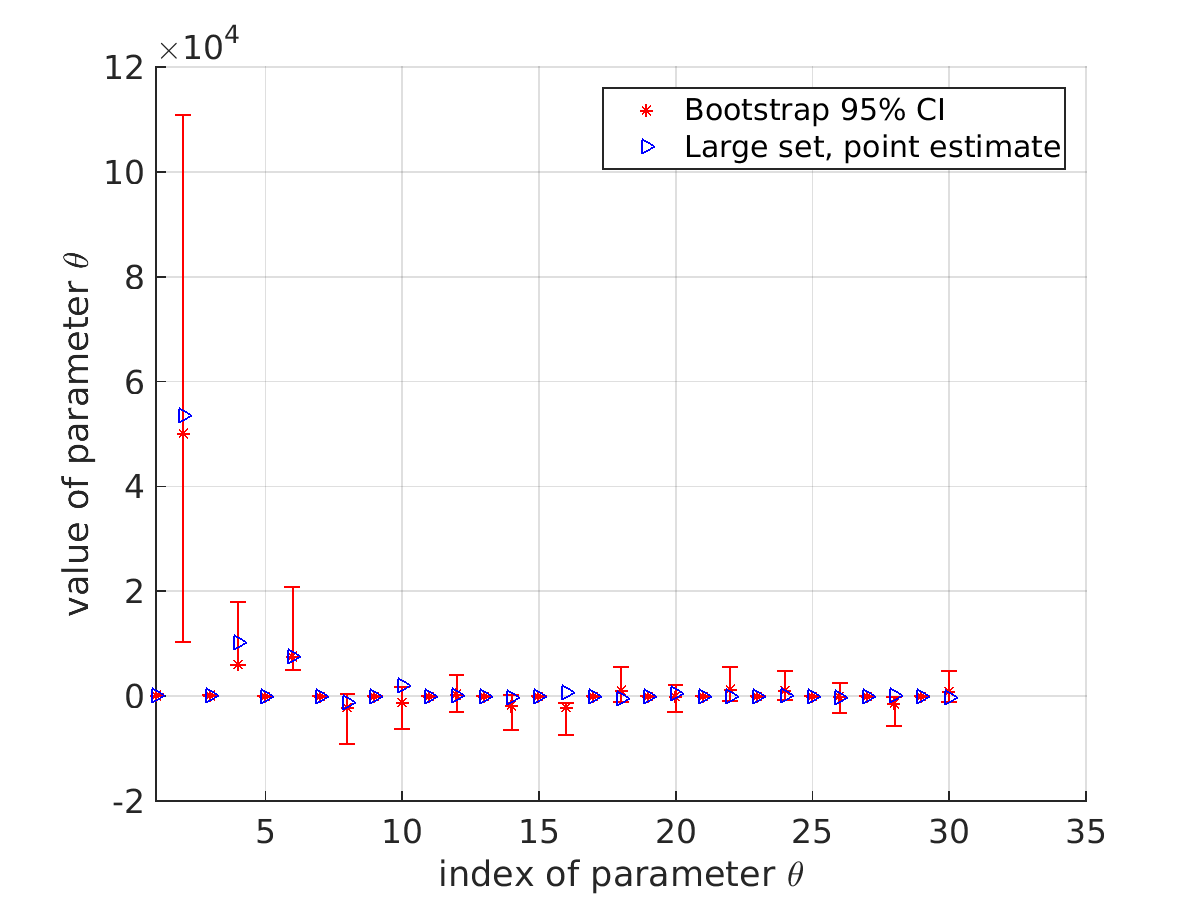}
  \includegraphics[width=0.48\textwidth,height=0.55\textheight,keepaspectratio]{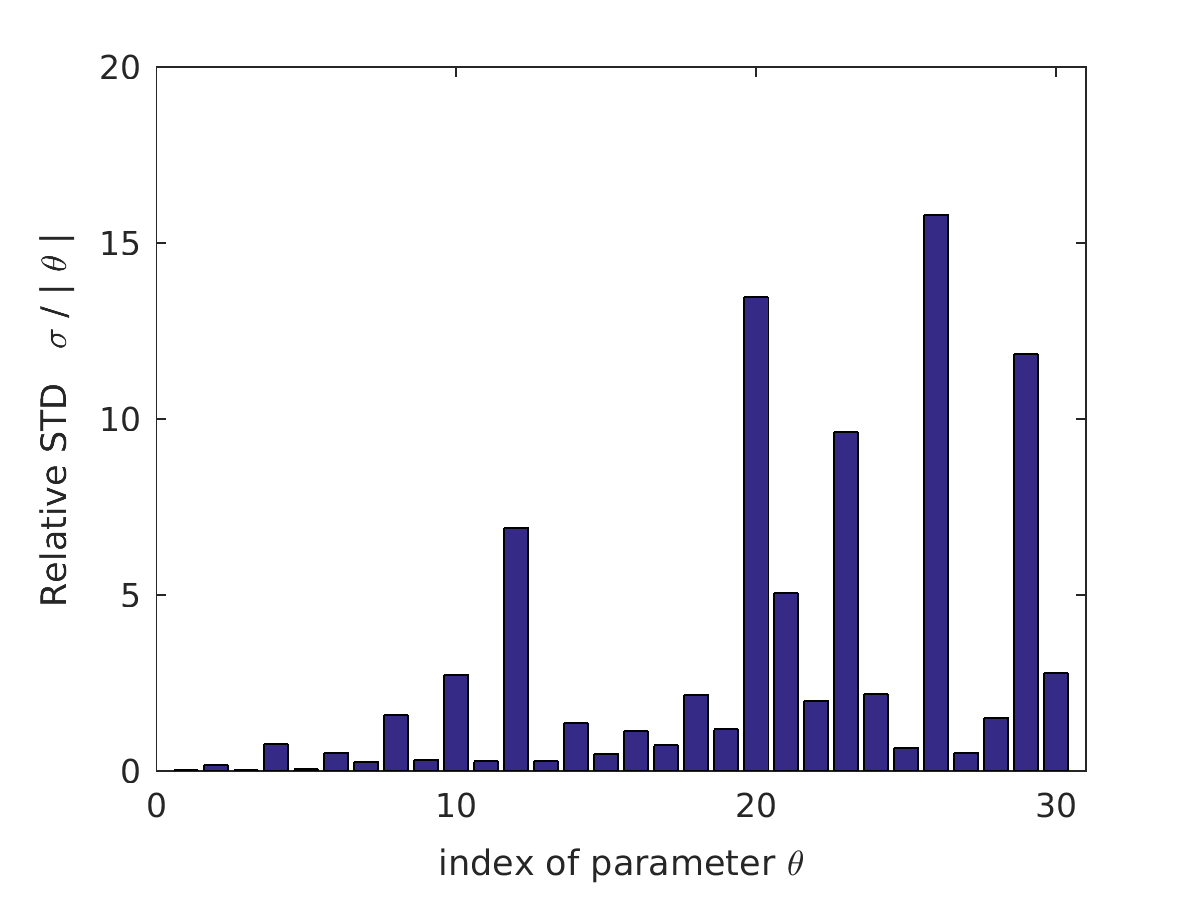}
\caption{ 
CG PE effective potential: (a)  Small sample   ($N=200$) parameter estimate, and bootstrap $ 95\% $ CI.  (b) Relative standard deviation of the parameters, reveals the most uncertain parameters.}
\label{fig:1.CIboot_RSTD}
\end{figure}

\begin{figure}[htbp]
  \includegraphics[width=0.48\textwidth,height=0.5\textheight,keepaspectratio]{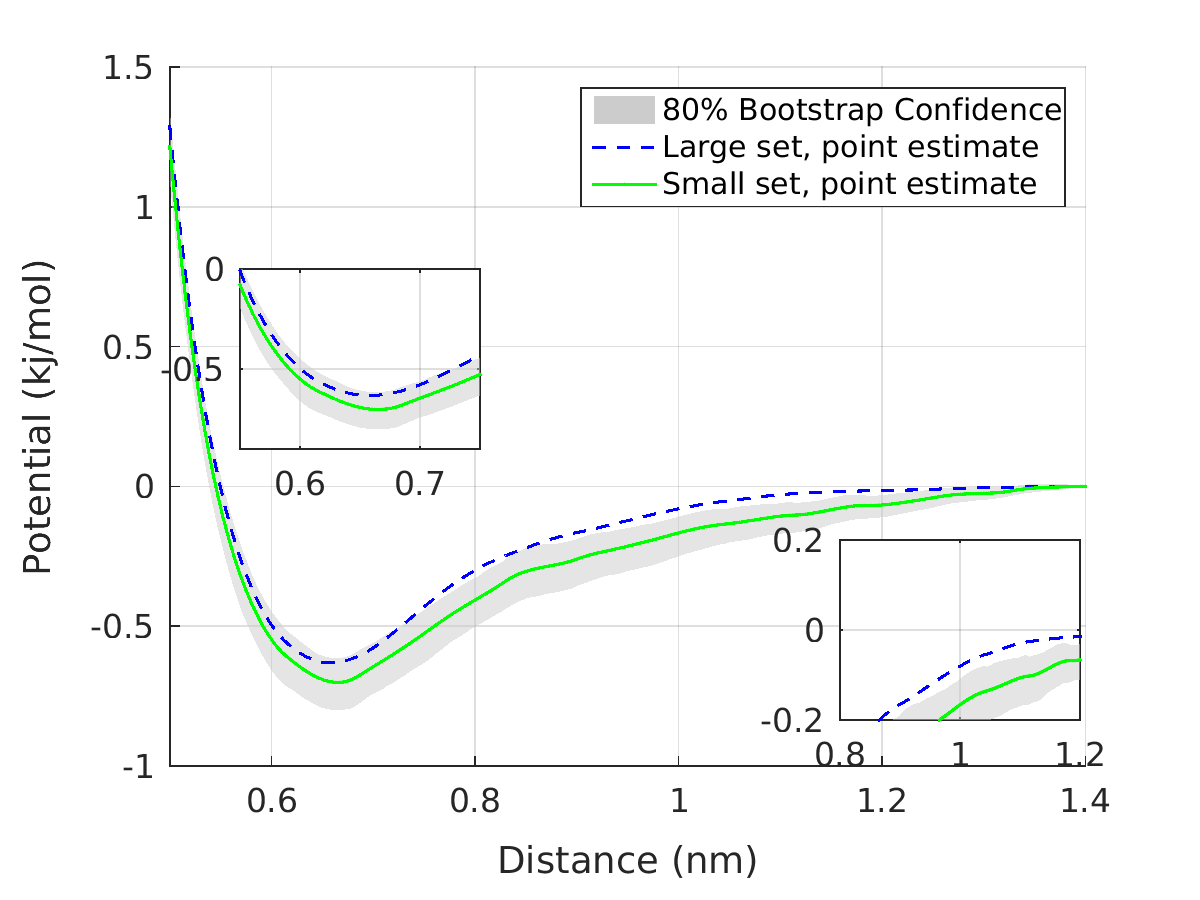}
   \includegraphics[width=0.48\textwidth,height=0.5\textheight,keepaspectratio]{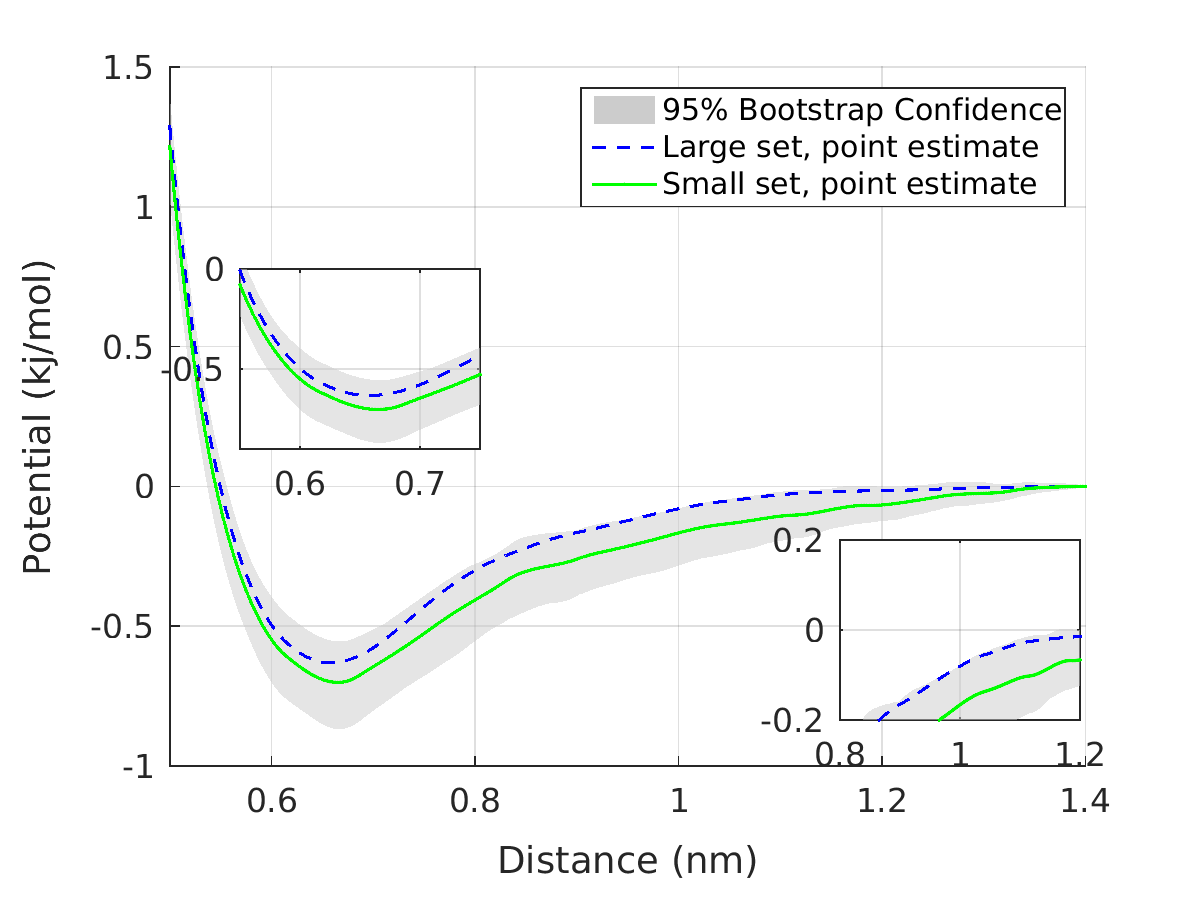} 
   \includegraphics[width=0.48\textwidth,height=0.5\textheight,keepaspectratio]{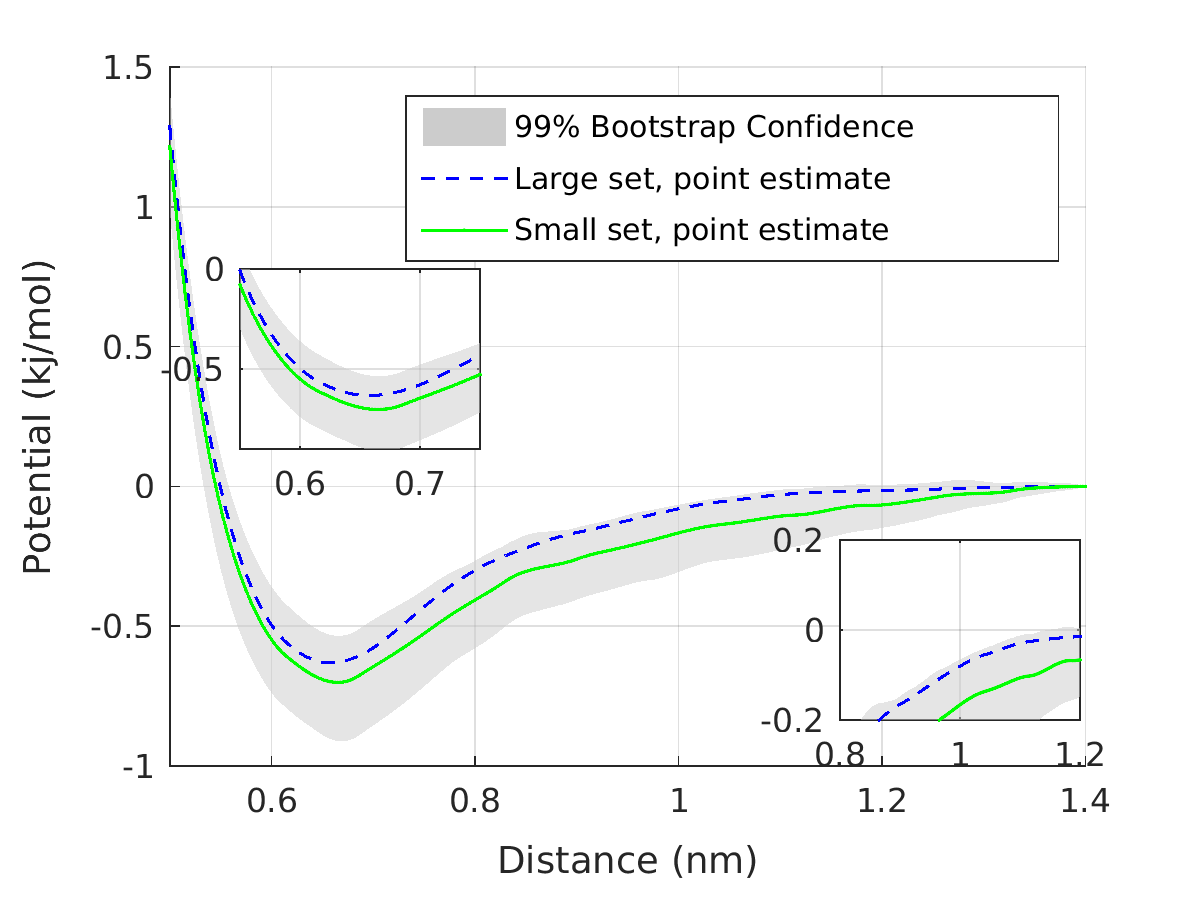}
\caption{Estimated CG PE pair interaction potential $\hat u(r)$  for a large  (2000) and a small (200) data set,  and the  $ 80 \%$,  $ 95 \%$, and $ 99 \%$ bootstrap percentile CI. }
\label{fig:2.PairPotCIs}
  \end{figure}

Next, we examine the bootstrap standard deviation (STD) of the CG pair potential values, as a function of the CG beads distance. 
Results for the bootstrap STD are shown in Figure~\ref{fig:4.STDsmallN}, evaluated for varying number of configurations. Two useful observations can be made out of these data.
First, it is interesting that the STD decreases with increasing the potential interaction range, i.e., the distance between CG particles, for all cases. Indeed, the most uncertain values of the CG pair potential are for small distances.
%{\  
This is not surprising if we consider that at larger distances the configurations are more 'homogeneous' (pair distribution function approaches one), and thus the variance is expected to be smaller.%}
Second, the STD decreases with increasing the number of configurations, and the deviation between them is lower as the data set increases. 
Thus, given a desired accuracy, the STD can  serve as a criterion for choosing a sufficient number of configurations. 
 
The $95\%$ jackknife CI for the CG PE pair potential is presented in figure~\ref{fig:5.Jackknife95CI}, for $N=200$. 
It is clear that the jackknife CI can also capture the reference potential for this size of the data set.

Furthermore, to examine the CG interaction potential predictions at specific {\  particle} distances, the mean, the standard deviation, and the percentile CI values are shown in Figure~\ref{fig:6.PDFur} and in tables~\ref{tab:MeanCI_r065}, \ref{tab:MeanCI_r045} and \ref{tab:MeanCI_r095} {\  for three  distances $ r=0.45, \ 0.65,\ 0.95 $.}
In more detail, Figure~\ref{fig:6.PDFur} and table \ref{tab:MeanCI_r065} depict the estimate and CI for the pair potential at distance $r=0.65\nm$ for various data sets. 
This distance corresponds to the reference potential minimum (see also Figure~\ref{tab:MeanCI_r045}, \ref{fig:2.PairPotCIs}). 
It is clear {\  to see }the change of the {\ probability density}, and the most probable CG potential value, with the increase of the data set size. 
Indeed, as the size of the available configurations change form 30 up to 200 a 'concentration of the density' is also observed. At the same time, the expected (average) value approaches the one of the underlying reference system ($N=2,000$), shown in table~\ref{tab:MeanCI_r065}.

Qualitatively similar are the results for the other two distances $r=0.45\nm$, which is in the repulsive part of the potential, and $r=0.95\nm$ that is in the attractive 'tail', shown in tables \ref{tab:MeanCI_r045} and \ref{tab:MeanCI_r095} respectively. For both distances, the bootstrap predictions become more accurate (CIs are reduces) as the size of the data set increases. For $N=200$ the bootstrap an jackknife predictions are very similar.

\begin{figure}[htbp]
 \centering
 \includegraphics[width=0.7\textwidth,height=0.7\textheight,keepaspectratio]{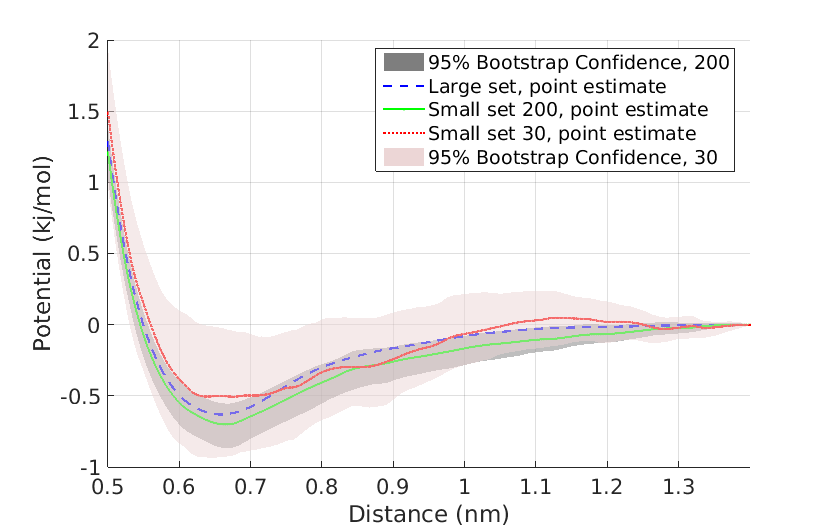}
\caption{$95\%$ bootstrap CI of the CG PE pair effective potentials, for two small data sets, $N=200$ and $30$ configurations.}
\label{fig:3.PairSmallN}
\end{figure}

\begin{figure}[htbp]
 \centering
 \includegraphics[width=0.7\textwidth,height=0.7\textheight,keepaspectratio]{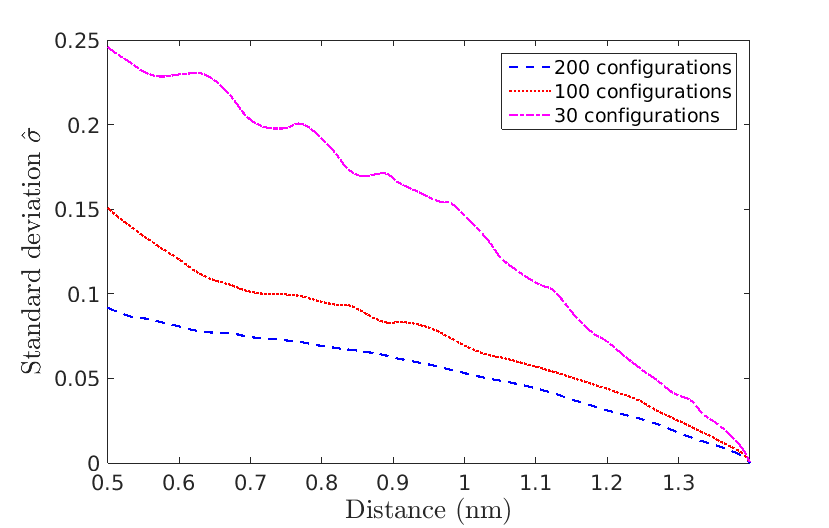}
\caption{Bootstrap standard deviation of the CG PE pair effective potentials for the  data sets of 200, 100, and 30 configurations.}
\label{fig:4.STDsmallN}
  \end{figure}  
  \begin{figure}
  \centering

 \includegraphics[width=0.7\textwidth,height=0.7\textheight,keepaspectratio]{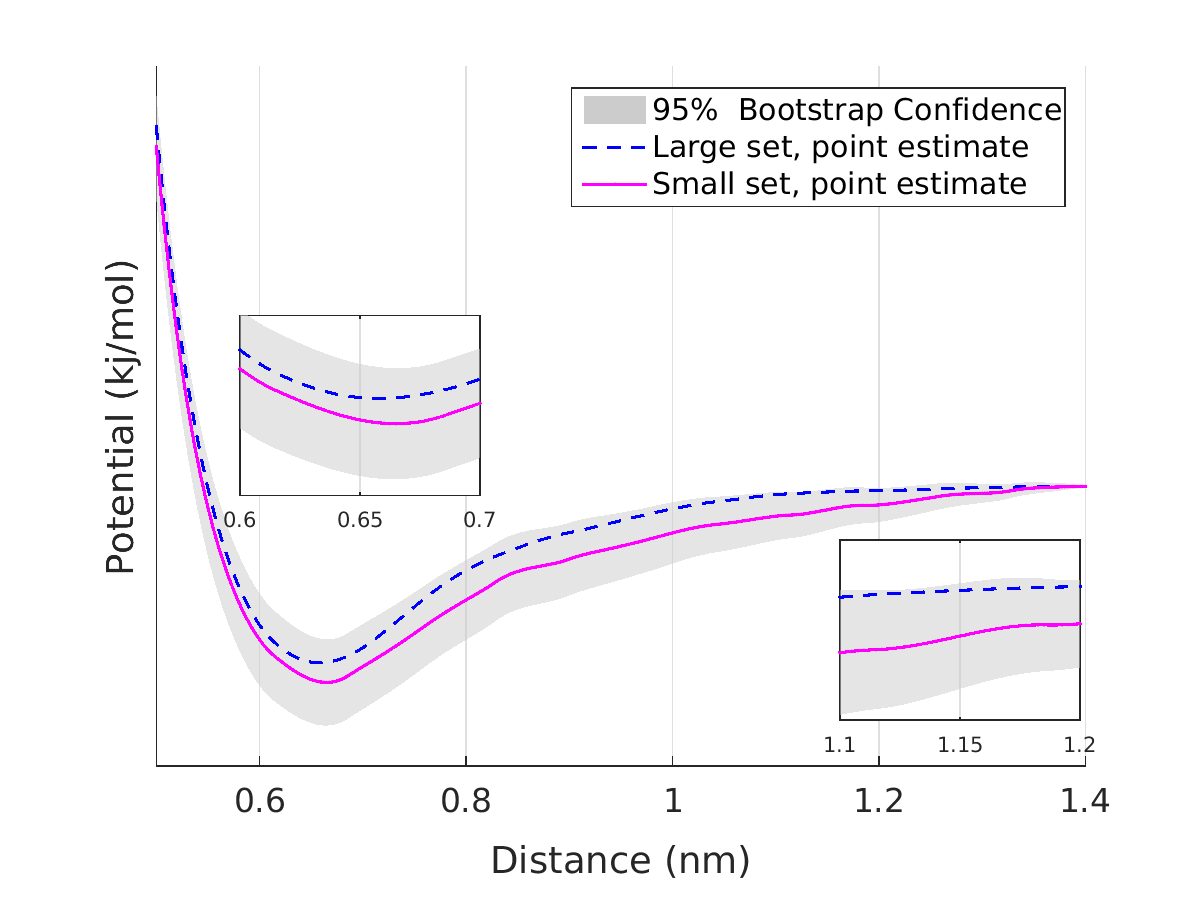}
\caption{The jackknife  95$\%$ CI for the estimated CG PE pair effective interaction potential $u(r)$, for the $N= 200$ configurations data set. }
\label{fig:5.Jackknife95CI}
\end{figure}

\begin{figure}[htbp]
\centering
 \includegraphics[width=0.7\textwidth,height=0.7\textheight,keepaspectratio]{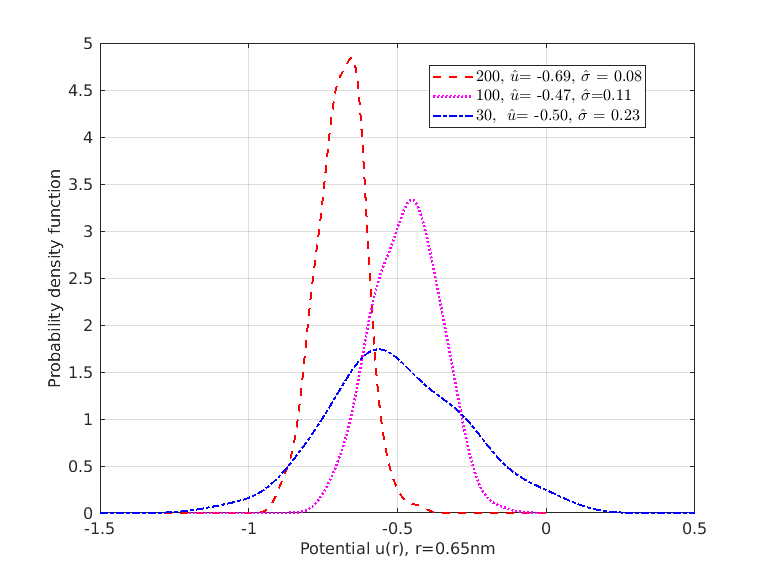}
\caption{Probability density function of CG PE pair effective potential $u(r)$, at $ r=0.65 nm$, derived from bootstrap, for three small data sets involving 200, 100, and 30 atomistic configurations. The corresponding reference value  is  $\hat u = -0.6289$.}
\label{fig:6.PDFur}
\end{figure}

\begin{table}[ht]
      \centering
      \begin{tabular}{|c|c|c|c|c|}
      \hline
            Method &  $\hat{u}$ & $ \hat{\sigma}_u$ & CI &Number of samples\\
            \hline
            \hline
          Large data set & $-0.6289$ & & & $2000$\\
          \hline
          Bootstrap  &
          $\begin{matrix} -0.5027  \\ -0.4706 \\ -0.6900 \end{matrix} $ &
          $\begin{matrix} 0.2259 \\ 0.1079 \\ 0.0770\end{matrix} $ &  
          $\begin{matrix} (-0.9343,\ 	-0.0145)  \\(-0.6847, \ 	-0.2796)   \\ (-0.8504, \ -0.5417 )\end{matrix} $ &  
          $\begin{matrix} 30 \\ 100 \\200 \end{matrix} $\\
          \hline
          Jackknife & $ -0.6900$  & $0.0794$   & $(-0.8463,\ -0.5348)$& $200$\\
           \hline 
      \end{tabular}
      \caption{Mean, standard deviation, and percentile CI for $u(r)$, $ r=0.65$}
      \label{tab:MeanCI_r065}
\end{table}
 
\begin{table}[ht]
      \centering
      \begin{tabular}{|c|c|c|c|c|}
      \hline
            Method &  $\hat{u}$ & $ \hat{\sigma}_u$ & CI &Number of samples \\
            \hline
            \hline
          Large data set & 4.3263  & & &2000\\
          \hline
          Bootstrap  &
          $\begin{matrix}  4.5663 \\ 4.3743 \\4.2631  \end{matrix} $ & 
          $\begin{matrix}  0.2896\\0.1716  \\ 0.1036\end{matrix} $ & 
          $\begin{matrix}  (3.9611, \ 	5.1605) \\ (4.0619,\  	4.6749)  \\ 	(4.0542,\ 	4.4686)  \end{matrix} $ &
          $\begin{matrix} 30 \\ 100 \\200 \end{matrix} $ \\
          \hline
          Jackknife & 	4.2631 & 0.1078	 & $(4.0522, \ 4.4748)$ &200\\
           \hline 
      \end{tabular}
      \caption{Mean, standard deviation, and percentile CI for $u(r)$, r=0.45}
      \label{tab:MeanCI_r045}
\end{table}

 \begin{table}[ht]
       \centering
       \begin{tabular}{|c|c|c|c|c|}
       \hline
             Method &  $\hat{u}$ & $ \hat{\sigma}_u$ & CI &Number of samples \\
             \hline
             \hline
           Large data set & $-0.1210$   & & &$2000$\\
           \hline
           Bootstrap  &
           $\begin{matrix}  -0.1488   \\ 0.0687 \\-0.2216  \end{matrix} $ & 
           $\begin{matrix} 0.1571  \\0.0799  \\  0.0582\end{matrix} $ & 
           $\begin{matrix}  (-0.3974, \  	0.1313) \\ (-0.0583,\ 	0.2453)  \\ 	(-0.3312,\  	-0.1156)  \end{matrix} $ &
           $\begin{matrix} 30 \\ 100 \\200 \end{matrix} $ \\
           \hline
          Jackknife & $-0.2138 $	  & $0.0613$	 & $(-0.3339, \ -0.0937)$ &200\\
           \hline 
       \end{tabular}
       \caption{Mean, standard deviation, and percentile CI for $u(r)$, r=0.95}
       \label{tab:MeanCI_r095}
 \end{table}

As a final check, and in order to understand the effect of the autocorrelated  data on the effective model we present in 
figure \ref{fig:Correlated}  the  pair potential  point estimates
obtained by the FM method with (a) a set of   $200$ correlated configurations  with distance $\tau=50\ps$, (b) the reference  large set of $N=2,000$ configurations  with distance $\tau=250\ps$ and  and (c) the set of $N=200$ uncorrelated configurations with distance $\tau=2,500\ps$. Recall,that the estimated relaxation time is $1,700\ps$. 

 \begin{figure}[htbp]
  \centering
 \includegraphics[width=0.7\textwidth,height=0.7\textheight,keepaspectratio]{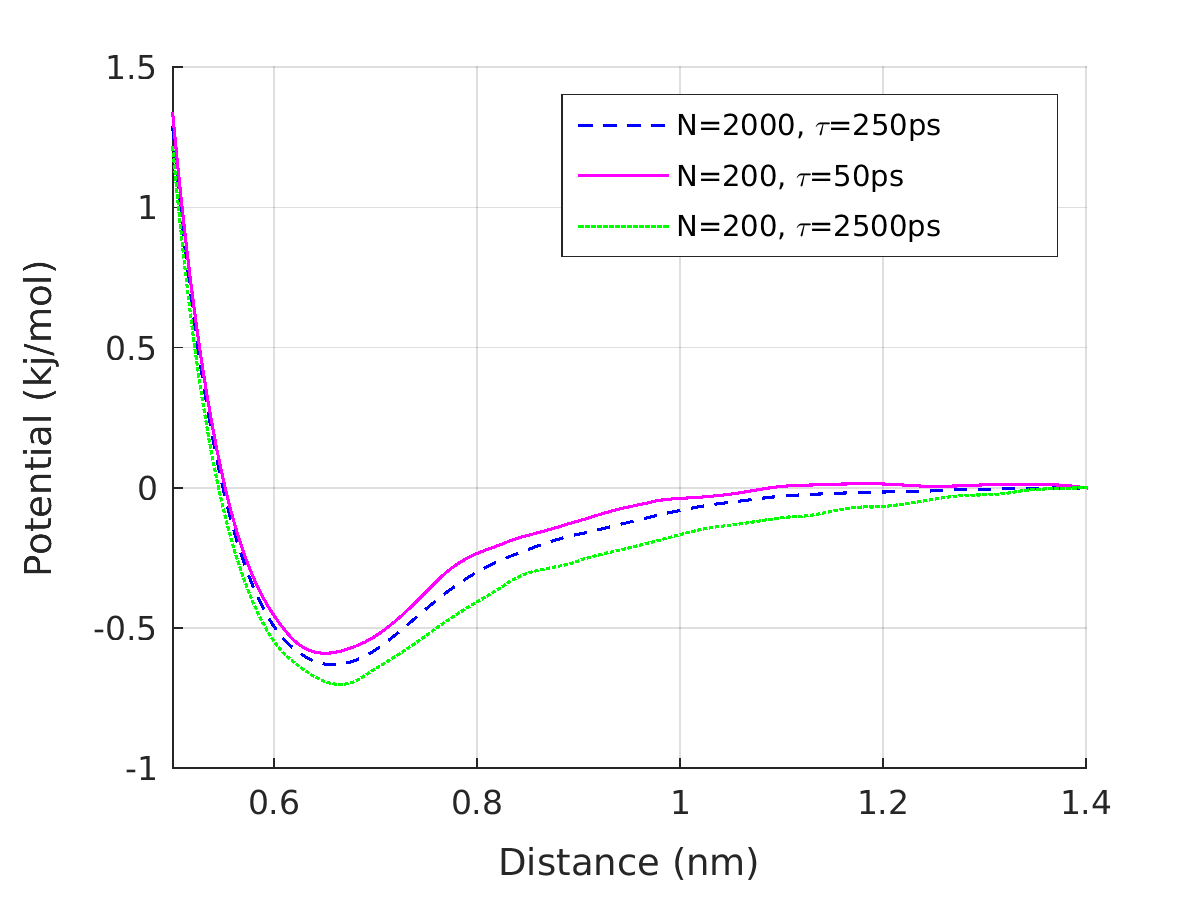}
\caption{FM estimates for correlated and uncorrelated data. }
\label{fig:Correlated}
\end{figure}

\section{Guidelines and Discussion}\label{sec:Conclusion}
To conclude, in this work, we presented an array of   methodologies to generate confidence intervals for systematic bottom-up coarse-grained models, derived from both equilibrium and path-space observations. 
The coarse-graining approach is physics and data driven, relating the true CG model to its digital-twin, the approximate CG model.

We have employed rigorous statistics theory tools for constructing asymptotic and non-asymptotic CIs, and examined their applicability to coarse-graining strategies. 
 We present a  schematic guideline in figure \ref{fig:work_flow2} for the methodology we propose.   
 The main features of the methodology, as depicted in the schematic guideline and observed in the test-bed problems, are: 
\begin{itemize}
    \item Asymptotic vs. non-asymptotic: The asymptotic methods need a parametric form of the variance since we compute the expectation of the first and second derivatives of the log density or transition probability function.
    While the non-asymptotic methods do not need a parametric form of  the variance, they have an additional computation cost due to the repeated optimization to compute sample estimates.  {\  Therefore, }if an analytic form of variance can be derived, asymptotic methods are more computationally efficient.
    \item  Time-series data vs. independent data: Independent data can provide {\  more information as  their  statistical analysis is well established, but obtaining independent data in real-world problem is often impractical}. Correlated data, such as time-series data, are more commonly used. Our proposed confidence intervals for the RER minimization, provides a useful uncertainty quantification of the estimated parameters for time-series data, under the assumption of stationary and ergodicity.
    \item Correlated data in multiple independent trajectories: we also demonstrated in table \ref{tab:nonparametric_CI_RER} that by using the independence between trajectories  a  resampling technique, jackknife and bootstrap,   can also construct non-asymptotic confidence intervals for this type of data.
\end{itemize}

\begin{figure}[htbp]
\centering
\resizebox{0.8\textwidth}{!}{\begin{tikzpicture}[align=center, node distance=2cm, auto]
\node (data1) [io] {iid data \\$\Dd_N$};
\node (data2) [io, right of=data1, node distance=6cm] {time-series data\\ $\Dd_{N_p,N_t}$};
\node (m1) [ below of=data1] {};
\node (m2) [ below of=data2] {} ;
\node (model) [startstop, right of=m1, node distance=3cm] {CG model: $\hat{\theta}$};
\node (size) [decision, below of=m1, node distance=3cm] {Is N big and is Fisher information computable?};
\node (traj) [decision, below of=m2, node distance=3cm] {Does data have multiple trajectories?};
\node (p1) [process, below of=size, node distance=3cm, xshift = -3cm] {Asymptotic CI};
\node (p2) [process, below of=size, node distance=3cm, xshift = 3cm] {Nonasymptotic CI \\(Jackknife, Bootstrap)};
\node (p3) [process, below of=traj, node distance=3cm, xshift = 3cm] {Asymptotic CI};
\draw [arrow] (data1) -- (size);
\draw [arrow] (data2) -- (traj);
\draw [arrow] (model) -- (m1); 
\draw [arrow] (model) -- (m2); 
\draw [arrow] (size) -| node[anchor=west, yshift=-1cm] {yes} (p1);
\draw [arrow] (size) -| node[anchor=east, yshift=-1cm] {no} (p2);
\draw [arrow] (traj) -| node[anchor=west, yshift=-1.03cm] {yes} (p2);
\draw [arrow] (traj) -| node[anchor=east, yshift=-1cm] {no} (p3);
\end{tikzpicture}}
\caption{Schematic methodology for confidence interval estimation in coarse-graining}
\label{fig:work_flow2}
\end{figure}
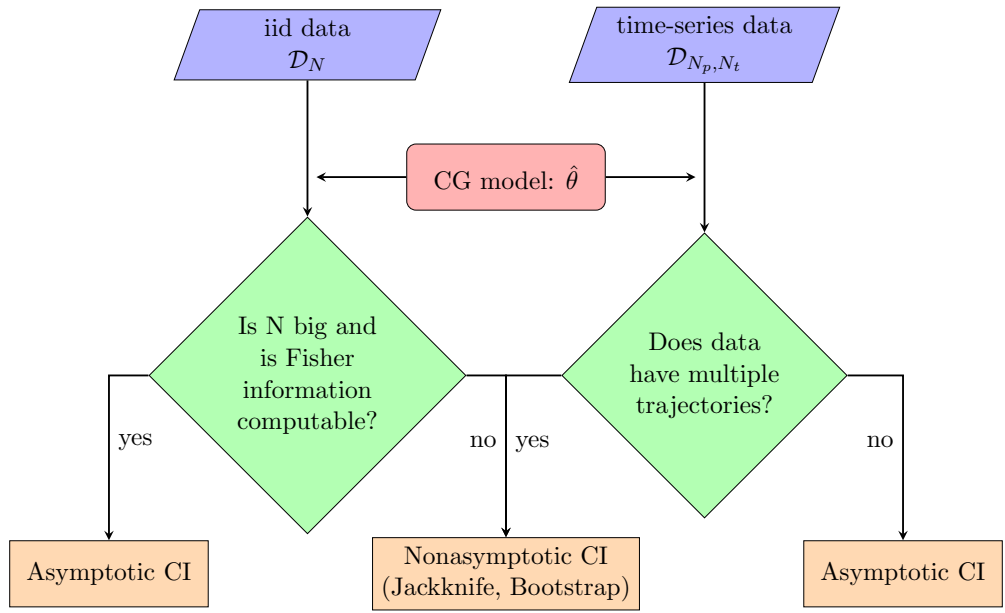

In short, we have demonstrated:
    \begin{itemize}
    \item  the need for employing  non-asymptotic methods in coarse-graining high dimensional molecular systems, and
    \item the benefit of applying time series, path-space techniques.
    \end{itemize}
As it is often extremely time-consuming to generate 'large' data sets of atomistic model configurations in molecular, and  especially macromolecular systems, the asymptotic confidence intervals are often not valid. Therefore, we propose non-parametric, non-asymptotic   methods, i.e.,  bootstrap    and jackknife methods  to provide guaranties of the coarse-grained output model in terms of the size of the available data. 

Moreover, we show with the benchmark example that the path-space   method, i.e., the RER optimization, is best in terms of the cost of generating simulated data, for which we can also provide confidence intervals. {\    Also,  the FM estimator for correlated data gives reliable point estimates though corresponding confidence sets cannot be obtained.}
Indeed, since the bottom-up CG methods are based on simulated data, often for high-dimensional systems, not discarding simulated data to achieve independence saves a large amount of computational time. 

For the polymer melt, at realistic conditions, we have presented the bootstrap and jackknife confidence intervals for the FM estimated parameters and the pair potential.  
A detailed analysis of the CIs for the derived effective CG non-bonded potential suggests that the sufficiency of the data size can be estimated along with the estimated bootstrap variance.

We believe that our work could stimulate further studies on the development and application of rigorous statistical inference methods for coarse-grained models of soft condensed matter, and in particular, of macromolecular systems. This is even more important for hybrid polymer-based complex materials, for which the relaxation times increase rapidly with the complexity of the underlying physico-chemical interactions, thus making the sampling of either a large number of atomistic  i.i.d or long time-correlated  configurations  not feasible~\cite{Petra2019,Tasia2019}.
 
\section*{Acknowledgements}\label{sec:Acknowl}
The research of M.K.  was partially supported by NSF TRIPODS  CISE-1934846 and by the Air Force Office of Scientific Research (AFOSR) under the grant FA-9550-18-1-0214. The research of T. J.   was partially supported by the National Science Foundation (NSF) under the grant DMS-1515712 and by the Air Force Office of Scientific Research (AFOSR) under the grant FA-9550-18-1-0214.
\\
 E.K.  acknowledges   support by the Hellenic Foundation for Research and Innovation
(HFRI) and the General Secretariat for Research and Technology (GSRT), under grant
agreement No [52].

\bibliographystyle{plain}
\bibliography{referenceUQCG}

%\documentclass[11pt,letterpaper]{article}
%\usepackage[latin1]{inputenc}
%\usepackage{amsmath}
%\usepackage{amsfonts}
%\usepackage{amssymb}
%\usepackage{amsthm}
%\usepackage{graphicx}
%\usepackage{epstopdf}
%\usepackage[margin=1in]{geometry}
%\usepackage[hidelinks]{hyperref}
%\usepackage{color}
%\usepackage{todonotes}
%\newcommand{\tdi}{\todo[inline]}
%\usepackage{indentfirst}
%\usepackage{multirow}
%\usepackage{bm}
%\usepackage{subfigure}
%\usepackage{tabularx}
%% 
%%\newtheorem{theorem}{Theorem}[section]
%%\newtheorem{corollary}{Corollary}[section] 
%% 
%\hypersetup{colorlinks,citecolor=blue}
%\input{macros_basic}
%\input{macros_PMF}

%\begin{center}
%{\bf \Large{Supplementary Information: Data-driven Uncertainty Quantification for Systematic Coarse-grained Models}}
%\end{center}

\section*{Supplementary Information: Data-driven Uncertainty Quantification for Systematic Coarse-grained Models}
%\date{}
%\begin{document}
\setcounter{section}{0}
\setcounter{equation}{0}
\maketitle
\newtheorem{theorem}{Theorem}[subsection]
\newtheorem{corollary}{Corollary}[subsection] 
\section{Asymptotic convergence results }
\subsection{Asymptotic theorem for i.i.d.\ data}
Suppose we have N i.i.d.\ fine-scale data
\begin{displaymath}
X_1,X_2,\ldots,X_N,
\end{displaymath}
where $X_i\in \mathcal{M}=\mathbb{R}^D$, $i=1,\ldots,N$. Assume that fine-scale data is distributed with probability density $p(x)$.

The coarse-graining (CG) map $\Pi$ is defined as
\begin{displaymath}
\Pi: \mathcal{M}\rightarrow \mathcal{M}_{CG},
\end{displaymath}
where $\mathcal{M}_{CG}=\mathbb{R}^d$, $d\ll D$.
Note that the CG map $\Pi$ is surjective here, that is, $\Pi^{-1}(\mathcal{M}_{CG}) = \mathcal{M}$. 

Let us  assume that the parametric family of the CG models $\mathcal{Q}^\theta$ has probability density $q^\theta$.
We obtain the optimal CG model by minimizing the relative entropy

\begin{equation}
\label{eq:RE}
\mathcal{R}(p|q^\theta\circ\Pi) := \EXPECT_{p}\left[\log\frac{p(X)}{q^\theta\circ\Pi (X)} \right].
\end{equation}
In addition, we have 
$\mathcal{R}(p|q^\theta\circ\Pi) = \lim_{N\rightarrow \infty} \hat{R}_N(p|q^\theta\circ\Pi)$ where
\begin{equation}
\label{eq:RE N}
\hat{R}_N(p|q^\theta\circ\Pi) := \frac{1}{N}\sum_{i=1}^N \log \frac{p(X_i)}{q^\theta\circ\Pi(X_i)}.
\end{equation}
Thus, the minimization of RE \eqref{eq:RE} is asymptoticly equivalent to the optimization problem

\begin{displaymath}
\min_{\theta} \hat{R}_N(p|q^\theta\circ\Pi) \Leftrightarrow \max_{\theta} \frac{1}{N} \sum_{i=1}^N \log\left( q^\theta\circ\Pi(X_i)\right)\PERIOD
\end{displaymath}
Let's also define 
\begin{equation*}
\label{eq: Likelihood iid}
\ell_N(\theta) := \frac{1}{N} \sum_{i=1}^N \log\left( q^\theta\circ\Pi(X_i)\right)\COMMA
\end{equation*}

\begin{align*}
\theta^* &:= \arg\min_{\theta} \mathcal{R}(p|q^\theta\circ\Pi), \\
\hat{\theta}_N &:= \arg\max_{\theta} \frac{1}{N} \sum_{i=1}^N \log \left(q^\theta\circ\Pi(X_i)\right)\PERIOD
\end{align*}

\begin{corollary}
\label{cor: expectation of gradient}
\begin{displaymath}
\EXPECT_p\big[\nabla_\theta \log q^\theta\circ\Pi(X) |_{\theta=\theta^*} \big] = \vec{0}.
\end{displaymath}
\end{corollary}

\begin{proof}
By the definition of $\theta^*$, the gradient of $\mathcal{R}(p|q^\theta\circ\Pi)$ equals to $\vec{0}$ at $\theta = \theta^*$. Therefore, 
$\\nabla_\theta\mathcal{R}(p|q^\theta\circ\Pi)= \nabla_\theta \EXPECT_p\big[ \log q^\theta\circ\Pi(X) |_{\theta=\theta^*} \big] = 0$. The fact that 
$\EXPECT_p\big[\nabla_\theta \log q^\theta\circ\Pi(X) |_{\theta=\theta^*} \big] = \nabla_\theta \EXPECT_p\big[ \log q^\theta\circ\Pi(X) |_{\theta=\theta^*} \big]$ yields the result.
\end{proof}

Let $\xrightarrow{P}  $ denote convergence in probability. We say $\hat{\theta}_n \xrightarrow{P} \theta^*$ if, for every $\epsilon>0$,
\begin{displaymath}
P(|\hat{\theta}_N - \theta^*| >\epsilon) \rightarrow 0 \quad \text{as}\quad N \rightarrow \infty.
\end{displaymath}

\begin{corollary}{(Consistency of the estimator)}
\label{cor: consistency of estimator}
Suppose that 
$$\sup_{\theta} |\hat{R}_N(p|q^\theta\circ\Pi) - R(p|q^\theta\circ\Pi)| \xrightarrow{P} 0, $$
and that, for every $\epsilon>0$, 
$$\inf_{\theta:|\theta-\theta^*|>\epsilon} R(p|q^\theta\circ\Pi) > R(p|q^{\theta^*}\circ\Pi). $$
Then
\begin{displaymath}
\hat{\theta}_N \xrightarrow{P} \theta^* \quad\textit{as} \quad  N\rightarrow \infty
\end{displaymath}
\end{corollary}

\begin{proof}
Since $\theta^*$ minimizes $R(p|q^{\theta}\circ\Pi)$, so $R(p|q^{\theta^*}\circ\Pi) \leq R(p|q^{\hat{\theta}_N}\circ\Pi)$.
\begin{align*}
    R(p|q^{\hat{\theta}_N}\circ\Pi) - R(p|q^{\theta^*}\circ\Pi) & = R(p|q^{\hat{\theta}_N}\circ\Pi) - \hat{R}_N(p|q^{\theta^*}\circ\Pi) + \hat{R}_N(p|q^{\theta^*}\circ\Pi) - R(p|q^{\theta^*}\circ\Pi) \\
    &\leq R(p|q^{\hat{\theta}_N}\circ\Pi) - \hat{R}_N(p|q^{\hat{\theta}_N}\circ\Pi) + \hat{R}_N(p|q^{\theta^*}\circ\Pi) - R(p|q^{\theta^*}\circ\Pi) \\
    &\leq \sup_\theta |R(p|q^{\hat{\theta}_N}\circ\Pi) - \hat{R}_N(p|q^{\hat{\theta}_N}\circ\Pi)| + \hat{R}_N(p|q^{\theta^*}\circ\Pi) - R(p|q^{\theta^*}\circ\Pi) \\
    &\xrightarrow{P} 0 \COMMA
\end{align*}
where the first inequality follows from $\hat{R}_N(p|q^{\theta^*}\circ\Pi) \geq \hat{R}_N(p|q^{\hat{\theta}_N}\circ\Pi)$, and the last line follows from the first assumption. Hence for any $\delta>0$, we have 
$$P\big(R(p|q^{\hat{\theta}_N}\circ\Pi) > R(p|q^{\theta^*}\circ\Pi) + \delta\big)\rightarrow 0 .$$
By the second assumption, for any $\epsilon>0$, there exists $\delta>0$ such that $|\theta - \theta^*|>\epsilon$ implies $R(p|q^\theta\circ\Pi) > R(p|q^{\theta^*}\circ\Pi) + \delta$, hence 
$$P(|\hat{\theta}_N - \theta^*|>\epsilon) \leq P\big(R(p|q^{\hat{\theta}_N}\circ\Pi) > R(p|q^{\theta^*}\circ\Pi) + \delta\big)\rightarrow 0\COMMA$$
 yields the consistency of the estimator.
\end{proof}

The Fisher information matrices are defined as 
\begin{subequations}
\begin{align}
 {\Ff}_1(\theta^*) &:= -\EXPECT_p\big[\nabla_\theta^2\log q^\theta\circ\Pi(X) |_{\theta=\theta^*}  \big] \label{eq:F1 iid}, \\
 {\Ff}_2(\theta^*) &:= \EXPECT_p\big[(\nabla_\theta \log q^\theta\circ\Pi(X))(\nabla_\theta \log q^\theta\circ\Pi(X))^{tr} |_{\theta=\theta^*} \big] \label{eq:F2_iid}.
\end{align}
\end{subequations}
Here $\cdot ^{tr}$ denotes matrix transpose.

\begin{corollary}
\label{cor: F1equalsF2}
If  $ p = q^{\theta^*}\circ\Pi $, then
\begin{displaymath}
  {\Ff}_1(\theta^*) =  {\Ff}_2(\theta^*).
\end{displaymath}
\end{corollary}

\begin{proof}
$$ -\nabla_\theta^2\log q^\theta\circ\Pi(X) = (\nabla_\theta\log q^\theta\circ\Pi(X))(\nabla_\theta\log q^\theta\circ\Pi(X))^{tr} - \frac{\nabla_\theta^2 q^\theta\circ\Pi(X)}{q^\theta\circ\Pi(X)}. $$
Take expectation with respect to the measure $p$ at $\theta=\theta^*$ on both sides, yields,
$$ {\Ff}_1(\theta^*) =  {\Ff}_2(\theta^*) - \EXPECT_p\left[\frac{\nabla_\theta^2 q^\theta\circ\Pi(X)}{q^\theta\circ\Pi(X)} |_{\theta=\theta^*}  \right]. $$
If $ p = q^{\theta^*}\circ\Pi $, the last term
\begin{align*}
    \EXPECT_p\left[\frac{\nabla_\theta^2 q^\theta\circ\Pi(X)}{q^\theta\circ\Pi(X)} |_{\theta=\theta^*}  \right] & = \int \frac{\nabla_\theta^2 q^{\theta^*}\circ\Pi(X)}{q^{\theta^*}\circ\Pi(X)} q^{\theta^*}\circ\Pi(X) dX \\
    &= \int \nabla_\theta^2 q^{\theta^*}\circ\Pi(X) dX \\
    &= \nabla_\theta^2 \int q^{\theta^*}\circ\Pi(X) dX \\
    &= \nabla_\theta^2  \ 1 \\
    & = 0
\end{align*}
\end{proof}

%Let $\xrightarrow{D}$ denote the convergence in distribution. 
Let $F_N$ denote the cumulative distribution function (CDF) of $\hat{\theta}_N$ and let $F$ denote the CDF of a normal random variable $\mathcal{N}(\mu, \sigma^2)$ with mean $\mu_0$ and variance $\sigma^2$. We say that 
$$\hat{\theta}_N \xrightarrow{D} \mathcal{N}(\mu_0, \sigma^2)$$ if
\begin{displaymath}
\lim_{N\rightarrow\infty} F_N(t) = F(t)
\end{displaymath}
at all $t$ for which $F$ is continuous.

\begin{theorem}
\label{thm:asymptotic iid 2}
\begin{enumerate}
    \item Under certain conditions,
\begin{displaymath}
\sqrt{N I^{-1}(\theta^*)}(\hat{\theta}_N - \theta^*) \xrightarrow{D} \mathcal{N}( {0}, \mathbb{I}),
\end{displaymath}
where
\begin{displaymath}
I(\theta^*) = (\Ff_1(\theta^*))^{-tr}\Ff_2(\theta^*)\Ff_1(\theta^*))^{-1}.
\end{displaymath}
\item
If $\Ff_1$ and $\Ff_2$ is estimated by 
\begin{align*}
\hat{\Ff_1}(\hat{\theta}_N) &= - \frac{1}{N} \sum_{i=1}^N \nabla_\theta^2 \log q^\theta\circ\Pi(X_i)|_{\theta=\hat{\theta}_N}, \\
\hat{\Ff_2}(\hat{\theta}_N) &= \frac{1}{N} \sum_{i=1}^N (\nabla_\theta \log q^\theta\circ\Pi(X_i))(\nabla_\theta \log q^\theta\circ\Pi(X_i))^{tr}|_{\theta=\hat{\theta}_N}.
\end{align*}
Then we have 
\begin{displaymath}
\sqrt{N \hat{I}^{-1}(\hat{\theta}_N)}  (\hat{\theta}_N - \theta^*) \xrightarrow{D} \mathcal{N}(\vec{0}, \mathbb{I}),
\end{displaymath}
where
\begin{displaymath}
\hat{I}(\hat{\theta}_N) = (\hat{\Ff_1}(\hat{\theta}_N))^{-tr}\hat{\Ff_2}(\hat{\theta}_N)(\hat{\Ff_1}(\hat{\theta}_N))^{-1}.
\end{displaymath}
\end{enumerate}

\end{theorem}

\begin{proof}
Let's define a score function 
$$\ell_N(\theta) = \frac{1}{N} \sum_{i=1}^{N} \log q^\theta\circ\Pi (X_i). $$

We notice that $\hat{\theta}_N = \arg\max_\theta \ell_N(\theta)$ by definition. If we assume the gradient of the score function with respect to $\theta$ exists, then the gradient of $\ell_N(\theta)$ at $\hat{\theta}_N$ must be zero.
$$\nabla_\theta\ell_N(\hat{\theta}_N) = \vec{0} .$$

Then use Taylor expansion to expand $\nabla_\theta \ell_N(\hat{\theta_N})$ at $\theta^*$, that is
$$\vec{0} = \nabla_\theta \ell_N(\hat{\theta_N}) = \nabla_\theta \ell_N(\theta^*) + \nabla^2_\theta \ell_N(\bar{\theta})(\hat{\theta}_N - \theta^*), $$
where $\hat{\theta}_N \leq \bar{\theta} \leq \theta^*$. Assume $-\nabla^2_\theta \ell_N(\bar{\theta}) $ is invertible and rearrange the equation to get
$$\sqrt{N}(\hat{\theta}_N - \theta^*) = (-\nabla^2_\theta \ell_N(\bar{\theta}) )^{-1} (\sqrt{N}\nabla_\theta \ell_N(\theta^*) ) . $$

Let $Y_i := \nabla_\theta \log q^\theta\circ\Pi(X_i)|_{\theta = \theta^*}$, then $\sqrt{N}\nabla_\theta \ell_N(\theta^*) = \sqrt{N} \bar{Y}$. $Y_i$'s are i.i.d random variables with mean $\vec{0}$ by Corollary \ref{cor: expectation of gradient} and variance $\Ff_2(\theta^*)$. By the Central Limit Theorem, we have the convergence in distribution 
$$\sqrt{N}\nabla_\theta \ell_N(\theta^*) \xrightarrow{D} \mathcal{N}(\vec{0}, \Ff_2(\theta^*)). $$

By Corollary \ref{cor: consistency of estimator}, $\hat{\theta}_N \xrightarrow{P} \theta^*$ as $N \rightarrow \infty$. 
Because $\hat{\theta}_N \leq \bar{\theta} \leq \theta^*$, $\bar{\theta}\xrightarrow{P} \theta^*$. And by the Law of Large Number, we have the convergence $-\nabla^2_\theta \ell_N(\bar{\theta}) \xrightarrow {P} \Ff_1(\bar{\theta})$. If $F_1$ is continuous, $\Ff_1(\bar{\theta})\xrightarrow{P} \Ff_1(\theta^*)$. Thus
$$-\nabla^2_\theta \ell_N(\bar{\theta}) \xrightarrow{P} F_1(\theta^*). $$

Using Slutsky's theorem to combine these two convergences together yields
$$\sqrt{N}(\hat{\theta}_N - \theta^*) \xrightarrow{D} \mathcal{N}(\vec{0}, (\Ff_1(\theta^*))^{-tr}\Ff_2(\theta^* )(\Ff_1(\theta^*))^{-1} ) .$$
This proves the theorem \ref{thm:asymptotic iid 2}. 
\end{proof}

\subsection{Asymptotic theorems for time-series data}
\label{sec: time series}
In this section, we assume fine-scale data
\begin{displaymath}
X_0^N := (X_0, X_1, \ldots, X_N)\COMMA
\end{displaymath}
being time-series data, generated by an unknown Markovian model $P$ with time invariant transition probability $p(x,x')$ and stationary distribution $\mu(x)$. That is 
\begin{displaymath}
P(X_0^N) = \mu(X_0)p(X_0, X_1)p(X_1, X_2) \ldots p(X_{N-1}, X_N).
\end{displaymath}
Here we assume $X_0$ starts in stationary measure. 

% The Coarse Grained map $\Pi$ is defined as
% \begin{displaymath}
% \Pi: \mathcal{M}\rightarrow \mathcal{M}_{CG},
% \end{displaymath}
% where $\mathcal{M}_{CG}=\mathbb{R}^d$, $d\ll D$.

% Note that the Coarse Grained map $\Pi$ must be surjective, that is, $\Pi^{-1}(\mathcal{M}_{CG}) = \mathcal{M}$. 

% Assume the parametric family of the Markovian Coarse Grained models $\mathcal{Q}^\theta$, associated with  transition probability $q^\theta(x,x')$ and stationary distribution $\nu^\theta$.
% Similarly, 
We can get the optimal coarse grained model by minimizing the Relative Entropy Rate
\begin{equation}
\label{eq:RER MC}
\mathcal{H}(P|Q^\theta\circ\Pi) := \EXPECT_{\mu}\big[ \sum_{X'}p(X,X')\log\frac{p(X,X')}{q^\theta\circ\Pi (X,X')} \big].
\end{equation}
In addition, the RER $\mathcal{H}(P|Q^\theta\circ\Pi)$ has an unbiased estimator $\hat{H}_N(p|q^\theta\circ\Pi)$

\begin{equation}
\label{eq:RER N MC}
\hat{H}_N(P|Q^\theta\circ\Pi) := \frac{1}{N}\sum_{i=0}^{N-1} \log \frac{p(X_i, X_{i+1})}{q^\theta\circ\Pi(X_i, X_{i+1})}.
\end{equation}

Thus the minimization of RER \eqref{eq:RER MC} is asymptotically equivalent to the optimization problem
\begin{displaymath}
\min_{\theta} \hat{H}_N(P|Q^\theta\circ\Pi) = \max_{\theta} \frac{1}{N} \sum_{i=0}^{N-1} \log q^\theta\circ\Pi(X_i, X_{i+1}).
\end{displaymath}

Similarly define 

\begin{equation}
\label{eq: Likelihood MC}
\ell_N(\theta) := \frac{1}{N} \sum_{i=0}^{N-1} \log q^\theta\circ\Pi(X_i, X_{i+1}).
\end{equation}

\begin{align*}
\theta^* &:= \arg\min_{\theta} \mathcal{H}(P|Q^\theta\circ\Pi), \\
\hat{\theta}_N &:= \arg\max_{\theta} \frac{1}{N} \sum_{i=0}^{N-1} \log q^\theta\circ\Pi(X_i, X_{i+1}).
\end{align*}

\begin{corollary}
\label{cor: expect of gradient MC}
\begin{displaymath}
\EXPECT_\mu\big[\sum_{X'}p(X,X') \nabla_\theta \log q^\theta\circ\Pi(X, X') |_{\theta=\theta^*} \big] = 0 .
\end{displaymath}
\end{corollary}

\begin{proof}
It could be directly proved by taking the gradient in Eq.~\eqref{eq:RER MC} and use the fact that $\theta^*$ is the argument of the minimum.
\end{proof}

\begin{corollary}[Consistency of the estimator]
\label{cor: consistency of estimator MC}
Suppose that 
$$\sup_{\theta \in \Theta} | \hat{H}_N(P|Q^\theta\circ\Pi) - \mathcal{H}(P|Q^\theta\circ\Pi) | \xrightarrow{P} 0 ,$$
 and that for every $\epsilon>0$,
 $$\sup_{\theta:|\theta - \theta^*|\geq \epsilon} \mathcal{H}(P|Q^\theta\circ\Pi) > \mathcal{H}(P|Q^\theta\circ\Pi)|_{\theta = \theta^*}. $$
 Then we have the consistency of the estimator,
\begin{displaymath}
\hat{\theta}_N \xrightarrow{P} \theta^* \quad\textit{as}\quad N \rightarrow \infty
\end{displaymath}
\end{corollary}

\begin{proof}
Proof is same as the proof of Corollary \ref{cor: consistency of estimator}. 
\end{proof}

\begin{corollary}
\label{cor: 2-step mc}
If $X_i$ is a Markov chain with stationary distribution $\mu(x)$ and transition probability $p(x, x')$. Then
$$S_i = (X_i, X_{i+1}) \text{ is also a Markov chain}$$
with stationary distribution 
$$\mu'((x_1, x_2)) = \mu(x_1)p(x_1, x_2)$$ and transition probability 
$$p'((x_1, x_2), (x_1', x_2')) = \left\{\begin{array}{ll}
p(x_2, x_2')     & \text{if}\quad x_2 = x_1' \\
0     & \text{otherwise}
\end{array} \right. .$$
\end{corollary}

\begin{proof}
$X_i$ follows
$$\int \mu(X_i) p(X_i, X_{i+1}) dX_i = \mu(X_{i+1}),\quad \text{ for all } i \PERIOD$$
Multiply the above equation by $p'((X_{i}, X_{i+1}), (X_{i+1}, X_{i+2}))$ to get 
$$\int \mu(X_i) p(X_i, X_{i+1}) p'((X_{i}, X_{i+1}), (X_{i+1}, X_{i+2})) dX_i = \mu(X_{i+1}) p'((X_{i}, X_{i+1}), (X_{i+1}, X_{i+2})), $$
that is $$\int \mu'((X_i,X_{i+1})) p'((X_{i}, X_{i+1}), (X_{i+1}, X_{i+2})) dX_i = \mu'((X_{i+1},X_{i+2}))\quad \text{ for all } i. $$
Clearly $S_i = (X_i, X_{i+1})$ is a Markov chain with stationary distribution $\mu'$ and transition probability $p'$.
\end{proof}

\label{def: FIM MC}
Two Fisher information matrices are defined as
\begin{subequations}
\begin{align}
\Ff_1(\theta^*) &:= -\EXPECT_{\mu'}\big[ \nabla_\theta^2\log q^\theta\circ\Pi(X, X') |_{\theta=\theta^*}  \big] \label{eq:F1 MC}, \\
\Ff_2(\theta^*) &:= \EXPECT_{\mu'}\big[ (\nabla_\theta \log q^\theta\circ\Pi(X, X'))(\nabla_\theta \log q^\theta\circ\Pi(X, X'))^T |_{\theta=\theta^*} \big] \label{eq:F2 MC}.
\end{align}
\end{subequations}
Note that here $\EXPECT_{\mu'}[(\cdot)] = \int\int \mu'(X,X') (\cdot) dXdX'$

\begin{corollary}
If $ p = q^{\theta^*}\circ\Pi$, then
\begin{displaymath}
\Ff_1(\theta^*) = \Ff_2(\theta^*)
\end{displaymath}
\end{corollary}

\begin{proof}
It is similar to the proof in Corollary \ref{cor: F1equalsF2}.
$$ -\nabla_\theta^2\log q^\theta\circ\Pi(X, X') = (\nabla_\theta\log q^\theta\circ\Pi(X, X'))(\nabla_\theta\log q^\theta\circ\Pi(X, X'))^{tr} - \frac{\nabla_\theta^2 q^\theta\circ\Pi(X, X')}{q^\theta\circ\Pi(X, X')}. $$
Take expectation with respect to $\mu'$ at $\theta=\theta^*$ on both sides, yields
$$\Ff_1(\theta^*) = \Ff_2(\theta^*) - \EXPECT_{\mu'}\left[\frac{\nabla_\theta^2 q^\theta\circ\Pi(X,X')}{q^\theta\circ\Pi(X,X')} |_{\theta=\theta^*}  \right]. $$
If $ p = q^{\theta^*}\circ\Pi $, the last term
\begin{align*}
    \EXPECT_{\mu'}\big[\frac{\nabla_\theta^2 q^\theta\circ\Pi(X,X')}{q^\theta\circ\Pi(X,X')} |_{\theta=\theta^*}  \big] & = \int\int \mu(X)p(X,X')\frac{\nabla_\theta^2 q^{\theta^*}\circ\Pi(X, X')}{q^\theta\circ\Pi(X,X')}  dXdX' \\
    &= \int\int \mu(X) \nabla_\theta^2 q^{\theta^*}\circ\Pi(X, X') dXdX' \\
    &= \nabla_\theta^2 \int\int \mu(X) q^{\theta^*}\circ\Pi(X,X') dXdX' \\
    &= \nabla_\theta^2  1 \\
    & = 0
\end{align*}
\end{proof}

\begin{theorem}
\label{thm:asymptotic mc}
If Markov chain $(X_i)$ is finite or bounded and functional $\nabla_\theta(q^\theta\circ\Pi)$ has finite second moment, then 
\begin{displaymath}
\sqrt{N} (\hat{\theta}_N - \theta^*) \xrightarrow{D} \mathcal{N}(0, I(\theta^*)),
\end{displaymath}
where
\begin{displaymath}
I(\theta^*) = (\Ff_1(\theta^*))^{-T}\mathbf{\Sigma}(\theta^*)\Ff_1(\theta^*))^{-1},
\end{displaymath}
\begin{align*}
\mathbf{\Sigma}(\theta^*) &= E_{\mu'}[(\nabla_\theta\log q^\theta(\Pi X_0, \Pi X_1))(\nabla_\theta\log q^\theta(\Pi X_0, \Pi X_1))^T]|_{\theta=\theta^*}\\
&+ 2 \sum_{i=1}^\infty E_{\mu'}[(\nabla_\theta\log q^\theta(\Pi X_i, \Pi X_{i+1}))(\nabla_\theta\log q^\theta(\Pi X_i, \Pi X_{i+1}))^T]|_{\theta=\theta^*},
\end{align*}
and $\mu'(X,X') := \mu(X)p(X,X')$.
\end{theorem}

\begin{proof}
It is similar to the proof of Theorem \ref{thm:asymptotic iid 2}, but here we need Markov chain Central Limit theorem(\cite{jones2004markov}). We still define a score function 
$$\ell_N(\theta) = \frac{1}{N} \sum_{i=0}^{N-1} \log q^\theta(\Pi X_i, \Pi X_{i+1}). $$
$\hat{\theta}_N = \arg\max_\theta \ell_N(\theta)$ by definition. If we assume the gradient of the score function with respect to $\theta$ exists, then the gradient of $\ell_N(\theta)$ at $\hat{\theta}_N$ must be zero.
$$\nabla_\theta\ell_N(\hat{\theta}_N) = \vec{0} .$$
Then use Taylor expansion to expand $\nabla_\theta \ell_N(\hat{\theta}_N)$ at $\theta^*$.
$$\vec{0} = \nabla_\theta \ell_N(\hat{\theta_N}) = \nabla_\theta \ell_N(\theta^*) + \nabla^2_\theta \ell_N(\bar{\theta})(\hat{\theta}_N - \theta^*), $$
where $\hat{\theta}_N \leq \bar{\theta} \leq \theta^*$. Assume $-\nabla^2_\theta \ell_N(\bar{\theta}) $ is invertible and rearrange the equation to get
$$\sqrt{N}(\hat{\theta}_N - \theta^*) = (-\nabla^2_\theta \ell_N(\bar{\theta}) )^{-1} (\sqrt{N}\nabla_\theta \ell_N(\theta^*) ). $$
Now let $Y_i = \nabla_\theta \log q^\theta(\Pi X_i, \Pi X_{i+1})|_{\theta = \theta^*}$, then $\sqrt{N}\nabla_\theta \ell_N(\theta^*) = \sqrt{N} \bar{Y}$. $Y_i$'s are functionals of Markov chains. The conditions which guarantee the Central Limit Theorem for $Y_i$'s are discussed in Jones's paper \cite{jones2004markov}. In our case, if $X_i$ is finite or bounded, then $X_i$ is uniformly ergodic Markov chain, as well as $(X_i, X_{i+1})$. If $\Ff_2(\theta^*)$ is finite, i.e., the second moment of functional of the Markov chain is finite, then we have the central limit theorem:
$$\sqrt{N}\bar{Y} \xrightarrow{D} \mathcal{N}(\vec{0}, \mathbf{\Sigma}(\theta^*)), $$
where 
\begin{align*}
\mathbf{\Sigma}(\theta^*) &:= Var(Y_0) + \sum_{i=1}^\infty Cov(Y_0,Y_i)\\
 &= \EXPECT_{\mu'}[(\nabla_\theta\log q^\theta(\Pi X_0, \Pi X_1))(\nabla_\theta\log q^\theta(\Pi X_0, \Pi X_1))^{tr}]|_{\theta=\theta^*}\\
&+ 2 \sum_{i=1}^\infty \EXPECT_{\mu'}[(\nabla_\theta\log q^\theta(\Pi X_i, \Pi X_{i+1}))(\nabla_\theta\log q^\theta(\Pi X_i, \Pi X_{i+1}))^{tr}]|_{\theta=\theta^*}.
\end{align*}
Here $\mu'(X_i, X_{i+1}) = \mu(X_i)p(X_i, X_{i+1}).$
The mean is $\vec{0}$ by Corollary \ref{cor: expect of gradient MC}.
By Corollary \ref{cor: consistency of estimator MC}, $\hat{\theta}_N \xrightarrow{P} \theta^*$ as $N \rightarrow \infty$. 
Since $\hat{\theta}_N \leq \bar{\theta} \leq \theta^*$, $\bar{\theta}\xrightarrow{P} \theta^*$. And, by the Law of Large Numbers, we have the convergence $-\nabla^2_\theta \ell_N(\bar{\theta}) \xrightarrow {P} \Ff_1(\bar{\theta})$. If $\Ff_1$ is continuous, 
$$-\nabla^2_\theta \ell_N(\bar{\theta}) \xrightarrow{P} \Ff_1(\theta^*). $$
Using Slutsky's theorem to combine these two convergences together yields
$$\sqrt{N}(\hat{\theta}_N - \theta^*) \xrightarrow{D} \mathcal{N}(\vec{0}, (\Ff_1(\theta^*))^{-T}\Sigma(\theta^* )(\Ff_1(\theta^*))^{-1} ) .$$
This proves the theorem \ref{thm:asymptotic mc}.
\end{proof}

\begin{corollary}
\label{cor:asymptotic mc estimated}
If $\Ff_1(\theta^*)$ is estimated by 
\begin{align*}
\hat{\Ff_1}(\hat{\theta}_N) &= - \frac{1}{N} \sum_{i=0}^{N-1} \nabla_\theta^2 \log q^\theta\circ\Pi(X_i, X_{i+1})|_{\theta=\hat{\theta}_N},
\end{align*}
and $\mathbf{\Sigma}(\theta^*)$ is estimated by batch means(\cite{jones2006fixed}) assuming $N$(N=ab) data are broken into $b$ batch of equal size $a$ that are assumed to be approximately independent.
\begin{displaymath}
\hat{\mathbf{\Sigma}}_{BM} = \frac{b}{a-1}\sum_{j=1}^a (\bar{Y}_j - \bar{Y})(\bar{Y}_j - \bar{Y})^T\,,
\end{displaymath}
where 
\begin{align*}
\bar{Y}_j &= \frac{1}{b} \sum_{i=(j-1)b}^{jb-1} \nabla_\theta \log q^\theta(\Pi X_i, \Pi X_{i+1}) |_{\theta=\hat{\theta}_N}, \\
\bar{Y} &= \frac{1}{N} \sum_{i=0}^{N-1} \nabla_\theta \log q^\theta(\Pi X_i, \Pi X_{i+1}) |_{\theta=\hat{\theta}_N}.
\end{align*}
Then we have 
\begin{displaymath}
\sqrt{N \hat{I}^{-1}(\hat{\theta}_N)}  (\hat{\theta}_N - \theta^*) \xrightarrow{D} \mathcal{N}(0, \mathbb{I}),
\end{displaymath}
where
\begin{displaymath}
\hat{I}(\hat{\theta}_N) = (\hat{\Ff_1}(\hat{\theta}_N))^{-T}\hat{\mathbf{\Sigma}}_{BM}(\hat{\theta}_N)(\hat{\Ff_1}(\hat{\theta}_N))^{-1}.
\end{displaymath}
\end{corollary}

\section{Test-bed 1:  Two-scale diffusion processes} 
\subsection{Invariant and transition probability density functions}
Denote  $  \sigma  =(x,y)^{tr} $ and $\sigma^{\epsilon}_t:= (X^\epsilon_t, Y^\epsilon_t)^T$,
 ${a}( \sigma ) = \begin{pmatrix}
-y \\ -\epsilon^{-1}(y-x)
\end{pmatrix}$,
and $ { b} = \begin{pmatrix}
1& 0 \\ 0 & \epsilon^{-1/2}
\end{pmatrix}$.
Then the two-scale diffusion SDE system is rewritten as   
\begin{equation}\label{eq:two-SDE}
d\sigma^{\epsilon}_t  =  {  a}(\sigma^{\epsilon}_t ) dt +  {  b} d {  \bf W}_t\COMMA
\end{equation}
where $ { \bf W}_t = (W^1_t, W^2_t)$, with   $W^1_t$ and $W^2_t$ are independent  standard Wiener processes.
We consider an approximation $p_h(\sigma,\sigma')$ of the  exact transition probability 
$p(\sigma,\sigma') $ of the process $ \sigma^{\epsilon}_t$,  by applying the Euler-Maruyama discretization scheme to  \eqref{eq:two-SDE}.
 That is, for the time step $\delta t=h$,   the probability to be at state $\sigma'$ after time $h$,  given that the system is at $\sigma$ is
\begin{equation}\label{eq:transition}
 p_h(\sigma,\sigma') = \frac{1}{Z} e^{-\frac{1}{2}[\sigma'-\sigma-{  a}(\sigma)h)^{tr} {  b}^{-2}(\sigma'-\sigma-{  a}(\sigma)h]}\COMMA \quad \sigma = (x,y), \ \sigma' = (x',y')\PERIOD 
\end{equation}

The invariant   probability density function $ \CGmu(x) $  for the process 
$X_t^{CG}$ satisfying equation 
\begin{equation*}
dX_t^{CG} = a(X_t^{CG};\theta) dt + dW_t \COMMA
\end{equation*}
%\eqref{eq: eff. model}, 
is given   by the solution of the corresponding stationary Fokker-Planck equation
\begin{equation*}
    -\frac{d}{dx}\left[ a(x;\theta)\CGmu(x) -\frac12 \frac{d\CGmu(x)}{dx}\right] =0 \PERIOD
\end{equation*}
Under appropriate boundary conditions and the normalization  $ \int \CGmu(x) dx =1$, we obtain  
\begin{equation}\label{eq:CGinvariant}
     \CGmu(x) = \frac{1}{Z^{\theta}} e^{-2 \bar U(x;\theta)}\,,
\end{equation}
where $\bar U(x;\theta) $ is defined by
$- \frac{d}{dx} \bar U(x;\theta):= {a}(x;\theta)  $ and %\begin{equation*}
$    Z^{\theta} = \int e^{-2 \bar U(x;\theta)} dx \,.
$
%\end{equation*}
 
The transition probability of the CG process $X_t^{CG}$ is approximated by $\bar{q}_h^\theta(x,x')$, for a discrete time step $\delta t=h$,  
\begin{equation}\label{eq:CGtransition}
\bar{q}_h^\theta (x,x') \sim e^{-\frac{1}{2}|x' - x -  {a}(x;\theta)h|^2} \PERIOD 
\end{equation}

\subsection{Relative Entropy Rate minimization reduces to Force Matching}
For the two-scale diffusion process, we here show that the RER minimization, i.e.,  by using the transition probability, reduces to the force matching. 
Denote $\Delta(\sigma) = \sigma + {\bf a}(\sigma)h$ and recall that the CG map is the orthogonal projection 
$\Pi\sigma  = x$.  
Then \eqref{eq:transition} is  decomposed as 
$$p_h(\sigma,\sigma')d\sigma' = \frac{1}{\bar{Z}}e^{-\frac{1}{2}|\Pi\sigma'-\Pi\Delta(\sigma)|^2}dx'\times \frac{1}{\tilde{Z}}e^{-\frac{\epsilon}{2}|\Pi^{\perp}\sigma'-\Pi^\perp\Delta(\sigma)|^2}d y'\COMMA$$
where $\Pi^\perp  $ is the orthogonal complement of $\Pi$, i.e. 
 $\sigma = \Pi \sigma + \Pi^\perp \sigma$ for any $\sigma$.  
Meanwhile, we define  the (non-unique) transition probability for the CG model in the original state space  
$$q_h^\theta(\sigma,\sigma') = \bar{q}_h^\theta(\Pi\sigma,\Pi\sigma')v(\sigma'|\Pi\sigma')\COMMA$$
where $\bar{q}_h^\theta (\Pi\sigma,\Pi\sigma')$ is defined in \eqref{eq:CGtransition}, and  $v(\sigma'|\Pi\sigma')$ is a non-unique back-mapping probability density.
%
%$ \sim e^{-\frac{1}{2}|\Pi\sigma' - \Delta^\theta(\Pi\sigma)|^2} $ and $\Delta^\theta(\Pi\sigma) = \Pi\sigma + \tilde{a}(\Pi\sigma;\theta)h$
Then  minimizing the RER  
$\min_\theta \Hh(P|Q^\theta)$
 is equivalent to 
$$
\min_\theta \left\{-\int \int \mu(\sigma)p_h(\sigma,\sigma')\log q_h^\theta(\sigma,\sigma')d\sigma d\sigma'\right\}\COMMA $$
and if  we  ignore the terms which are independent of $\theta$, 
$$\min_\theta \{-\int \int \mu(\sigma)p_h(\sigma,\sigma')\log \bar{q}_h^\theta(\Pi\sigma,\Pi\sigma')d\sigma d\sigma'\} \PERIOD$$
Integrate over $\sigma'_2$ to get 
$$\min_\theta \left\{-\int\int \frac{1}{\bar{Z}}e^{-|\Pi\sigma'-\Pi\Delta(\sigma)|^2}\log \bar{q}_h^\theta(\Pi\sigma',\Pi\sigma')d\sigma_1' \mu(\sigma) d\sigma \right\} \PERIOD$$

It is equivalent to 
$$\min_\theta \int \int e^{-|\Pi\sigma'-\Pi\Delta(\sigma)|^2} |\Pi\sigma'-\Delta^\theta(\Pi\sigma)|^2d\sigma'_1\mu(\sigma)d\sigma \PERIOD$$

If write $\Pi\sigma'-\Delta^\theta(\Pi\sigma) = \Pi\sigma' - \Pi\Delta(\sigma) + \Pi\Delta(\sigma) -\Delta^\theta(\Pi\sigma)$, then integrate over $d\sigma'_1(=d\Pi\sigma')$, we have
$$\min_\theta \int |\Pi\Delta(\sigma)-\Delta^\theta(\Pi\sigma)|^2\mu(\sigma)d\sigma \PERIOD$$

Notice that $\Pi\Delta(\sigma) = X_t^\epsilon - Y_t^\epsilon h$ and $\Delta^\theta(\Pi\sigma) = X_t^\epsilon + {a}(X_t^\epsilon;\theta)h$, yields
$$\min_\theta \int |Y_t^\epsilon +  {a}(X_t^\epsilon;\theta)|^2\mu(\sigma)d\sigma\PERIOD$$

\subsection{RE minimization}
For the case with i.i.d.\  data 
$\{(X_i, Y_i)\}_{i=1}^N$, samples from the microscopic stationary probability density $p_s(x,y)$. 
The CG maping is    $ \Pi(X_i, Y_i) = X_i $ and the corresponding   RE   minimization problem is 
\begin{eqnarray*}
   {\theta}^{iid, re}  &=& \arg\max_\theta \EXPECT_{\mu}[  \log \CGmu ] \\
             &=& \arg\max_\theta \left\{-2 \EXPECT_{\mu}[   \bar U(\cdot;\theta) ] - \log Z^{\theta} \right\}\PERIOD
  \end{eqnarray*}

We apply  the Newton-Raphson optimization algorithm to calculate an estimation of $\hat{\theta}$. 
The k-th iteration of the Newton - Raphson algorithm is 
\begin{eqnarray*}
 \hat{\theta}^{(k+1)}_{N} =  \hat{\theta}^{(k)}_{N} -  \hat H^{-1}(\theta^{(k)}) \hat J(\theta^{\theta^{(k)}}) \COMMA
\end{eqnarray*}
where  $\hat J(\theta)$ and $\hat H(\theta)$ are estimators of the 
 Jacobian and Hessian matrix respectively.
The Jacobian is 
\begin{equation*}
    J(\theta) =  -2 \EXPECT_{\mu}[   \nabla_{\theta}  \bar U(\cdot;\theta) ] + 2 \EXPECT_{\bar{\mu}^{\theta}}[   \bar \nabla_{\theta}U(\cdot;\theta) ]\COMMA
\end{equation*}
an estimator of which is 
\begin{equation*}
    \hat J(\theta) =  -2\frac1N \sum_{i=1}^N \nabla_{\theta} \bar U(X_i;\theta)   + 2    \frac1M \sum_{j=1}^M  \nabla_{\theta}\bar U(\bar X_j;\theta)  \COMMA
\end{equation*}
where $\{X_i = \COP x_i\}_{i=1}^N$ is the CG projection of a sample set  generated from the fine model, and $\{\bar X_j\}_{j=1}^M$ is a sample set generated from the coarse  model , for the given value of $\theta$.
The Hessian matrix has elements $$ H_{ij} = 4\EXPECT_{\CGmu}\left[\PD{\bar U}{\theta_i}  \PD{\bar U}{\theta_j} \right]  - 4 \EXPECT_{\CGmu}\left[\PD{\bar U}{\theta_i} \right] \EXPECT_{\CGmu}\left[\PD{\bar U}{\theta_j} \right]$$
for which an estimator  is 

$$ \hat H_{ij} =  \frac4N \sum_{k=1}^N  \PD{\bar U}{\theta_i}   (X_k;\theta) \PD{\bar U}{\theta_j}(X_l;\theta) - \frac4N \sum_{k=1}^N   \PD{\bar U}{\theta_i}   (X_k;\theta) \frac1N\sum_{l=1}^N \PD{\bar U} {\theta_j}(X_l;\theta)\,,$$
where 
$$\PD{\bar U}{\theta_i}   (x;\theta)=\frac{x^{i+1}}{i+1}\,, \quad i=0,\dots 4 \PERIOD$$

\subsection{Additional numerical results}

 \begin{table}[htbp]
    \centering
    \resizebox{\columnwidth}{!}{\begin{tabular}{|c|c|c|c|c|c|c|c|}
    \hline
        N &$\hat\theta$ & Asymptotic $\hat{\sigma}^2$ &Jackknife  $\hat{\sigma}^2$ & Bootstrap $\hat{\sigma}^2$ & Asymptotic CI & Jackknife CI & Bootstrap CI  \\ \hline \\
     50   & $\left[\begin{matrix}  0.0484 \\  -0.8575\\ -0.0340\\ -0.0337\\ -0.0627\end{matrix} \right] $ & $\begin{bmatrix}0.0205  \\  0.0654   \\ 0.1765 \\   0.0209 \\   0.0269
 \end{bmatrix}$ & $\left[\begin{matrix}  0.0182\\ 0.0968\\ 0.2454\\ 0.0514\\ 0.0500\end{matrix} \right] $& $\left[\begin{matrix}0.0172\\ 0.0982\\ 0.2276\\ 0.1102\\ 0.0721\end{matrix} \right]$ & $\begin{bmatrix} -0.2321 &   0.3289 \\
   -1.3589 &  -0.3562 \\
   -0.8573 &   0.7894 \\
   -0.3172 &   0.2498 \\
   -0.3844 &   0.2589 \end{bmatrix}$ & $\begin{bmatrix} -0.2164 &   0.3132 \\
   -1.4674 &  -0.2477 \\
   -1.0049 &   0.9370 \\
   -0.4779 &   0.4104 \\
   -0.5009 &   0.3754 \end{bmatrix}$ & $\begin{bmatrix} -0.2084 &   0.3052 \\
   -1.4717 &  -0.2433 \\
   -0.9690 &   0.9010 \\
   -0.6845 &   0.6170 \\
   -0.5890 &   0.4635 \end{bmatrix}$  \\ \hline \\
     100 & $\left[\begin{matrix}   -0.0323\\ -1.1140\\ 0.3858\\ 0.0706\\ -0.2068\end{matrix} \right]  $& $\begin{bmatrix} 0.0113 \\   0.0359 \\   0.1000 \\   0.0156  \\  0.0184 \end{bmatrix}$ & $\left[\begin{matrix}  0.0083\\ 0.0469\\  0.1432\\ 0.0209\\ 0.0262\end{matrix} \right]$& $\left[\begin{matrix}0.0073\\ 0.0442\\ 0.1377\\ 0.0288\\ 0.0363\end{matrix} \right]$ & $\begin{bmatrix} -0.2403 &   0.1757 \\
   -1.4853 &  -0.7427 \\
   -0.2341 &   1.0057 \\
   -0.1739 &   0.3152 \\
   -0.4729 &   0.0594\end{bmatrix}$ & $\begin{bmatrix} -0.2105  &   0.1459 \\
   -1.5387 &  -0.6894 \\
   -0.3560 &   1.1276 \\
   -0.2128 &   0.3541 \\
   -0.5243 &   0.1107 \end{bmatrix}$ & $\begin{bmatrix} -0.1995  &  0.1349 \\
   -1.5263 &  -0.7017 \\
   -0.3416 &   1.1132 \\
   -0.2618 &   0.4031 \\
   -0.5804 &   0.1668 \end{bmatrix}$   \\\hline \\
    200  &$ \begin{bmatrix}  -0.1122\\ -0.9702\\ 0.1548\\  -0.0587\\ -0.1148\end{bmatrix}$  & $\begin{bmatrix} 0.0053 \\   0.0167 \\   0.0446  \\  0.0060 \\   0.0074 \end{bmatrix}$ & $\left[\begin{matrix}   0.0053\\ 0.0200\\ 0.0631\\ 0.0086\\ 0.0130\end{matrix} \right] $& $\left[\begin{matrix}0.0054\\ 0.0197\\ 0.0703\\ 0.0085\\ 0.0144\end{matrix} \right]$ & $\begin{bmatrix}-0.2545 &   0.0302 \\
   -1.2233 &  -0.7172 \\
   -0.2592 &   0.5688 \\
   -0.2109 &   0.0935 \\
   -0.2839 &   0.0543  \end{bmatrix}$  & $\begin{bmatrix} -0.2543 &   0.0300 \\
   -1.2476 &  -0.6929 \\
   -0.3375 &   0.6472 \\
   -0.2404 &   0.1230 \\
   -0.3385 &   0.1090 \end{bmatrix}$ & $\begin{bmatrix} -0.2558 &   0.0315 \\
   -1.2452 &  -0.6953 \\
   -0.3648 &   0.6745 \\
   -0.2398 &   0.1224 \\
   -0.3500 &   0.1204 \end{bmatrix}$  \\\hline
     \end{tabular}}
    \caption{ Asymptotic, jackknife and bootstrap estimates of the variance and 95\% CI for the FM with i.i.d.\ data.}%$B = 200$ is used for Bootstrap. }
     %  For $N=500$, $\theta^* = [ 0.0236,\  -1.0240,  \ 0.003,\  -0.0012, \ -0.0338   ] $ }
    \label{tab:jack_boot}
\end{table}

\begin{table}[htbp]
    \centering
    \resizebox{\columnwidth}{!}{\begin{tabular}{|c|c|c|c|c|c|c|c|}
    \hline
     N& &Asymptotic with FM & Jackknife with FM & Bootstrap with FM & asymptotic with RE & Jackknife with RE & Bootstrap with RE\\\hline
      \multirow{2}{*}[-2em]{50} & $\hat{\theta}$ & \multicolumn{3}{c|}{$\left[\begin{matrix}  0.0484  &  -0.8575 & -0.0340 & -0.0337 & -0.0627\end{matrix} \right] $ } &  \multicolumn{3}{c|}{$\begin{bmatrix} -0.0117 &  -0.8720  & -0.1158 &  -0.0321 &   0.0201 \end{bmatrix}$}\\\cline{2-8}
     & CI & $\begin{bmatrix} -0.2321 &   0.3289 \\
   -1.3589 &  -0.3562 \\
   -0.8573 &   0.7894 \\
   -0.3172 &   0.2498 \\
   -0.3844 &   0.2589 \end{bmatrix}$ & $\begin{bmatrix} -0.2164 &   0.3132 \\
   -1.4674 &  -0.2477 \\
   -1.0049 &   0.9370 \\
   -0.4779 &   0.4104 \\
   -0.5009 &   0.3754 \end{bmatrix}$ & $\begin{bmatrix} -0.2084 &   0.3052 \\
   -1.4717 &  -0.2433 \\
   -0.9690 &   0.9010 \\
   -0.6845 &   0.6170 \\
   -0.5890 &   0.4635 \end{bmatrix}$ & $\begin{bmatrix} -0.1465 &   0.1231 \\
   -1.1165 &  -0.6275 \\
   -0.4993 &   0.2677 \\
   -0.1739 &   0.1096 \\
   -0.0984 &   0.1386 \end{bmatrix}$ & $\begin{bmatrix}-0.9145 &   0.8912 \\
   -2.2984 &   0.5543 \\
   -1.3890 &   1.1573 \\
   -0.8599 &   0.7957 \\
   -0.5091 &   0.5492 \end{bmatrix}$ & $\begin{bmatrix} -0.1704 &   0.1470 \\
   -1.0670 &  -0.6771 \\
   -0.3251 &   0.0934 \\
   -0.1710 &   0.1068 \\
   -0.0756 &   0.1157 \end{bmatrix}$ \\\hline
     \multirow{2}{*}[-2em]{200} & $\hat{\theta}$ & \multicolumn{3}{c|}{$\begin{bmatrix}  -0.1122 & -0.9702 & 0.1548 &  -0.0587 & -0.1148\end{bmatrix}$} &  \multicolumn{3}{c|}{$\begin{bmatrix} 0.0659 &  -0.9759 &  -0.0085 &  -0.0533 & -0.0272\end{bmatrix}$}\\\cline{2-8}
     &CI & $\begin{bmatrix} -0.2545 &    0.0302 \\
   -1.2233 &  -0.7172\\
   -0.2592 &   0.5688\\
   -0.2109 &   0.0935\\
   -0.2839  &  0.0543 \end{bmatrix}$ & $\begin{bmatrix}-0.2543 &  0.0300 \\
   -1.2476 &  -0.6929 \\
   -0.3375 &   0.6472 \\
   -0.2404 &   0.1230 \\
   -0.3385 &   0.1090 \end{bmatrix}$ & $\begin{bmatrix} -0.2558 & 0.0315 \\
   -1.2452 &  -0.6953 \\
   -0.3648 &   0.6745 \\
   -0.2398 &   0.1224 \\
   -0.3500 &   0.1204 \end{bmatrix}$ & $\begin{bmatrix} -0.0740 &   0.2057 \\
   -1.2574 &  -0.6945 \\
   -0.5077 &   0.4906 \\
   -0.1924 &   0.0858 \\
   -0.1801 &   0.1258 \end{bmatrix}$ & $\begin{bmatrix} -0.6344 &  0.7661 \\
   -2.9675 &   1.0157 \\
   -2.0745 &   2.0575 \\
   -1.3104 &   1.2037 \\
   -0.7309 &   0.6765 \end{bmatrix}$ & $\begin{bmatrix} -0.0520 &  0.1837 \\
   -1.1523 &  -0.7996 \\
   -0.2072 &   0.1902 \\
   -0.1587 &   0.0521 \\ 
   -0.0828 &   0.0285 \end{bmatrix}$\\\hline
   \multirow{2}{*}[-2em]{500} & $\hat{\theta}$ & \multicolumn{3}{c|}{$\begin{bmatrix}  0.0236& -1.0240& 0.0039& -0.0012& -0.0338 \end{bmatrix}$} &  \multicolumn{3}{c|}{$\begin{bmatrix}   0.0247& -0.9827& 0.0260&  -0.0640 & 0.0001 \end{bmatrix}$}\\\cline{2-8}
     &CI & $\begin{bmatrix} -0.0663  &  0.1135\\
   -1.1790 &  -0.8689\\
   -0.2265  &  0.2342\\
   -0.0947  &  0.0922\\
   -0.1189  &  0.0513 \end{bmatrix}$ & $\begin{bmatrix}  -0.0667 &   0.1139 \\
   -1.1900 &  -0.8579\\
   -0.2477 &   0.2555\\
   -0.1099 &   0.1075\\
   -0.1350 &   0.0674 \end{bmatrix}$ & $\begin{bmatrix} -0.0703 &  0.1175 \\
   -1.1868 &  -0.8611 \\
   -0.2593 &   0.2670 \\
   -0.1080 &   0.1056 \\
   -0.1396 &   0.0721 \end{bmatrix}$ & $\begin{bmatrix}-0.1390 & 0.1151\\
   -1.2287 &  -0.7505 \\
   -0.3005 &   0.3572 \\
   -0.1345 &   0.1509 \\
   -0.0913 &   0.0572 \end{bmatrix}$ & $\begin{bmatrix} -1.4482 &  1.5079 \\
   -4.2963 &   2.3272 \\
   -3.7320 &   3.7412 \\
   -2.0751 &   1.9495 \\
   -1.1753 &   1.2034 \end{bmatrix}$ & $\begin{bmatrix} -0.0542 &  0.1139 \\
   -1.1491 &  -0.8200 \\
   -0.1631 &   0.1723 \\
   -0.1682 &   0.0427 \\
   -0.0465 &   0.0746 \end{bmatrix}$ \\\hline
    \end{tabular}}
    \caption{ Comparison of FM and RE with jackknife and bootstrap, with i.i.d.\ samples.}
    \label{tab:FM_RE_nonparametric}
\end{table}

\begin{table}[htbp]
\begin{minipage}[t]{0.95\linewidth}
    \centering
    \begin{tabular}{|c|c|c|c|}
    \hline
 Estimator & $N_p$ & $N_t$ & $\hat{\theta}$ \\\hline
  FM & 1 & 500 &$\begin{bmatrix}\ \ 0.1207 &  -0.4328 &  -0.5375 &  -3.0699 &  -2.0022 \end{bmatrix}$ \\\hline
  RER    & 1 & 500  &$\begin{bmatrix}-1.5032 &  -6.0944 & \ \  2.1876 & \ 20.4666  & 15.4355\end{bmatrix}$\\\hline
  FM & 1  & 5,000 &$\begin{bmatrix}\ \  0.0050  &  -0.9485 &  \ \ 0.0453 &  -0.0221 &  -0.0277 \end{bmatrix}$  \\\hline
  RER  &   1  & 5,000& $\begin{bmatrix}-0.1530 &  -1.1251  & \ \ 0.1439 &  -0.0260 &  -0.0481 \end{bmatrix}$ \\\hline
   FM & 1 & 50,000 &$\begin{bmatrix} \ \ 0.0026  & -0.9734 & -0.0105 &  -0.0062  & \ \ 0.0038   \end{bmatrix}$  \\\hline
  RER & 1 & 50,000 &$\begin{bmatrix} \ \ 0.0746  & -0.9805 & \ \ 0.0483 &  -0.0313 &  -0.0255 \end{bmatrix}$  \\\hline
    FM & 10 & 500 &$\begin{bmatrix}  \ \  0.0136 &    -0.8960  &  \ \ 0.0839 &  -0.0339 &  -0.0485 \end{bmatrix}$ \\\hline
    PSRE  & 10 & 500 &$\begin{bmatrix}  -0.3615 &   -1.3288 &  \ \ 0.9739 &  -0.0195  & -0.4129 \end{bmatrix}$ \\\hline  
    FM  &100 & 500 &$\begin{bmatrix} -0.0005 &  -0.9728  &  \ \  0.0058 &  -0.0081 &  -0.0005 \end{bmatrix}$ \\\hline
    PSRE &100  & 500 &$\begin{bmatrix}  -0.1359 &  -0.7976  & \ \ 0.0756 &  -0.0803 & \ \ 0.0188 \end{bmatrix}$ \\\hline
    FM & 100 & 5,000 & $\begin{bmatrix}   -0.0003   & -0.9777  & \ \ 0.0002 &  -0.0017  & \ \  0.0003 \end{bmatrix}$ \\\hline
    PSRE & 100 & 5,000 & $\begin{bmatrix}  -0.0222   & -0.9673  & \ \  0.0798 &    -0.0223  & -0.0171\end{bmatrix}$ \\\hline
 %    10 & 5000 & $\begin{bmatrix} -0.0705 &  -0.9043  & \ 0.0996  & -0.0538  &  \ 0.0015\end{bmatrix}$ \\\hline
%    PSRE-2 & 100 & 5000 & $\begin{bmatrix}  -0.0222 &  -0.9673  & \ 0.0798 &  -0.0223 &  -0.0171 \end{bmatrix}$ \\\hline
    \end{tabular}
    \caption{ Point estimates for the correlated time-series data, with the different estimators for the PSRE and RER.}% $N_p$ trajectories times-series data, data size in each trajectory is $N_t$.}
    \label{tab:FM_multiple_trajectories}
%\end{table}
\end{minipage}
\end{table}

% \begin{table}[htbp]
%     \centering
%     \begin{tabular}{|c|c|c|}
%     \hline
%     $N_p$ & $N_t$ & $\hat{\theta}$ \\\hline
%     1 & 500 &$\begin{bmatrix}\ 0.1207 &  -0.4328 &  -0.5375 &  -3.0699   -2.0022 \end{bmatrix}$ \\\hline
%     1  & 5,000 &$\begin{bmatrix}\   0.0050  &  -0.9485 &  \ \ 0.0453 &  -0.0221 &  -0.0277 \end{bmatrix}$  \\\hline
%     1 & 50,000 &$\begin{bmatrix} -0.0022 &  -0.9857 &   \  0.0101 & \   \   0.0023 &   -0.0049 \end{bmatrix}$  \\\hline
%      10 & 500 &$\begin{bmatrix}  \ \  0.0136 &    -0.8960  &  \  0.0839 &  -0.0339 &  -0.0485 \end{bmatrix}$ \\\hline
%      100 & 500 &$\begin{bmatrix} -0.0005 &  -0.9728  &  \  0.0058 &  -0.0081 &  -0.0005 \end{bmatrix}$ \\\hline
%     \end{tabular}
%     \caption{{\color{red} Point estimates by FM estimator with correlated time-series data.}}% $N_p$ trajectories times-series data, data size in each trajectory is $N_t$.}
%     \label{tab:FM_multiple_trajectories}
% \end{table}

%\begin{eqnarray*}
%  \hat{\theta}_N& = &\arg\max_\theta\frac{1}{N}\sum_{i=1}^N \log q_s^\theta(  X_i) = \\
%  &=&  \arg\max_\theta\frac{1}{N}\sum_{i=1}^N [ -2 \bar U(  X_i;\theta) - \log Z^{\theta} ]
%\end{eqnarray*}

\section{Test-bed 2: Effective force-fields and confidence in coarse-graining of linear polymer chains}
\subsection{Model parameters and functional forms of bonded and non-bonded interactions of the atomistic force-field}

\begin{center}
\setlength{\tabcolsep}{15pt}
\renewcommand{\arraystretch}{1.2}
\begin{tabularx}{\textwidth}{ c c c c   }
 \hline 
 \multicolumn{4}{ c }{\textbf{Non-Bonded Interactions}} \\
 \hline
  \multicolumn{4}{ c }{ $ V_{LJ} = 4 \epsilon \left[ \left( \frac{\sigma}{r}\right)^{12} - \left(\frac{\sigma}{r} \right)^6 \right] $ }\\
 \hline
 \textbf{Atom Type} & \textbf{mass (g/mol)} & \textbf{$\sigma$ (nm)} & \textbf{$\epsilon$ (kj/mol)} \\ 
CH3 & 15.0 & 0.375 & 0.8156 \\
CH2 & 14.0 & 0.395 & 0.3827 \\
 \hline
\end{tabularx}
\end{center}

\bigskip
\begin{center}
\setlength{\tabcolsep}{15pt}
\renewcommand{\arraystretch}{1.2}
\begin{tabularx}{\textwidth}{ c c c }
 \hline 
 \multicolumn{3}{ c }{\textbf{Bonded Interactions}} \\
 \hline
  \multicolumn{3}{ c }{ $ V_b(r) = \frac{1}{2} k \left( r - b \right)^2 $ }\\
 \hline
 \textbf{Bond Type} & \textbf{b (nm)} & \textbf{k $(kj/mol/nm^2)$} \\ 
  CH3-CH2   &  0.154  &   83736.0 \\
  CH2-CH2   &  0.154  &   83736.0 \\
 \hline
\end{tabularx}
\end{center}

\bigskip

\begin{center}
\setlength{\tabcolsep}{15pt}
\renewcommand{\arraystretch}{1.2}
\begin{tabularx}{\textwidth}{ c c c }
 \hline 
 \multicolumn{3}{ c }{\textbf{Angular Interactions}} \\
 \hline
  \multicolumn{3}{ c }{ $ V_a(\theta) = \frac{1}{2} k \left( \theta - \theta_0 \right)^2 $ }\\
 \hline
 \textbf{Angle Type} & \textbf{$\theta_0$ (degrees)} & \textbf{k $(kj/mol/rad^2)$}  \\
  CH3-CH2-CH2  & 112.00  & 482.319 \\
  CH2-CH2-CH2  & 112.00  & 482.319 \\
  \hline 
\end{tabularx}
\end{center}

\bigskip
\begin{center}
\setlength{\tabcolsep}{5pt}
\renewcommand{\arraystretch}{1.2}
\begin{tabularx}{\textwidth}{ c c c c c c c c c}
 \hline 
 \multicolumn{9}{ c }{\textbf{Dihedral Interaction}} \\
 \hline
  \multicolumn{9}{ c }{ $ V_d(\phi) = \sum_{n=0}^8 C_n cos(\phi)^n $ \textit{(IUPAC/IUB convention)}}\\
 \hline
 \textbf \textbf{$C_0$ (kj/mol)} & \textbf{$C_1$} & \textbf{$C_2$} & \textbf{$C_3$} & \textbf{$C_4$} & \textbf{$C_5$} & \textbf{$C_6$} & \textbf{$C_7$} & \textbf{$C_8$} \\
  8.33 &  -17.72 & -2.52 & 30.06 & 18.53 & -16.36 & -37.36 & 14.45 & 23.44 \\
  \hline 
\end{tabularx}
\end{center}

\subsection{CG force-field: Bonded interactions}
Figures \ref{fig:bondlength}, \ref{fig:bondangle}, \ref{fig:dihangle} depict the bonded   interaction potentials, bond length, bond angle, and dihedral angle respectively, for the 3:1 coarse grained  polyethylene model. The bonded interaction were estimated with the Iterative Inverse Boltzmann method.  
  \begin{figure}[htbp]
   \centering
  \includegraphics[width=0.7\textwidth,height=0.7\textheight,keepaspectratio]{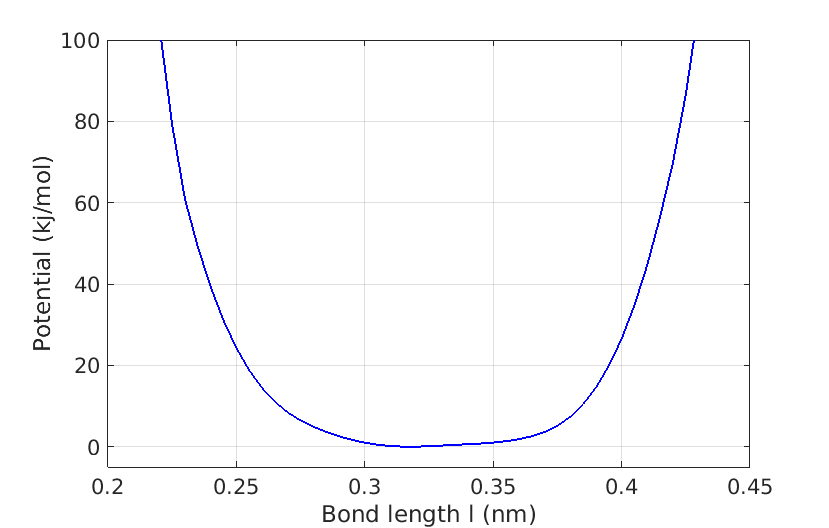}
  \caption{Bond length interaction potential}
\label{fig:bondlength}
  \end{figure}
  
    \begin{figure}[htbp]
   \centering
   \includegraphics[width=0.7\textwidth,height=0.7\textheight,keepaspectratio]{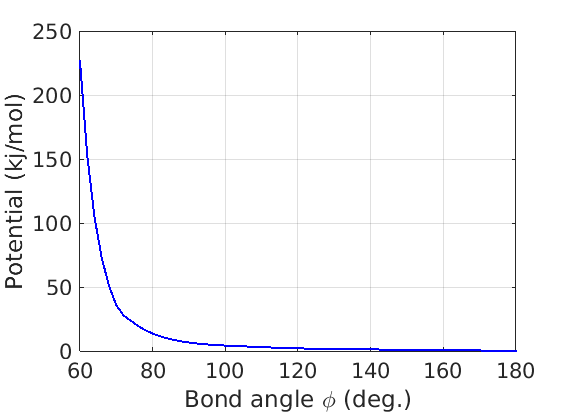}
  \caption{Bond angle interaction potential}
\label{fig:bondangle}
  \end{figure}
  
    \begin{figure}[htbp]
   \centering
  \includegraphics[width=0.7\textwidth,height=0.7\textheight,keepaspectratio]{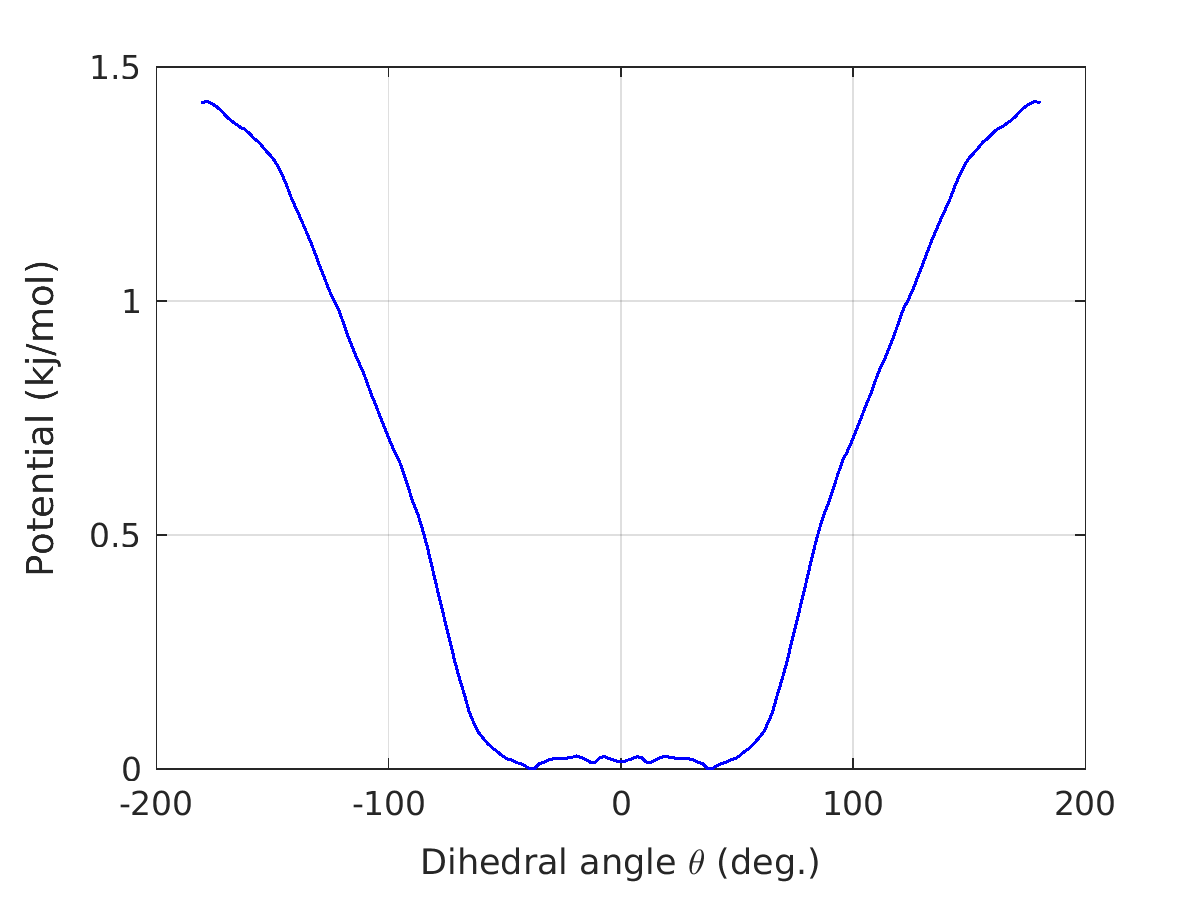}
  \caption{Dihedral angle interaction potential}
\label{fig:dihangle}
  \end{figure}
  
 \begin{figure}[htbp]
   \centering
  \includegraphics[width=0.7\textwidth,height=0.7\textheight,keepaspectratio]{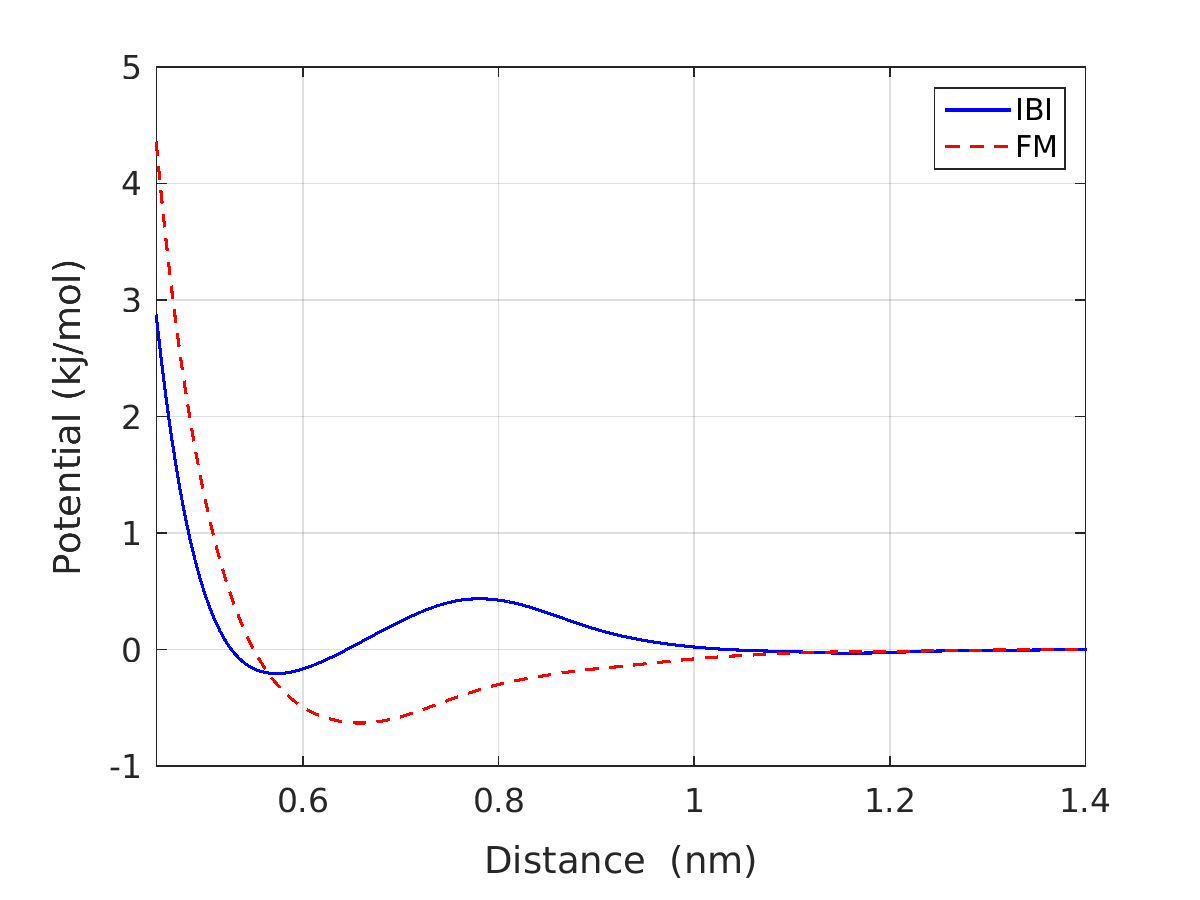}
  \caption{   Comparison of the pair interaction potential $u(r)$ obtained with IBI  and FM (cubic  B-splines).} 
\label{fig:IBIvsFMCubic}
  \end{figure}

\subsection{CG force-field: Non-bonded   interactions}
Figure~\ref{fig:LinearVSCubic} shows a comparison of different expansions for the pair interaction potential, estimated with the FM method. The expansions we compare are  linear B-splines and cubic B-splines, both with 30 parameters and both infered from the same  set of  2000 microscopic configurations. We can observe small differences between the linear and cubic splines.

  \begin{figure}[htbp]
   \centering
  \includegraphics[width=0.7\textwidth,height=0.7\textheight,keepaspectratio]{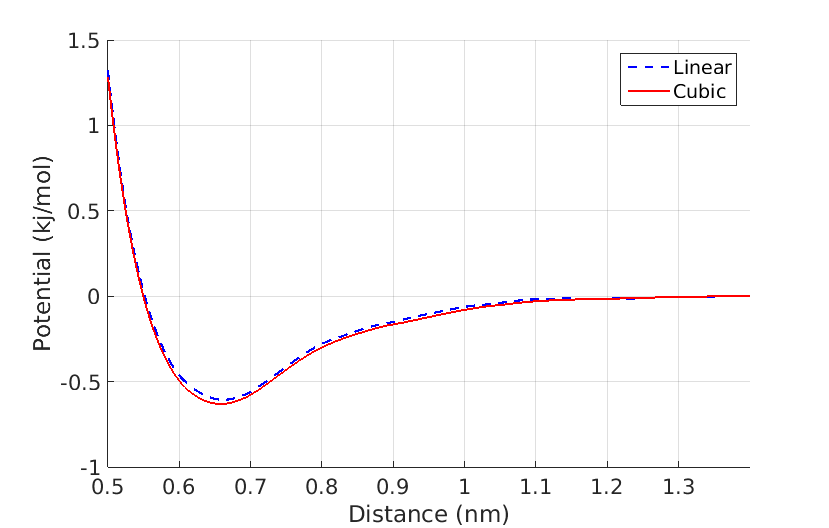}
  \caption{ Pair interaction potential $u(r)$. Linear vs cubic  B-splines.}
\label{fig:LinearVSCubic}
  \end{figure}

 Figures \ref{fig:Parameter_linear_large} and   \ref{fig:potential_linear_boot} depict the results of the FM estimation and the corresponding confidence sets for an expansion with linear B-splines and $75$ parameters. 
    We  estimate the bootstrap mean and confidence intervals for a small $N_B = 300$ set of configurational samples, and $B=300$ bootstrap samples.   The parameters confidence sets are shown in 
 the left figure of \ref{fig:Parameter_linear_large}, depicting higher uncertainty for the first coefficients, 
 The relative standard deviation (RSTD) for the non-zero parameters, shown in the right figure of  \ref{fig:Parameter_linear_large}.
\begin{figure}[!htbp]
 \includegraphics[width=0.5\textwidth,height=0.5\textheight,keepaspectratio]{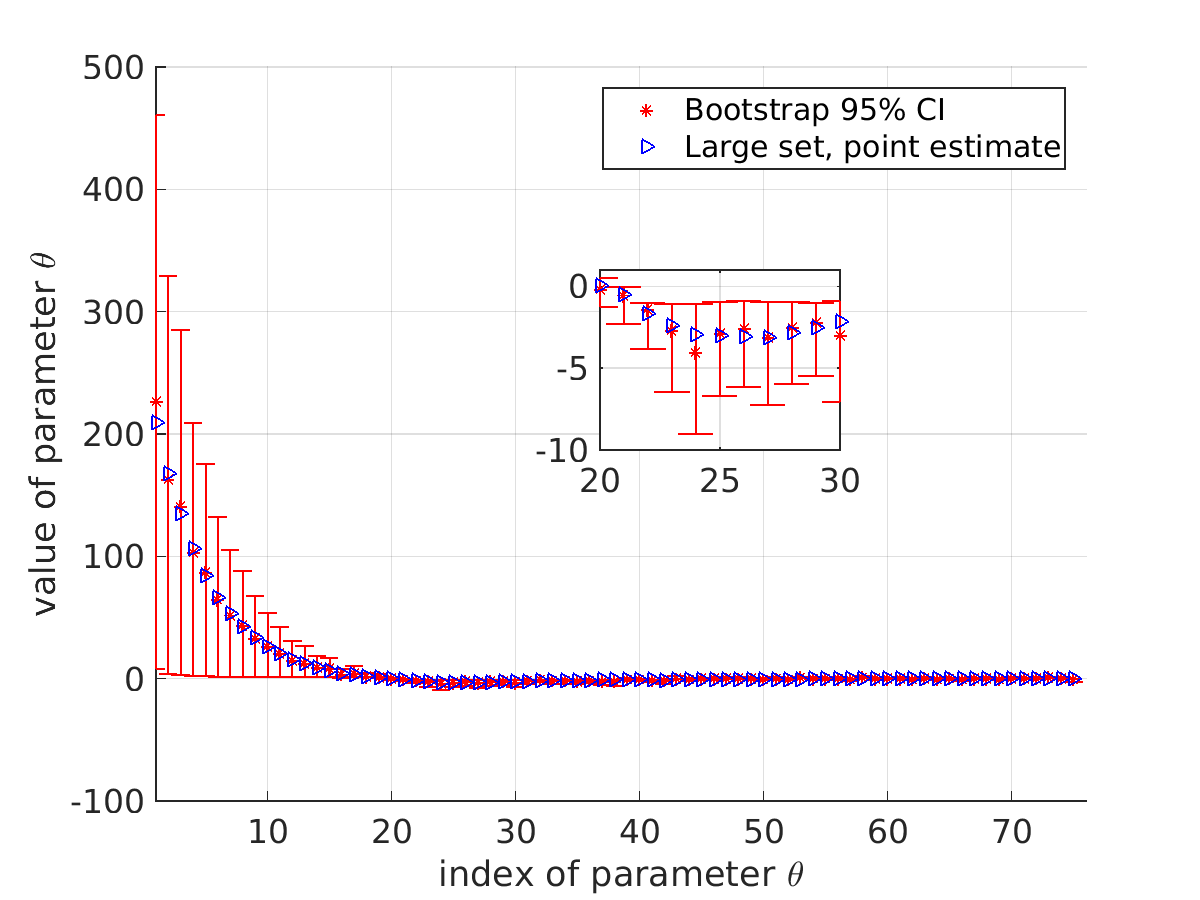}
  \includegraphics[width=0.5\textwidth,height=0.5\textheight,keepaspectratio]{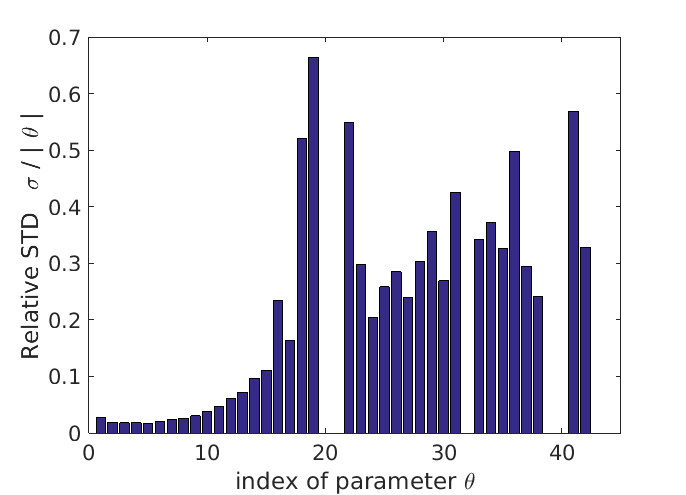}
\caption{Linear splines, 75 parameters. Large sample (5000) vs bootstrap (300) estimates. $ 95\% $ percentile CI}
\label{fig:Parameter_linear_large}
\end{figure}
   
The $80\%$, $95\%$, and $99\%$ bootstrap confidence intervals for the pair interaction potential are presented in figure~\ref{fig:potential_linear_boot}, along with the point estimate of the FM method for a set of $N=5000$ configuration samples.
   \begin{figure}[htbp]
     \includegraphics[width=0.5\textwidth,height=0.5\textheight,keepaspectratio]{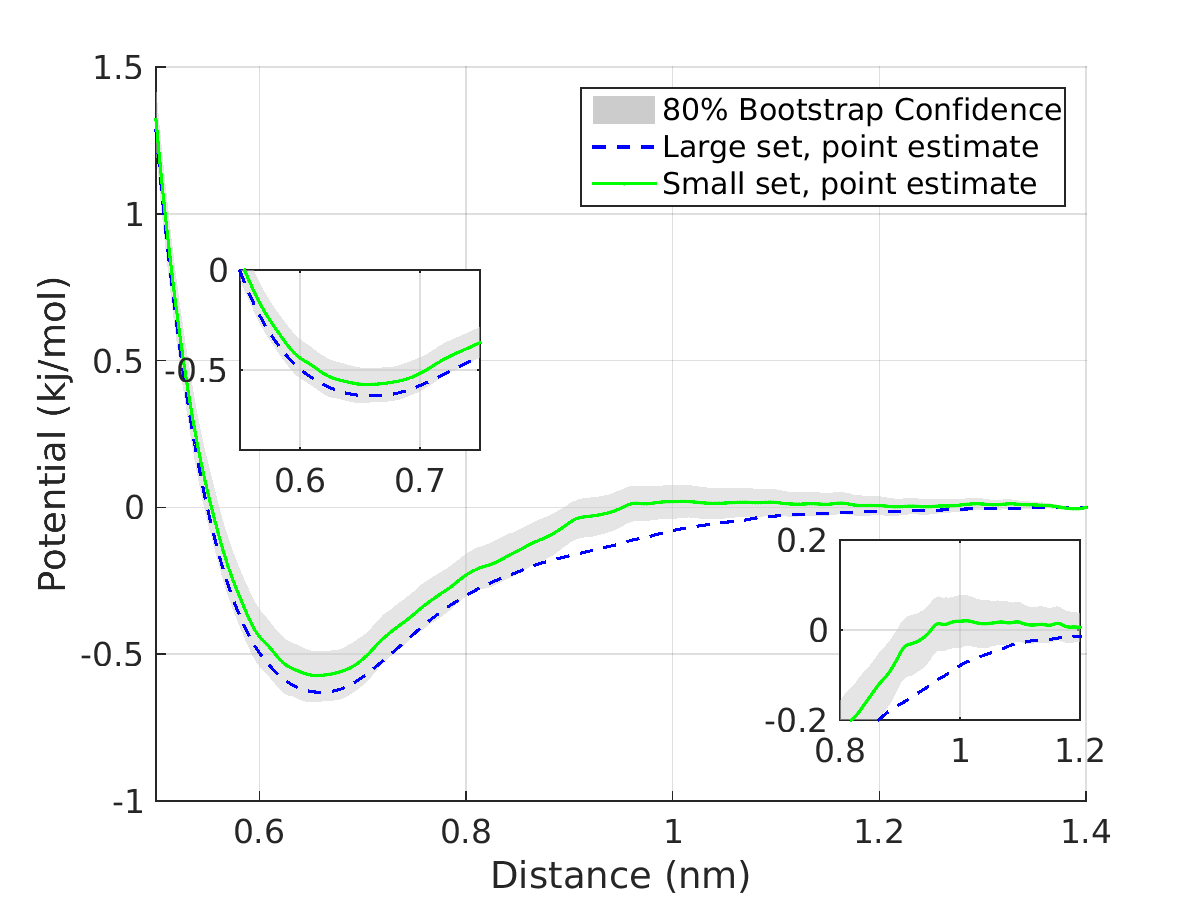}
   \includegraphics[width=0.5\textwidth,height=0.5\textheight,keepaspectratio]{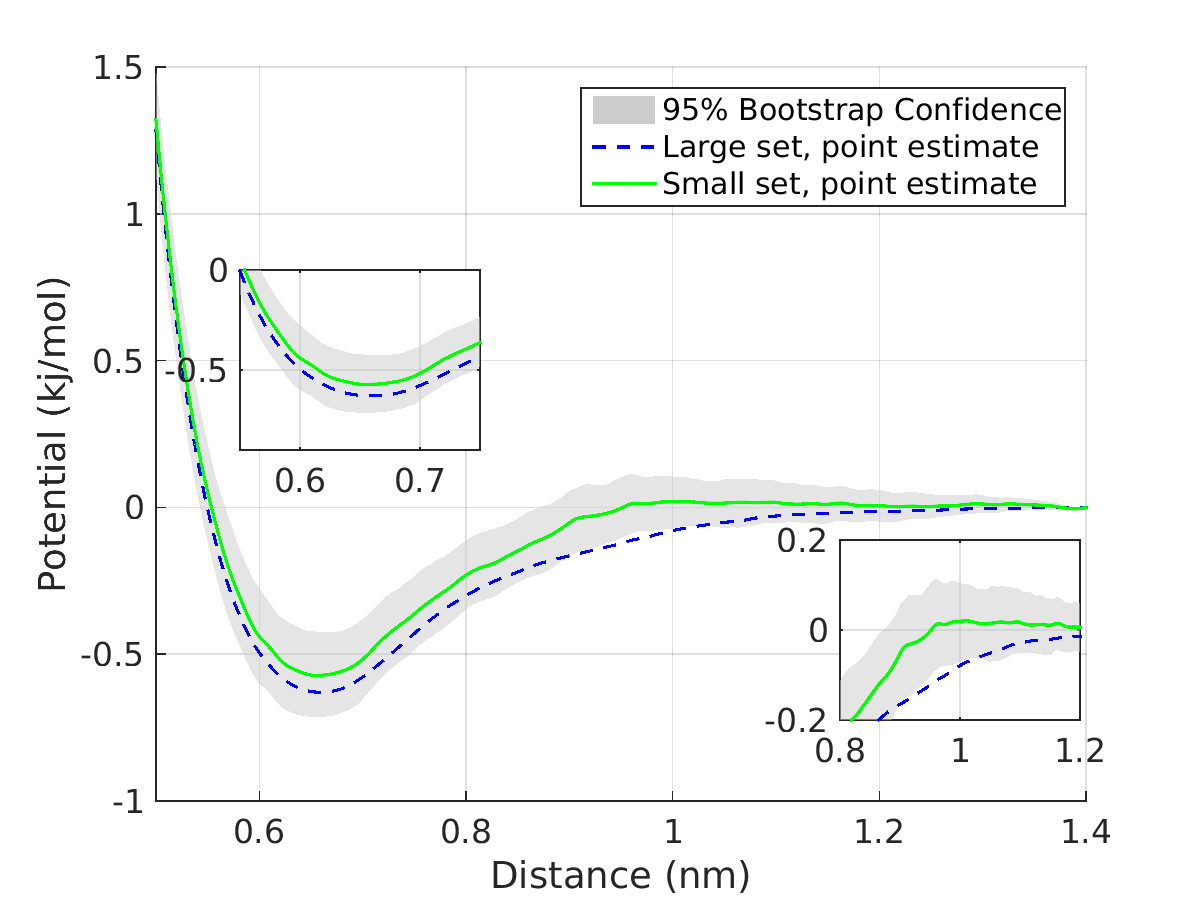}
   \includegraphics[width=0.5\textwidth,height=0.5\textheight,keepaspectratio]{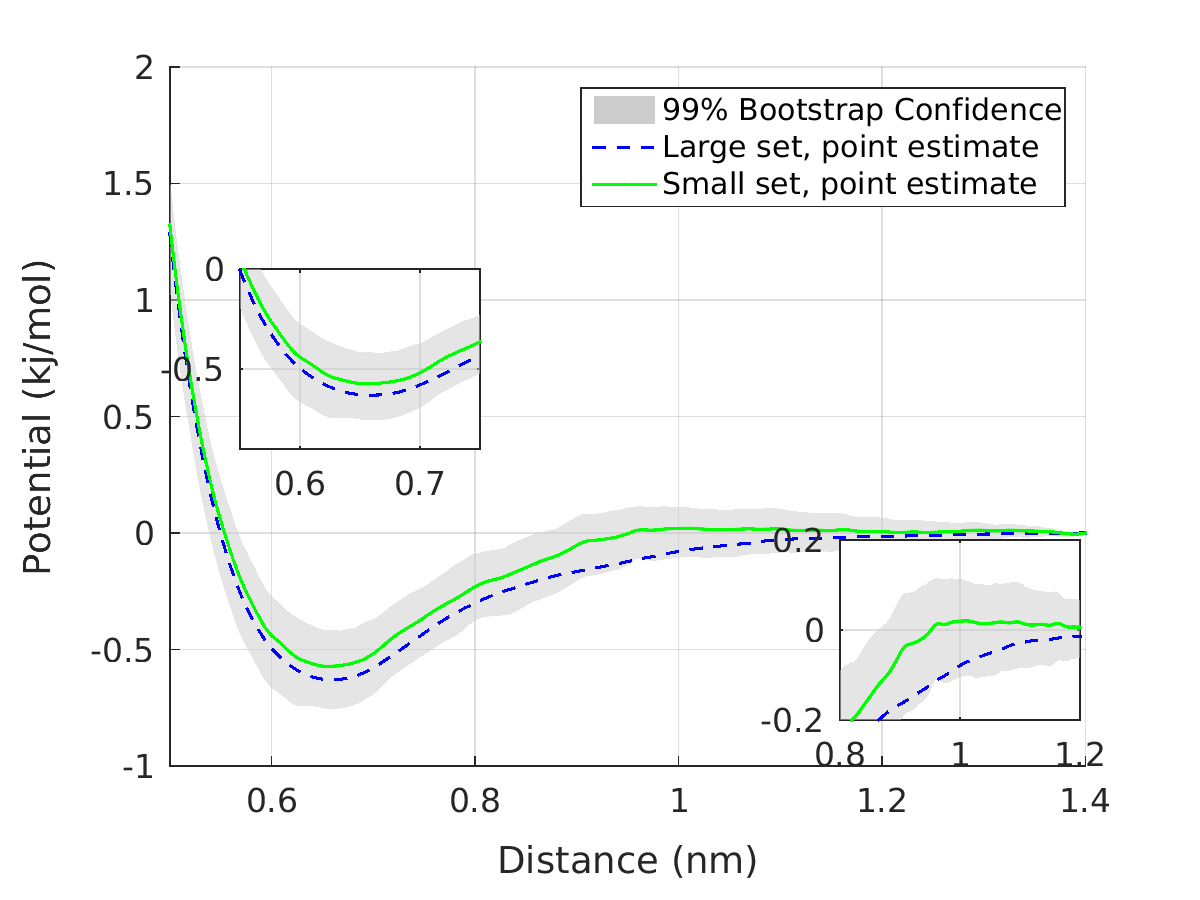}
\caption{  Pair interaction potential $u(r)$   with  $ 80 \%$,  $ 95 \%$, and $ 99 \%$ bootstrap confidence intervals, for  linear splines with a large  (5000) and a small (300) data set. The number of bootstrap samples is 300.}
\label{fig:potential_linear_boot}
  \end{figure}
We present the jackknife mean and  $ 95 \%$ confidence interval in figure \ref{fig:potential_linear_jack}. 
  \begin{figure}[htbp]
   \centering
   \includegraphics[width=0.7\textwidth,height=0.7\textheight,keepaspectratio]{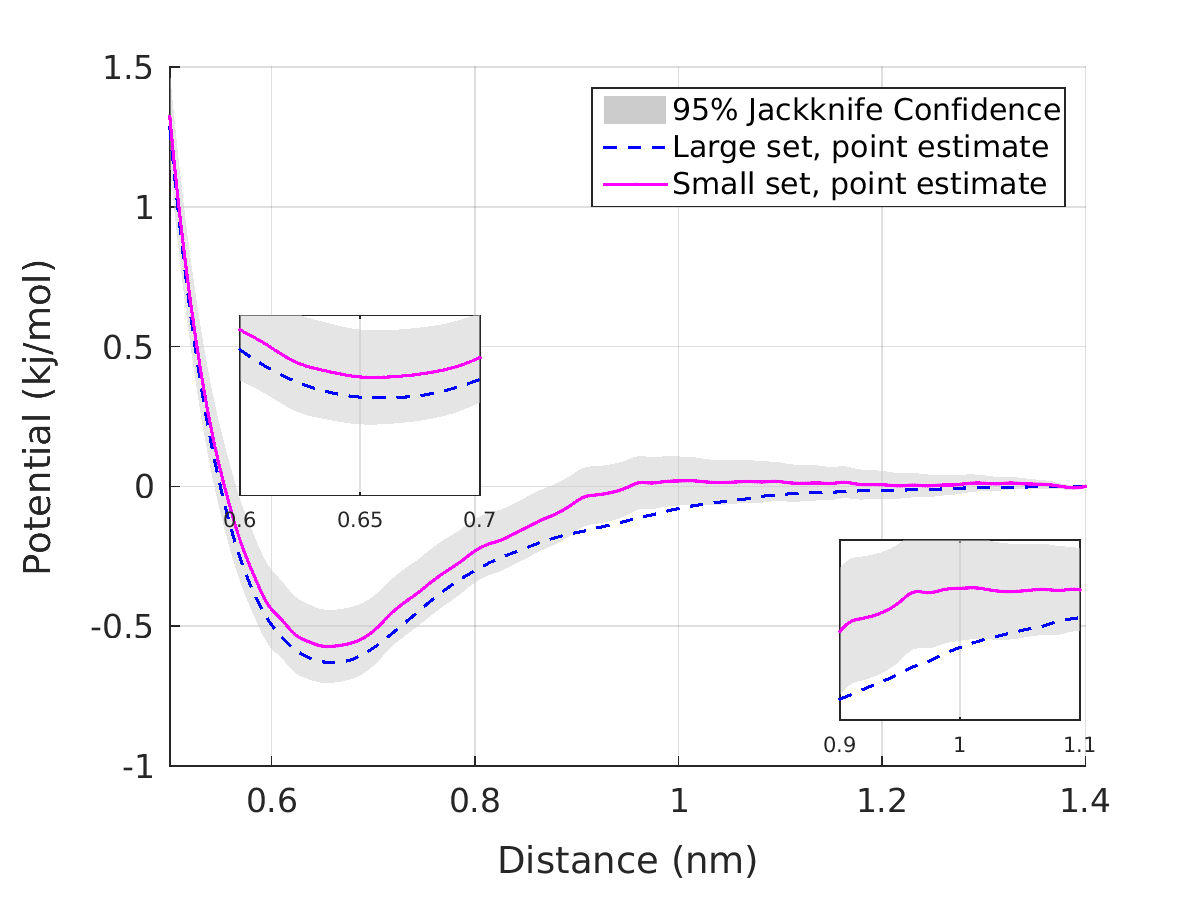}
  \caption{Pair interaction potential $u(r)$   with $ 95 \%$ jackknife confidence interval, for  linear splines with a large  (5000) and a small (300) data set.}
\label{fig:potential_linear_jack}
  \end{figure}

   Figure \ref{fig:meanvspoint}  verifies that for the chosen parametric models the bootstrap and jackknife mean coincide with the point estimate. This is because the model is linear in the parameters and the model is unbiased.
  \begin{figure}[htbp]
  \centering
  \includegraphics[width=0.7\textwidth,height=0.7\textheight,keepaspectratio]{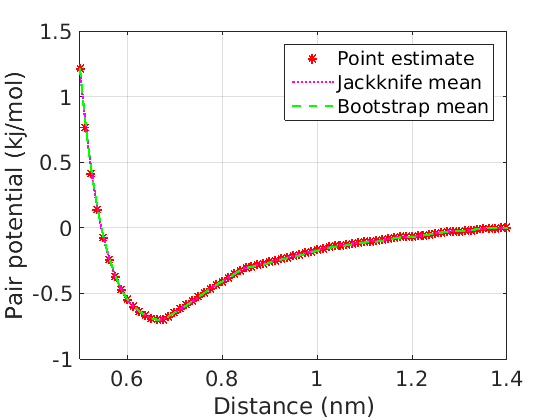}
 \caption{  Pair interaction potential $u(r)$ comparison of the point estimate to the bootstrap and jackknife mean. Cubic splines with 30 parameters. }
\label{fig:meanvspoint}
 \end{figure}

\bibliographystyle{plain}
\bibliography{ref_supp_info}
%\end{document}

%\appendix 
%
%\newpage

\end{document}